\documentclass[pdflatex,sn-mathphys-num]{sn-jnl}

\usepackage{amsmath, amssymb, bbm, xspace}
\usepackage{longtable}
\usepackage{color}
\usepackage{mathrsfs}
\usepackage{comment}
\usepackage{ifthen}
\newboolean{showcomments}
\setboolean{showcomments}{true}
\usepackage{courier}
\usepackage{xcolor}
\usepackage{todonotes}
\usepackage[framemethod=tikz]{mdframed}
\usepackage{lineno}
\newmdenv[leftmargin=\dimexpr-0.4em, innerleftmargin=0.5em,
rightmargin=\dimexpr-0.4em, innerrightmargin=0.5em,
linewidth=2pt,linecolor=red, topline=false, bottomline=false,
innertopmargin=0pt,innerbottommargin=0pt,skipbelow=0pt,skipabove=0pt,%
]{notex}

\newenvironment{note}%
{\vskip\dimexpr\dp\strutbox-\prevdepth\relax\notex\strut\ignorespaces}%
{\xdef\notetpd{\the\prevdepth}\endnotex\vskip-\notetpd\relax}

\let\oldtodo\todo

\makeatletter%
\DeclareDocumentCommand{\todo}{ O{} +g +d<> }{%
		\setlength{\marginparwidth}{1.5cm}%
	\IfNoValueTF{#2}{\relax}{%
		\oldtodo[caption={#2},size=\scriptsize,#1]{\renewcommand{\baselinestretch}{0.8}\selectfont\sffamily#2\par}%
	}%
	\IfNoValueTF{#3}{\relax}{%
		\IfNoValueTF{#2}{
			\begin{note}%
				\begin{internallinenumbers}%
					\indent%
					#3%
				\end{internallinenumbers}%
			\end{note}%
		}{
			\vspace{-0\baselineskip}%
			\begin{note}%
				\begin{internallinenumbers}%
					\indent%
					#3%
				\end{internallinenumbers}%
			\end{note}%
		}%
	}%
}%
\makeatother
\usepackage{soul}

\newcommand{\hlc}[2][yellow]{{%
		\colorlet{foo}{#1}%
		\sethlcolor{foo}\hl{#2}}%
}

\usepackage{amsmath}
\usepackage{bm}
\usepackage{xcolor}

\newcommand{\dotprod}[2]{#1^x #2^x - #1^y #2^y + #1^z #2^z}

\newcommand{\removetodo}[2]{\todo[color=pink]{\textbf{delete:} ``#1'' #2}\hlc[pink]{#1}}
\newcommand{\inserttodo}[1]{\todo[color=green!40]{\textbf{insert:} #1}}

\newcommand{\hltodoy}[2]{\todo[color=yellow!40]{#2}\hl{#1} }

\newcommand{\hltodoc}[3]{\todo[color=#3!40]{#2}\hlc[#3]{#1} }

\newcommand{\hltodo}[2]{\todo[color=orange!40]{#2}\hlc[orange!40]{#1} }
\newcommand{\replacetodo}[2]{\todo[color=pink!40]{\textbf{replace with:}``#2'' }\hl{#1} }

\usepackage{marginnote}
\newcommand{\todol}[1]{{%
		\let\marginpar\marginnote
		\reversemarginpar
		\renewcommand{\baselinestretch}{0.8}%
		\todo{#1}}}

\newcommand{\inserttodol}[1]{{%
		\let\marginpar\marginnote
		\reversemarginpar
		\renewcommand{\baselinestretch}{0.8}%
		\inserttodo{#1}}}

\newcommand{\removetodol}[2]{{%
		\let\marginpar\marginnote
		\reversemarginpar
		\renewcommand{\baselinestretch}{0.8}%
		\removetodo{#1}{#2}}}

\newcommand{\hltodol}[2]{{%
		\let\marginpar\marginnote
		\reversemarginpar
		\renewcommand{\baselinestretch}{0.8}%
		\hltodo{#1}{#2}}}

\newcommand{\replacetodol}[2]{{%
		\let\marginpar\marginnote
		\reversemarginpar
		\renewcommand{\baselinestretch}{0.8}%
		\replacetodo{#1}{#2}}}

\newcommand{\hltodoyl}[2]{{%
		\let\marginpar\marginnote
		\reversemarginpar
		\renewcommand{\baselinestretch}{0.8}%
		\hltodoy{#1}{#2}}}

\newcommand{\hltodocl}[3]{{		\let\marginpar\marginnote
		\reversemarginpar
		\renewcommand{\baselinestretch}{0.8}%
		\hltodoc{#1}{#2}{#3}}}

\newtheorem{theorem}{Theorem}[section]

\newtheorem{lemma}[theorem]{Lemma}

\newtheorem{prop}[theorem]{Proposition}
\newtheorem{corollary}[theorem]{Corollary}

\newtheorem{definition}{Definition}[section]

\def\bkE{{\rm I\kern-.17em E}}
\def\bk1{{\rm 1\kern-.17em l}}
\def\bkD{{\rm I\kern-.17em D}}
\def\bkR{{\rm I\kern-.17em R}}
\def\bkP{{\rm I\kern-.17em P}}

\def\bkZ{{\bf{Z}}}

\def\bkE{{\rm I\kern-.17em E}}
\def\bk1{{\rm 1\kern-.17em l}}
\def\bkD{{\rm I\kern-.17em D}}
\def\bkR{{\rm I\kern-.17em R}}
\def\bkP{{\rm I\kern-.17em P}}

\makeatletter
\newcommand{\pushright}[1]{\ifmeasuring@#1\else\omit\hfill$\displaystyle#1$\fi\ignorespaces}
\newcommand{\pushleft}[1]{\ifmeasuring@#1\else\omit$\displaystyle#1$\hfill\fi\ignorespaces}
\makeatother

\def\bkZ{{\bf{Z}}}
\def\b12{(\beta_1,\beta_2)}

\newenvironment{example}{{\noindent \bf Example}}{\hfill $\square$\hspace{-4.5pt}\vspace{6pt}}
\newcounter{example}
\renewcommand{\theexample}{\thesection.\arabic{example}}

\newcounter{remark}
\renewcommand{\theremark}{\thesection.\arabic{remark}}

\newenvironment{remark}{{\noindent \it Remark: }}{\hfill $\square$}

\def\Bscr{\mathscr{B}}

\def\Ebb{\mathbb{E}}
\newlength{\noteWidth}
\long\def\notes#1{\ifinner
{\tiny #1}
\else
\marginpar{\parbox[t]{\noteWidth}{\raggedright\tiny #1}}
\fi\typeout{#1}}

 \def\notes#1{\typeout{read notes: #1}} 

\newcommand{\ie}{i.e.\@\xspace} 

\newcommand{\e}[2]{{\small e}$\scriptstyle#1$#2}

\def\Ebb{\mathbb{E}}

\def\Pbb{{\mathbb{P}}}

\def\det{\mathop{{\rm det}}}
\def\diag{\mathop{\hbox{\rm diag}}}

\def\half  {{\textstyle{1\over 2}}}

\def\spose#1{\hbox to 0pt{#1\hss}}

\def\text #1{\hbox{\quad#1\quad}}

\def\nthinsp{\mskip -2   mu}

\def\superstar{^{\raise 0.5pt\hbox{$\nthinsp *$}}}
\def\SUPERSTAR{^{\raise 0.5pt\hbox{$*$}}}

\def\lamstarT {\lambda^{\raise 0.5pt\hbox{$\nthinsp *$}T}}

\def\Bscr{{\cal B}}

\def\non{\nonumber}

\let\forallnew\forall
\renewcommand{\forall}{\forallnew\ }
\let\forall\forallnew

		\def\bkE{{\rm I\kern-.17em E}}
		\def\bk1{{\rm 1\kern-.17em l}}
		\def\bkD{{\rm I\kern-.17em D}}
		\def\bkR{{\rm I\kern-.17em R}}
		\def\bkP{{\rm I\kern-.17em P}}
		\def\bkY{{\bf \kern-.17em Y}}
		\def\bkZ{{\bf \kern-.17em Z}}
		\def\bkC{{\bf  \kern-.17em C}}

{\begin{list}{}%
         {\setlength{\leftmargin}{#1}}%
         \item[]%
}
{\end{list}}

		\def\bsp{\begin{split}}
		\def\beq{\begin{eqnarray}}
		\def\bal{\begin{align*}}
		\def\bc{\begin{center}}
		\def\be{\begin{enumerate}}
		\def\bi{\begin{itemize}}
		\def\bs{\begin{small}}
		\def\bS{\begin{slide}}
		\def\ec{\end{center}}
		\def\ee{\end{enumerate}}
		\def\ei{\end{itemize}}
		\def\es{\end{small}}
		\def\eS{\end{slide}}
		\def\eeq{\end{eqnarray}}
		\def\eal{\end{align*}}
		\def\esp{\end{split}}
		\def\qed{ \vrule height7.5pt width7.5pt depth0pt}  

	\def\cp2problem#1#2#3#4{\fbox
		 {\begin{tabular*}{0.9\textwidth}
			{@{}l@{\extracolsep{\fill}}l@{\extracolsep{6pt}}l@{\extracolsep{\fill}}c@{}}
				#1 & & $#4 $
			\end{tabular*}}}

		\def\bkE{{\rm I\kern-.17em E}}
		\def\bk1{{\rm 1\kern-.17em l}}
		\def\bkD{{\rm I\kern-.17em D}}
		\def\bkR{{\rm I\kern-.17em R}}
		\def\bkP{{\rm I\kern-.17em P}}

		\def\bkZ{{\bf{Z}}}

\newcommand {\beeq}[1]{\begin{equation}\label{#1}}
\newcommand {\eeeq}{\end{equation}}
\newcommand {\bea}{\begin{eqnarray}}
\newcommand {\eea}{\end{eqnarray}}

\def\texitem#1{\par\smallskip\noindent\hangindent 25pt
               \hbox to 25pt {\hss #1 ~}\ignorespaces}

\def\bsp{\begin{split}}
		\def\beq{\begin{eqnarray}}
		\def\bal{\begin{align*}}
		\def\bc{\begin{center}}
		\def\be{\begin{enumerate}}
		\def\bi{\begin{itemize}}
		\def\bs{\begin{small}}
		\def\bS{\begin{slide}}
		\def\ec{\end{center}}
		\def\ee{\end{enumerate}}
		\def\ei{\end{itemize}}
		\def\es{\end{small}}
		\def\eS{\end{slide}}
		\def\eeq{\end{eqnarray}}
		\def\eal{\end{align*}}
		\def\esp{\end{split}}
		\def\qed{ \vrule height7.5pt width7.5pt depth0pt}  

\usepackage{amsmath, amssymb, xspace}
\usepackage{epsfig}
\usepackage{longtable}
\usepackage{color}
\usepackage{mathrsfs}
\usepackage{subfig}
\newenvironment{proof}[1][]{{\noindent \emph {Proof} #1: }}{\hfill \qed \vspace{3pt}\\ }

\newcommand{\rev}[1]{#1}
\newcommand{\revii}[1]{#1}
\newtheorem{rmk}[theorem]{Remark}
\usepackage{amsmath,amssymb,amsfonts}
\usetikzlibrary{shapes,arrows,fit,calc,positioning,automata}
\usepackage{algorithmic}
\usepackage{subfig}
\usepackage{graphicx}

\usepackage{textcomp}
\usepackage{xcolor}
\usepackage{lipsum,adjustbox}
\usepackage{geometry}
\usepackage{amsmath}
\usepackage{physics}
\usepackage{tikz}
\usepackage{caption}
\usepackage{comment}

\usetikzlibrary{shapes,arrows}

\def\ei{e^{0}_i}

 \def\ea{e^{0}_0}
\def\eb{e^{0}_1}
\def\evi{\boldsymbol{e_i}}
\def\eva{\boldsymbol{e_0}}
\def\evb{\boldsymbol{e_1}}
\def\fj{f^{0}_j}
\def\fa{f^{0}_0}
\def\fb{f^{0}_1}
\def\fvj{\boldsymbol{f_j}}
\def\fva{\boldsymbol{f_0}}
\def\fvb{\boldsymbol{f_1}}

\def\e00{e^{0}_0}
\def\ev0{\boldsymbol{e_0}}
\def\f00{f^{0}_0}
\def\fv0{\boldsymbol{f_0}}

\def\e01{e^{0}_1}
\def\ev1{\boldsymbol{e_1}}
\def\f01{f^{0}_1}
\def\fv1{\boldsymbol{f_1}}

\def\Jsq{J^\mathsf{S}_{\mathsf{Q}}}
\def\Jsc{J^\mathsf{S}_{\mathsf{C}}}
\def\Jsp{J^\mathsf{S}_{\perp}}
\def\Jdq{J^\mathsf{D}_{\mathsf{Q}}}
\def\Jdc{J^\mathsf{D}_{\mathsf{C}}}
\def\Jdp{J^\mathsf{D}_{\perp}}

\def\Jsqo{J^\mathsf{S*}_{\mathsf{Q}}}
\def\Jsco{J^\mathsf{S*}_{\mathsf{C}}}
\def\Jspo{J^\mathsf{S*}_{\perp}}
\def\Jdqo{J^\mathsf{D*}_{\mathsf{Q}}}
\def\Jdco{J^\mathsf{D*}_{\mathsf{C}}}
\def\Jdpo{J^\mathsf{D*}_{\perp}}

\def\bie{\textbf{\textit{E}}}
\def\bif{\textbf{\textit{F}}}
\def\Asf{\mathsf{A}}
\def\E{\mathbb{E}}

\def\Jsqop{J^\mathsf{S*}_{\mathsf{Q}}(p_1,p_2)}
\def\Jscop{J^\mathsf{S*}_{\mathsf{C}}(p_1,p_2)}
\def\Jspop{J^\mathsf{S*}_{\perp}(p_1,p_2)}

\def\Jdqop{J^\mathsf{D*}_{\mathsf{Q}}(p_s,p_c)}
\def\Jdcop{J^\mathsf{D*}_{\mathsf{C}}(p_s,p_c)}
\def\Jdpop{J^\mathsf{D*}_{\perp}(p_s,p_c)}
\def\Jsqp{J^\mathsf{S}_{\mathsf{Q}}(\textbf{\textit{E}},\textbf{\textit{F}} ; p_1, p_2)}
\def\Jscp{J^\mathsf{S}_{\mathsf{C}}(\gamma_1, \gamma_2 ; p_1, p_2)}
\def\Jspp{J^\mathsf{S}_{\perp} (\boldsymbol{\theta}, \boldsymbol{\phi} ; p_1, p_2)}
\def\Jdqp{J^\mathsf{D}_{\mathsf{Q}}(\textbf{\textit{E}},\textbf{\textit{F}} ; p_s, p_c)}
\def\Jdcp{J^\mathsf{D}_{\mathsf{C}}(\gamma_1, \gamma_2 ; p_s, p_c)}
\def\Jdpp{J^\mathsf{D}_{\perp}(\boldsymbol{\theta}, \boldsymbol{\phi} ; p_s, p_c)}

\def\pat{\tilde{p}_1}
\def\pbt{\tilde{p}_2}
\def\pst{\tilde{p}_s}
\def\pct{\tilde{p}_c}
\def\Be{\mathcal{B}_1}
\def\Bek{\mathcal{B}_k}
\def\prob{I\kern-0.15em P}

\def\Ber{\textrm{Ber}}

\newcommand{\A}{Arjuna}
\newcommand{\B}{Bhima}
\def\Pbb{{\mathrm{Pr}}}
\def\bkP{\Pbb}

\title{\LARGE \bf Nonlocal Teams and Information Structures}

\author[1,4]{Drishti Baruah}
\equalcont{These authors contributed equally to this work.}
\author[1,4]{Sachin Teli}
\equalcont{These authors contributed equally to this work. This work was done while the authors were affiliated to the Indian Institute of Technology Bombay.}
\author*[2,3]{Ankur A. Kulkarni}
\email{kulkarni.ankur@iitb.ac.in}

\affil[1]{Department of Physics, Indian Institute of Technology Bombay, Mumbai, India}

\affil[2]{Centre for Systems and Control, Indian Institute of Technology Bombay, Mumbai, India}
\affil[3]{Centre of Excellence in Quantum Information, Computing, Science and Technology, Indian Institute of Technology Bombay, Mumbai, India}
\affil[4]{Max-Planck-Institut f\"ur Quantenoptik, Hans-Kopfermann-Stra{\ss}e 1, D-85748 Garching, Germany}

\begin{document}
\maketitle


\begin{abstract}

We look at Bell inequalities \cite{bell1964epr} from the lens of information structures in stochastic teams. We consider the usual \revii{Clauser--Horne--Shimony--Holt (CHSH)} game and a dynamic variant of the same to study how various classes of strategies, classical, projective and quantum, behave under team theoretic solution concepts. \rev{We endow players with a fixed quantum state that we take as the Bell state, and allow variation of the measurement strategies.} \revii{In the static game, projective strategies (where each player performs projective measurements) are team optimal for all parameter values and are closed under best response.} \rev{\revii{In the dynamic game these properties fail:} projective strategies may be strictly outperformed by classical ones, best responses to projective strategies need not be projective, and, beyond a sharp threshold on the channel noise, team optimal projective strategies cease to be person-by-person optimal, even though person-by-person optimal projective pairs always exist.} These results shed light on the delicate interplay of information structure in quantum strategies and \revii{the fragility, under a change of information structure, of structural properties that are taken for granted in the static setting}.
\end{abstract}


\section{Introduction}

Bell's theorem \cite{bell1964epr}, asserting that quantum theory's predictions cannot be explained by any local theory, stands as a milestone in the foundations of physics, revealing the inherent conflict between quantum mechanics and the principle of locality. This theorem demonstrated that no local hidden variable theory can reproduce all the predictions of quantum mechanics.

Viewing Bell inequalities as games offers an intuitive understanding, as exemplified by the CHSH game.
 Bell nonlocal games \cite{brunner2014bell} involve multiple players who receive inputs and produce actions to minimize a cost that depends on their actions and a state of nature that is correlated with their inputs, but not fully known to either player.
Strategies utilising quantum resources such as shared entanglement and quantum measurements often outperform classical strategies in these games, providing an advantage via the deployment of such resources.

This paper uncovers new structural phenomena in nonlocal games through the lens of team theory. Bell nonlocal games can also be thought of as (stochastic) \textit{team} problems. In such a problem, players receive partial information of an underlying state, and must jointly design a strategy that minimizes a cost that depends on their collective actions and the unknown state. \revii{We call such a problem a \emph{nonlocal team} when its cost is one arising from a Bell inequality, as is the case for the CHSH cost studied throughout this paper. The qualifier is a statement about the cost function and not about the team: it is the players' strategies and the correlations they generate that may be nonlocal, not the team itself, and the term carries no claim that the team's correlations violate a Bell inequality.} The \textit{information structure} of a team is a description of \textit{who knows what} in the team. An information structure is \textit{static} if no player can affect the underlying state of a system; else it is \textit{dynamic}. In a dynamic \textit{classical} information structure, all players have access to the information of players that have acted prior to them whereas in a non-classical information structure, this assumption does not hold. Problems with non-classical information structures are known to defy commonly expected properties (as first observed by Witsenhausen~\cite{witsenhausen_counterexample_1968} in the context of Linear-Quadratic-Gaussian control problems), and their discrete versions are known to be NP-complete~\cite{papadimitriou85intractable}. On the other hand problems with classical information structure admit tractable formulations and recursive algorithms~\cite{borkar88convex,kulkarni2014optimizer}. In this paper, when we refer to dynamic information structure, we mean its non-classical variant.

Since information structure is known to have such a dramatic effect on stochastic teams, a natural curiosity arises in understanding the effect of the information structure on the phenomenon of the quantum advantage. Bell nonlocal games in their usual formulation can  be viewed as static team problems and exhibit a quantum advantage as evidenced in the CHSH inequality. We ask -- what is the character of this advantage if the information structure is dynamic and non-classical?

\rev{Our paper follows a flurry of recent research on nonlocality in networks~\cite{tavakoli2022BellNetworks,pozasKerstjens2022FullNetwork}, Bell communication~\cite{brask2017BellCommunication,moreno2020MultistageGames}, and finally quantum causality~\cite{wolfe2021QuantumInflation,ligthart2023ConvergentInflation}. It \revii{may thus be seen as a} team-theoretic take, adding to this rapidly evolving line of scholarly inquiry.}

\revii{Those lines of work enlarge the \emph{causal} scenario --- more parties, more sources, or an added communication link --- and ask which correlations remain certifiably nonclassical. Our question is different: the Bell cost and the two parties are held fixed, and what varies is the team's information structure, so the object of study is the optimality of a strategy \emph{class} rather than the boundary of a correlation set.}

\subsection{Main findings}
The usual CHSH game is a static team problem.
To make a comparison sensible between the static and dynamic information structure, we construct a dynamic team problem with the same cost as a Bell nonlocal game, but where the players act in a sequence, and a noise-corrupted version of the first player's action is available to the second player, while  the first player's information is not. We find that the two problems are linked by an explicit identification of parameters, $p_1\leftrightarrow p_s$ and $p_2\leftrightarrow p_c$, under which the team optimal cost over classical strategies is the \textit{same} for the static and dynamic problems (Theorems~\ref{thm:jsco} and \ref{thm:jdco}). \revii{We stress that this is an equality of optimal \emph{values}, not an isomorphism between the two team problems: the strategy sets and the joint laws they induce differ structurally.}

We then study the quantum advantage that the players can obtain.
Our aim is to compare and contrast the structural properties of the optimal strategies with the change of information structure. Our main finding is that projective strategies, which are known to \revii{suffice for optimality} in the usual CHSH game, cease to be always optimal when the information structure is dynamic. In fact, for certain values of the problem parameters, projective strategies yield a strict \textit{disadvantage} over classical strategies when the information structure is dynamic.

%

We then analyse two other solution concepts or properties arising in game theory: person-by-person optimality and best response strategies. A person-by-person optimal strategy is one that players can be expected to play in the absence of any pre-play communication.   Each player picks a strategy that is optimal given the other player's strategy; in other words a Nash equilibrium  \cite{nash50equilibrium,aumann74subjectivity}. \revii{Eisert et al.~\cite{eisert1999quantumgames} study quantum strategies in strategic-form games, where players have conflicting objectives; here the players share one objective, and person-by-person optimality is the natural non-cooperative counterpart of team optimality.}
 The best response refers to the optimal strategy chosen by a player in response to the actions of the other player.
We investigate these properties from the vantage of  strategy classes, and ask if classical or projective strategy classes enjoy the property that contains a person-by-person optimal or best response. Here too we find that the information structure is critical. \rev{For example, person-by-person optimal projective pairs exist for both information structures, but in the dynamic game, beyond a sharp threshold on the channel noise, they part ways with team optimality: the team optimal projective pairs are not person-by-person optimal, and there are person-by-person optimal pairs that are strictly outperformed even by classical strategies.} Comprehensive summary of all results is found in Table \ref{table:summary}.
Our results shed light on the subtle and fascinating role played by the information structure of a stochastic team with reference to quantum strategies and serve as instructive extensions of the landmark CHSH game.

\rev {Due to abstract specification of general quantum strategies, any reduction to a finitely specified class of strategies (projective strategies, in this case) is especially relevant theoretically and practically.
Our \revii{results show} that this reduction is not robust to a change of information
structure. The logical shape of these findings \revii{mirrors} Witsenhausen's. In his
counterexample~\cite{witsenhausen_counterexample_1968} the content is not that nonlinear control is useful in a dynamic information structure; it is that
\emph{linearity}, which is optimal under a classical and static information structure,
ceases to be optimal under a dynamic non-classical one. Analogously we find that projective strategies which enjoy a number of properties in the static information structure, do not necessarily enjoy them in the dynamic one.}

\revii{The correspondence is structural: linear strategies play in Witsenhausen's problem the role that projective strategies play here, a finitely parameterised and tractable subclass that is optimal under a static information structure. In both settings the \emph{dual effect} of the first player's action on the second player's observation destroys that optimality.}

%
%

\subsection{\revii{Assumptions and limitations}}
\rev{In the usual setting, players are allowed any amount of quantum resources (\ie, any quantum state on any state space). We consider a more modest version: we fix the quantum resource to be a Bell state and allow players to vary their choice of measurement strategies. Although this choice limits the universality of our results, it opens the door \revii{to} analysing game-theoretic notions such as best-response and person-by-person optimality which require that there be no pre-play communication. Indeed, if the choice of the shared quantum state is also included a possible strategic choice, without pre-play communication it would not be possible for players to coordinate on such a state. In summary our results are contingent on a given quantum resource.} \revii{More precisely, we regard the entangled state as a resource distributed to the players \emph{before} the parameters of the game are revealed: the measurement strategies may depend on the parameters $(p_1,p_2)$ or $(p_s,p_c)$, but the state may not.} \revii{Finally, we fix attention throughout to the CHSH cost $\Be$. Section~\ref{sec:bell1} shows that the team optimal cost under every strategy class is the same for all eight Bell costs, so this entails no loss of generality.}

\section{Preliminaries}
We first lay out some preliminary concepts required for our study beginning with a formulation of a CHSH game. We will be using $\Pbb$ to denote the joint probability of the random variables involved, and also for the probability of events. \rev{Two notational conventions are used throughout, and are distinct: the complement of a \emph{probability} is written with a tilde, as $\pat=1-p_1$, whereas an overline is reserved for the complement of a \emph{binary variable}, as in $\overline{y_1}=1\oplus y_1$.} In particular, $\Pbb(y_1,y_2)$ denotes the probability that the random variables $y_1,y_2$ take values `$y_1$' and '$y_2$', respectively. Likewise, $\Pbb(y_1=0)$ is the probability that the random variable $y_1$ takes value $0$. The notation $y_1 \perp\!\!\perp y_2$ means $y_1,y_2$ are statistically independent.
\subsection{The CHSH Game}
Clauser, Horne, Shimony and Holt \cite{clauser1969chsh} devised a thought experiment to show how strategies that exploit shared quantum entanglement can violate a particular Bell inequality which later came to be known as the CHSH inequality. The setting that forms the backdrop of this inequality has often been referred to as the \textit{CHSH game}.

The CHSH game comprises of two players \A{} and \B{}.
\A{} receives an observation of a random variable $ y_1 \in \{ 0,1 \} $ and \B{} receives $ y_2 \in \{ 0,1 \} $.
 In the standard CHSH game, $y_1$ and $y_2$ are taken to be uniformly random and independent of each other. \A{} then chooses $ u_1 \in \{ 0,1 \} $ based on the observation $ y_1 $, and \B{} chooses $ u_2 \in \{ 0,1 \} $ by observing $ y_2 $. They win the game if $ u_1 \oplus u_2 = y_1 \cdot y_2 $ and lose otherwise. The highest winning probability that the players can achieve classically is $75\%$. In this context,  \rev{``classically" means that each player independently and locally randomizes its action as a function of its own observation only.} On the other hand, if the players choose actions resulting from a measurement performed on a Bell state,  then the best winning probability increases to about $85\%$. This increase reflects the advantage provided by quantum correlations.

In this paper we view the CHSH game through the lens of  team theory.
\A{} and \B{} form a \textit{team} in the sense of Radner and Marschak~\cite{marschak1955elements,radner1962team} -- decision makers with asymmetric information but a common goal, in this case maximizing the winning probability. \A{} and \B{} are referred to as \textit{players}. \revii{As noted in the introduction, the cost defined below is one arising from a Bell inequality, so the problem studied here is a \emph{nonlocal team} in the sense defined there.}
Maximizing the winning probability is equivalent to the minimization of the expected value of a \textit{cost function} $\Be$, defined as:
\[
\Be (u_1, u_2; y_1, y_2) := u_1 \oplus u_2 \oplus y_1 \cdot y_2.
\]
This function vanishes when the players win the game ($u_1 \oplus u_2 = y_1 \cdot y_2$) and is unity when they lose ($u_1 \oplus u_2 \neq y_1 \cdot y_2$). As a result, the expected value $\mathbb{E}[\Be]$ represents the probability of \textit{losing} the game.

A key part of the theory of teams is \textit{information structure} which describes ``who knows what'' in a team. A team theoretic view of the CHSH game now opens the possibility of studying the impact of information structures on the players' strategies and winning probabilities. The goal of this paper is to pursue this direction to its logical conclusion. We aim to understand the impact of the information structure, and team theoretic solution concepts on the outcomes of the CHSH game.

The above game belongs to a broad class of games known as XOR games \cite{brunner2014bell}. Players win if the XOR of their actions, $u_1 \oplus u_2$, equals $c=c(y_1,y_2)$,  some binary function of $y_1$ and $y_2$. Varying the function $c$ yields eight distinct XOR games, each of which corresponds to a version of the CHSH inequality. In our team theoretic view, this yields eight distinct cost functions, numbered $\Be,\hdots,\Bscr_8$, and thus eight distinct team problems.
We refer to these cost functions as ``Bell costs". The results of our analysis are qualitatively independent of the choice of the specific cost function; we show this in Section \ref{sec:bell1}. Thus we work with $\Be$ for convenience.

\subsection{CHSH-type Games with \revii{V}arying Information Structures} 
This section formalizes nonlocal games under two distinct types of information structures -- \textit{static} and \textit{dynamic} -- that describe how information propagates between players.

\subsubsection{Static Information Structure}
An information structure is said to be \textit{static} if the information available to any player is unaffected by the actions of other players; thus it depends only on the randomness from the environment. In this information structure we assume  \A{} and \B{} receive independent binary inputs $ y_1 \sim \Ber(1-p_1) $ and $ y_2 \sim \Ber(1-p_2) $ and $y_1$ or $y_2$ is not affected by the action of the other player; this is demonstrated in Figure \ref{fig:info}.
\begin{align*}
     \bkP(y_1 = 0) &= p_1 \quad  \bkP(y_1 = 1) = 1-p_1 =: \pat, \\
     \bkP(y_2 = 0) &= p_2 \quad  \bkP(y_2 = 1) = 1- p_2 =: \pbt.
\end{align*}
The goal of the players is to minimize the expected cost $\Be$ over a certain class of strategies; the exact expressions of the cost for these strategy classes is given in Section \ref{sec:strategy}.

This paper focuses on the class of classical, quantum and projective strategies, denoted $\mathsf{C},\mathsf{Q},\mathsf{\perp}$ respectively (formally defined in Section \ref{sec:strategy}). Here we formulate the team problem for the static information structure in a way that is agnostic to the choice of strategy class.
Let $\mathsf{I}=\mathsf{S}$ denote the static information structure, and let
$\Asf$ denote an arbitrary \emph{strategy class} ($\Asf \in\{\mathsf{C},\mathsf{Q},\mathsf{\perp}\}$).
A strategy in the class $\Asf$ induces a conditional law
\[
\Pbb_{\Asf}(u_1,u_2 \mid y_1,y_2),
\qquad u_i,y_i\in\{0,1\},
\]
which specifies the probability with which actions $u_1,u_2$ are taken given that the players observe $y_1,y_2.$
Under the static structure the observations are independent,
\[
y_1 \sim \Ber(1-p_1), \qquad y_2 \sim \Ber(1-p_2), \qquad y_1 \perp\!\!\perp y_2,
\]
and the actions $u_1,u_2$ do not affect the distribution of $y_1,y_2$.

The team-optimal cost under strategy class $\Asf$ is defined as
\begin{equation}
\label{eq:J_S_static}
J^{\mathsf{S}}_{\Asf}(p_1,p_2) \;:=\; \min_{\Pbb_{\Asf}} \;
\E_{y_1,y_2 \sim \Ber(1-p_1)\otimes\Ber(1-p_2)}\Big[
\E_{(u_1,u_2)\sim\Pbb_{\Asf}(\cdot\mid y_1,y_2)}\big[\Be(u_1,u_2;y_1,y_2)\big]
\Big],
\end{equation}
where the outer expectation is over the random inputs $(y_1,y_2)$ and the inner expectation is over the actions induced by the chosen strategy, and the minimization is over all laws $\Pbb_{\Asf}$ allowed by the strategy class $\Asf$.
No assumption (classical, quantum, or projective) is made here on the form of $\Pbb_{\Asf}$ beyond membership in the abstract class $\Asf$.
In particular when $p_1=p_2=\tfrac{1}{2}$, the inputs are uniformly distributed as in the original CHSH game and
then \eqref{eq:J_S_static} recovers the standard CHSH game objective for the particular choice of $\Asf$.





\tikzstyle{block} = [draw, rectangle,
minimum height=3em, minimum width=4em]
\tikzstyle{sum} = [draw, circle, node distance=1cm]
\tikzstyle{input} = [coordinate]
\tikzstyle{output} = [coordinate]
\tikzstyle{pinstyle} = [pin edge={to-,thin,black}]

\begin{figure}[!ht]
	\centering
	\subfloat[Static Information Structure $\mathsf{S}$\label{subfig1}]
	{
		\begin{tikzpicture}[auto, node distance=1cm,>=latex',baseline]
			\node [input, name=input] {};
			\node [input, right of = input, node distance= 2.5cm](gap) {};
			\node [block, right of=gap,node distance=2cm] (b1) {\A{}};
			\node [block, right of=b1, node distance=7cm] (b3) {\B{}};
			\node[input, right of= b3, node distance=2cm](b4){};
			\node [input, name=int, right of=b3, node distance=4cm] {};
			\node [input, below of=b1, node distance=2cm] (u1) {d};
			\node [input, below of=b3, node distance=2cm] (u2) {d};

			\draw [draw,->] (gap) -- node {$y_1$} (b1);
			\draw [->] (b1) -- node {$u_1$} (u1);
			\draw [->] (b3) -- node {$u_2$} (u2);
			\draw [<-] (b3) -- node {$y_2$} (b4);
		\end{tikzpicture}
	}
	\hfill
	\subfloat[Dynamic Information Structure $\mathsf{D}$\label{subfig2}]
	{
		\begin{tikzpicture}[auto, node distance=1cm,>=latex']
			\node [input, name=input] {};
			\node [block, right of=input,node distance=2cm] (b1) {\A{}};
			\tikzset{XOR/.style={draw,circle,append after command={
						[shorten >=\pgflinewidth, shorten <=\pgflinewidth,]
						(\tikzlastnode.north) edge (\tikzlastnode.south)
						(\tikzlastnode.east) edge (\tikzlastnode.west)
					}
				}
			}
			\node [XOR, right of=b1,node distance=5cm, scale = 2.5] (b2) {};
			\node [block, right of=b2, node distance=5cm] (b3) {\B{}};
			\draw [->] (b1) -- node[name=u] {$x_1 = y_1 \oplus u_1$} (b2);
			\draw [->] (b2) -- node[name=b] {$y_2 = x_1 \oplus v$} (b3);
			\node [input, name=int, right of=b3, node distance=4cm] {};
			\node [input, above of=b2, node distance=2cm] (g_d) {d};
			\node [input, below of=b1, node distance=2cm] (u1) {d};
			\node [input, below of=b3, node distance=2cm] (u2) {d};

			\draw [draw,->] (input) -- node {$y_1$} (b1);
			\draw [->] (g_d) -- node {Noise $v$} (b2);
			\draw [->] (b1) -- node {$u_1$} (u1);
			\draw [->] (b3) -- node {$u_2$} (u2);
		\end{tikzpicture}
	}
	\caption{A block diagram explaining the static and dynamic information structure. In the dynamic information structure the information and actions of the two players are interdependent, unlike the static information structure.}
	\label{fig:info}
\end{figure}
\subsubsection{Dynamic Information Structure}
In the dynamic information structure, \B{}'s input depends on \A{}'s action. \A{} observes a source 
whose probability distribution  is given by
\begin{equation*}
    \bkP(y_1 = 0) = p_s \quad \bkP(y_1 = 1) = 1 - p_s =: \pst.
\end{equation*}
\A{} takes action $ u_1 $ resulting in a signal $ x_1 = y_1 \oplus u_1 $. This signal $ x_1 $ is transmitted through a binary symmetric channel that adds an error $ v \sim \Ber(1-p_c)$ resulting in the channel output  $y_2:=x_1 \oplus v$.
The distribution of $v$ is
\begin{equation*}
    \bkP(v = 0) = p_c \quad \bkP(v = 1) = 1 - p_c =: \pct.
\end{equation*}
\A{} and \B{} must choose actions $u_1, u_2$
by observing $y_1,y_2$ respectively with the goal of minimizing the cost $\Be$. The setup is explained in Figure \ref{fig:info}.  The exact expression for the expected cost depends on the strategic class under consideration and are derived in the following section.

For the dynamic information structure, denote $\mathsf{I}=\mathsf{D}$.
As before, let $\Asf$ be an arbitrary admissible strategy class specifying conditional laws $\Pbb_{\Asf}(u_1,u_2\mid y_1,y_2)$ that are implementable under $\mathsf{D},\Asf$.  Notice that $y_2$ now \textit{depends} on $(y_1,u_1,v)$.  The feasible set of conditional laws allowed by $\Asf$ must respect the causal order depicted in Figure~\ref{fig:info} (in particular, $\B$'s action may be allowed to depend only on $y_2$ and whatever side information the class $\Asf$ permits).

The team-optimal cost under $\Asf$ for the dynamic structure is
\begin{equation}
\label{eq:J_S_dynamic}
J_{\Asf}^{\mathsf{D}}(p_s,p_c) \;:=\; \min_{\Pbb_\Asf}
\; \E_{y_1\sim\Ber(1-p_s),\,v\sim\Ber(1-p_c)}
\Big[ \E_{(u_1,u_2)\sim\Pbb_{\Asf}(\cdot\mid y_1,y_2)}\big[ \Be(u_1,u_2;y_1,y_2)\big] \Big],
\end{equation}
where $y_2 = y_1\oplus u_1 \oplus v$ inside the inner distribution, i.e., the inner expectation is taken with respect to the joint law of actions induced by the strategy and the causal channel realization. Again, no structural assumption on $\Pbb_{\Asf}$ is made other than membership in the abstract class $\Asf$.

In the dynamic information structure the action of the first acting player has a secondary, indirect, effect -- the primary effect of the action is on the cost -- by influencing the strategies available to the second acting player. Indeed this secondary effect is present even in the absence of any primary effect (as seen for example, in a communication problem). This phenomenon is known as the \textit{dual effect} in the theory \rev{of} teams and decentralized control, and is known to render decentralized decision problems computationally intractable~\cite{papadimitriou85intractable}. The dual effect was brought to fore by Witsenhausen~\cite{witsenhausen_counterexample_1968} wherein he presented a simple-looking example with a dynamic information \revii{structure} that remains unsolved till today; see also for more discussion~\cite{kulkarni2014optimizer}. \revii{Concretely, $u_1$ enters the cost directly and also enters $y_2=y_1\oplus u_1\oplus v$, so \A{}'s choice both incurs cost and fixes how informative \B{}'s observation is. A choice good for one purpose need not be good for the other. No such tension arises in the static problem, where $y_2$ is fixed by nature.}

Notice that the dynamic problem reduces to a static one when the channel is maximally noisy ($p_c=\tfrac{1}{2}$), which severs the causal dependence of $y_2$ on $u_1$ and recovers independent observations for the two players. \revii{Explicitly, at $p_c=\tfrac12$ the noise $v\sim\Ber(\tfrac12)$ is uniform and independent of $(y_1,u_1)$, so $y_2=y_1\oplus u_1\oplus v$ is uniform and independent of $(y_1,u_1)$ whatever strategy is played. The joint law of $(u_1,u_2,y_1,y_2)$ therefore coincides with that of the static problem at $p_1=p_s$, $p_2=\tfrac12$, for every strategy in every class considered here.} This property has fascinating implications, as we will see later in the paper.

\rev{We make two remarks on the dynamic problem. First, once \A{}'s action reaches \B{} through the channel the two players are no longer space-like separated: \A{} signals to \B{}. The dynamic problem is therefore a sequential, signalling team problem with a non-classical information structure. The correlations attained by \A{} and \B{} are not constrained by any Bell inequality, \revii{and we make no claim of a Bell violation in this setting}, and none of our statements about it is a statement about nonlocality. Where we retain the term ``Bell cost'' for $\Be$, it is only because the cost function is the one arising from a Bell inequality.

Second, the comparison with the static problem is nonetheless controlled. The two problems share the cost $\Be$ and the same observation and action alphabets, and --- by Theorems \ref{thm:jsco} and \ref{thm:jdco} --- the same team optimal classical cost under the identification $p_1 \leftrightarrow p_s$, $p_2 \leftrightarrow p_c$. The classical benchmark is thus held fixed while only the information structure changes, so any difference in the behaviour of the strategy classes is attributable to the information structure and not to a change of problem.}

\subsection{Strategy Classes}
\label{sec:strategy}
This section describes the various strategy classes that players can adopt.
\subsubsection{Classical Strategies}
\label{sec:class}
A general mixed classical strategy is given by stochastic mappings $\mathbb{P}(u_i|y_i)$, $u_i, y_i \in \{0,1\}$. Formally, a mixed classical strategy is a pair of functions $(\Gamma_1, \Gamma_2)$, where
$$ \Gamma_1: \{0,1\} \rightarrow [0,1], \quad \Gamma_2: \{0,1\} \rightarrow [0,1] $$
specify the probability of choosing action $0$ given the observation, i.e., $\mathbb{P}(u_1=0|y_1) = \Gamma_1(y_1)$ and $\mathbb{P}(u_2=0|y_2) = \Gamma_2(y_2)$. Thus under a classical strategy, for the static information structure,
\[ \Pbb(u_1,u_2|y_1,y_2) = \Pbb(u_1|y_1)\Pbb(u_2|y_2). \]
whereas in a dynamic information structure, we only have
    \begin{equation*}
    \bkP(u_1,u_2,y_1,y_2) =\bkP(u_2| y_1, u_1, y_2)  \bkP(y_2 |u_1, y_1)\bkP(u_1 | y_1) \bkP(y_1).
\end{equation*}
In a static information structure $y_2$ is independent of $u_1$, whereas in a dynamic information structure $y_2 = u_1\oplus y_1\oplus v$ depends on $u_1$, which implies $$\Pbb(y_2|u_1,y_1)= \Pbb(u_1\oplus y_1\oplus v=y_2|u_1,y_1).$$\rev{We denote by $\mathsf C$ the class of mixed classical strategies, in which each
player independently randomizes its action as a function of its own observation, and by
$\mathsf C_{\mathrm{det}}\subseteq\mathsf C$ its deterministic elements
$\Gamma_i(y_i)\in\{0,1\}$. As a function of the kernels $\Gamma_1,\Gamma_2\in[0,1]$, the
cost is not jointly linear but is separately affine in each.}



Under $\mathsf C_{\mathrm{det}}$, players choose actions as a deterministic function of their observations.
Formally, classical strategy is a pair of functions $ (\gamma_1, \gamma_2) $,
\begin{align*}
\gamma_1&: \{0,1\} \rightarrow \{0,1\},  \quad \gamma_2: \{0,1\} \rightarrow \{0,1\},
\end{align*}
so that the actions of the players are $  u_1 = \gamma_1(y_1),
u_2 = \gamma_2(y_2).
$ Costs for the static and dynamic information structure will be denoted with a superscript $\mathsf{S},\mathsf{D}$ respectively, with a subscript added to denote the strategy class.

The expected costs under classical strategies for the static and dynamic information structure are denoted $\Jscp,\Jdcp$, and are given by
\begin{subequations}
\begin{align}
\Jscp &= \underset{y_1,y_2}{\Ebb}[\Be(\gamma_1(y_1), \gamma_2(y_2); y_1, y_2)], \text{ where } y_i \sim \Ber(1-p_i) \\
\Jdcp &= \underset{y_1,y_2}{\Ebb}[\Be(\gamma_1(y_1), \gamma_2(y_2); y_1, y_2)]\\
&=\underset{y_1,v}{\Ebb}[\Be(\gamma_1(y_1), \gamma_2(\rev{y_1 \oplus}\gamma_1(y_1) \oplus v); y_1, \rev{y_1 \oplus}\gamma_1(y_1)\oplus v)] \\
& \qquad \text{ where } y_1 \sim \Ber(1 - p_s) \text{ and } v \sim \Ber(1 - p_c).
\end{align}
\label{jic}
\end{subequations}
Here $ \mathop{\Ebb}\limits_{ \bullet} $ denotes the expectation with respect to the variables `$\bullet$'. 
 The optimal costs are defined as
\begin{subequations}
\begin{align}
\Jscop &=  \min_{\gamma_1, \gamma_2} \Jscp \\
\Jdcop &=  \min_{\gamma_1, \gamma_2} \Jdcp.
\end{align}
\label{jico}
\end{subequations}
Technically these optimal costs are called the \textit{team optimal costs}; this and other solution concepts are discussed later in the paper.



\subsubsection{Quantum Strategies}
\label{sec:genq}
To define a quantum strategy, assume that players conduct measurements on a bipartite quantum state $ \ket{\Psi} $ lying in a 4-dimensional Hilbert space with the basis $ \{ \ket{00}, \ket{01}, \ket{10}, \ket{11} \} $.  We take $\ket{\Psi}$ as the Bell state
\[
\ket{\Psi} = \frac{1}{\sqrt{2}} \left( \ket{00} + \ket{11} \right).
\]
In a quantum strategy the players choose a positive operator-valued measure (POVM) as a function of their observations $y_i$. Let $\mathbb{H}_2$ be the set of $2\times 2$ Hermitian matrices. A POVM is a set of Hermitian positive semi-definite matrices which sum up to the identity.  Players then  perform measurements on $\ket{\Psi}$ using these POVMs and choose actions as a function of the outcome of their measurement outcomes. 
Note that while optimizing over the choices of quantum strategies, we hold the state $\ket{\Psi}$ fixed and only vary the POVMs.

Suppose \A{} performs a POVM measurement $\{E_i, I-E_i \}$, where $E_i \in \mathbb{H}_2$, if $y_1 = i$,
and \B{} performs a POVM measurement $\{F_i, I-F_i \}$ with $F_i \in \mathbb{H}_2$ if $y_2 = i$, where $i\in \{0,1\}.$
Let $E_{y_1} - ( I - E_{y_1}) = 2E_{y_1} - I$ be the $\pm 1$ valued observable that \A{} measures. Let the measurement outcome be expressed as $(-1)^{m_1}$. Similarly \B{} measures the observable $2 F_{y_2} - I$ and let the measurement outcome be denoted by $(-1)^{m_2}$. Clearly $m_i \in  \{ 0, 1 \}$ for each $i$.

We assume that the action chosen by the players is $ m_i$, \ie, the players play $u_i=m_i$. In general a strategy can be  $u_i = \gamma(m_i, y_i)$ for  $\gamma: \{0,1\}^2 \to \{0,1\}$. But we show in Appendix \ref{app:povm}, that taking $u_i =m_i$ is without loss of generality. \revii{Since the actions $u_i$ are binary, two-outcome POVMs are likewise without loss of generality: a POVM $\{G_k\}$ with more outcomes, followed by a decision rule $k\mapsto u\in\{0,1\}$, induces the same conditional law on $u$ as the two-outcome POVM $\{\sum_{k:u(k)=0}G_k,\ \sum_{k:u(k)=1}G_k\}$, which is again a valid POVM.}
Since the action is exactly equal to the measurement outcome, we denote these POVMs by $\{E^{u_1}_{y_1}\}_{u_1}$ and $\{F^{u_2}_{y_2}\}_{u_2}$ for each $y_1,y_2$.

A quantum strategy is given by the tuple $(\bie,\bif)$ where
$\bie=(E_0,E_1), \bif=(F_0,F_1)$, and $\mathsf{Q}$ is the set of all such $(\bie,\bif)$. In our notation, $E_i$ stands for $E_i^{0}$, that is the POVM element $E_{y_1=i}^{u_1=0}$. Similarly, $F_j=F^{u_2=0}_{y_2=j}$.  With a slight abuse of notation we will also use $\mathsf{Q}$ to denote any \textit{one} player's strategy set too -- thus $\mathsf{Q}$ is also the set of POVMs $\bie$ (resp., $\bif$).
The expected costs under quantum strategies for the static and dynamic information structure are:
\begin{subequations}
\begin{align}
\Jsqp &= \mathop{\Ebb}_{u_1,u_2,y_1,y_2}[\Be(u_1,u_2; y_1, y_2)] \label{eq:jsqp}\\
\Jdqp &= \mathop{\Ebb}_{u_1,u_2,y_1,y_2}[\Be(u_1,u_2; y_1, y_2)] \label{eq:jdqp}
\end{align}
\label{jiq}
\end{subequations}
As in the  case of classical strategies, we again note that the above expressions differ in the joint distribution  $\Pbb(u_1,u_2,y_1,y_2)$, which is in turn induced by the strategy.

\begin{lemma}
For a quantum strategy $(\bie,\bif)$ we have, for a
	\begin{enumerate}
		\item static information structure,
		\begin{equation}
			\bkP(u_1,u_2,y_1,y_2) =  \Tr(E^{u_1}_{y_1} \otimes F^{u_2}_{y_2} \rho) \bkP(y_1) \bkP(y_2) \label{jointS}
		\end{equation}
		\item  and dynamic information structure,
		\begin{align}
			\bkP(u_1,u_2,y_1,y_2)
            & = \Tr(E^{u_1}_{y_1} \otimes F^{u_2}_{y_2} \rho)\bkP(y_2 |u_1, y_1) \bkP(y_1) \label{jointD}
		\end{align}
	\end{enumerate}
	\label{lemma:joint}
\end{lemma}
\begin{proof}
%
%
\begin{enumerate}
	\item For static information structure, we factorize the joint distribution as

    \begin{equation*}
        \bkP(u_1,u_2,y_1,y_2) = \bkP(u_1,u_2 | y_1, y_2) \bkP(y_1, y_2)
    \end{equation*}
   Moreover, by the specification of the quantum strategy above we have
    $\bkP(u_1,u_2|y_1,y_2) = \Tr(E^{u_1}_{y_1} \otimes F^{u_2}_{y_2} \rho)$. Hence, using the independence of $y_1,y_2$, \begin{equation}
		\bkP(u_1,u_2,y_1,y_2) = \Tr(E^{u_1}_{y_1} \otimes F^{u_2}_{y_2} \rho) \bkP(y_1) \bkP(y_2). \non
	\end{equation}
	\item For dynamic information structure, we write
    \begin{equation*}
    	\bkP(u_1,u_2,y_1,y_2) =\bkP(u_2| y_1, u_1, y_2)  \bkP(y_2 |u_1, y_1)\bkP(u_1 | y_1) \bkP(y_1).
    \end{equation*}
    We
    first analyse the state post \A's measurement. Upon observing $y_1$, \A{} performs a measurement using the POVM $\{E_{y_1}, I-E_{y_1}\}$ and obtains outcome $m_1$, which is also equal to his action $u_1$. \rev{Thus} we have that the probability
 	$\bkP(u_1 | y_1)=\Tr((E^{u_1}_{y_1} \otimes I)\rho)$. The post-measurement state at this stage is given by \begin{equation*}
		\rho'=\frac{\sqrt{E^{u_1}_{y_1}} \otimes I\rho\sqrt{E^{u_1}_{y_1}} \otimes I}{\bkP(u_1 | y_1)}.
	\end{equation*} 
	After this, \B{} performs his measurement on $\rho'$, and chooses action $u_2=m_2$. We thus have,
	$$\bkP(u_2| y_1, u_1, y_2)=\Tr((I \otimes F^{u_2}_{y_2})\rho')=\frac{\Tr((I \otimes F^{u_2}_{y_2}) (E^{u_1}_{y_1} \otimes I)\rho)}{\Tr(E^{u_1}_{y_1} \otimes I \rho)} = \frac{\Tr(E^{u_1}_{y_1} \otimes F^{u_2}_{y_2} \rho)}{\bkP(u_1 | y_1)}.$$
	Hence,
		\begin{equation}
			\bkP(u_1,u_2,y_1,y_2) = \Tr(E^{u_1}_{y_1} \otimes F^{u_2}_{y_2} \rho)\bkP(y_2 |u_1, y_1) \bkP(y_1).\non
		\end{equation}
\end{enumerate}
\end{proof}

	Here, note that $\bkP(y_2 |u_1, y_1)$ is either $p_c$ or $1-p_c$ since $y_2=y_1 \oplus u_1 \oplus v$.
%
The corresponding team optimal costs are defined as
\begin{subequations}
\begin{align}
\Jsqop &= \min_{\bie,\bif} \Jsqp\\
\Jdqop &= \min_{\bie,\bif} \Jdqp,
\end{align}
\label{jiqo}
\end{subequations}
where the minimization is over all possible POVMs $\bie,\bif$.

\subsubsection{Parameterisation of Quantum Strategies}
\label{sec:paramquant}

For a general quantum strategy, we express the POVM elements as a linear combination of the Pauli matrices. Let
\[
E_i = \ei I + \evi \cdot \boldsymbol{\sigma}, \qquad
F_j = \fj I + \fvj\cdot \boldsymbol{\sigma},
\]
where $e^0_i,f^0_j \in \mathbb{R}$, $\evi, \fvj\in \mathbb{R}^3$, and $\boldsymbol{\sigma} = (\sigma_x,\sigma_y,\sigma_z)$ is the Pauli vector. The vectors are defined as $\evi = (e_{i}^x, e_{i}^y, e_{i}^z), $ $ \fvj = (f_{j}^x, f_{j}^y, f_{j}^z) .$
Positivity of the POVM elements \rev{$0 \preceq E_i \preceq I$, $0 \preceq F_j \preceq I$, is equivalent to}
\[
|\evi| \leq \min(e^0_i,1-e^0_i), \qquad
|\fvj| \leq \min(f^0_j,1-f^0_j)\rev{;}
\]
\rev{in particular $\ei,\fj\in[0,1]$ for \emph{every} POVM element, classical or not.}
For the Bell state chosen, define the \emph{signed} inner product $\evi \odot \fvj := \dotprod{e_i}{f_j}$, used for all Bloch-vector products in this paper; the norm $|\cdot|$ remains Euclidean. \rev{Equivalently $\evi\odot\fvj=\evi \cdot M\fvj$ with $M=\mathrm{diag}(1,-1,1)$.

Since $M$ is orthogonal, $|M\fvj|=|\fvj|$ and $\evi\odot\fvj\le|\evi|\,|\fvj|$ with equality if and only if  $\evi$ is parallel to $ M\fvj$.}
With this, the trace simplifies to:
\begin{equation}
	\Tr E_i = 2 \ei, \qquad \Tr F_j = 2 \fj,
	\label{eq:trace_e}
\end{equation}
\begin{equation}
	\operatorname{Tr}(E_i \otimes F_j \rho) = e^0_i f^0_j + \evi \odot \fvj.
	\label{eq:trace_ef}
\end{equation}

\begin{prop}
\label{prop:embedding}
For both the static and dynamic information structure, for each classical strategy, there exists a quantum strategy $(\bie,\bif)$ attaining the same joint distribution on $u_1,u_2,y_1,y_2$.
\end{prop}
\begin{proof}
	\begin{enumerate}
		\item \textit{Static information structure:} 
        First consider the static information structure.
We have the following relations
		\begin{subequations}
			\begin{align}
			\bkP(u_1|y_1) &= \Tr(E^{u_1}_{y_1} \otimes I \rho)=\frac12 \Tr E^{u_1}_{y_1},\label{eq:static_prob1} \\
			\bkP(u_2|y_2) &= \Tr(I \otimes F^{u_2}_{y_2} \rho)=\frac12 \Tr F^{u_2}_{y_2},
			\label{eq:static_prob2}
		\end{align}
		\end{subequations}
where the second equality holds because the Bell state's reduced density matrix is $\half I$.

Now, consider a classical mixed strategy
$(\Gamma_1, \Gamma_2)$ given by $$\mathbb{P}(u_1=0|y_1=i) = \Gamma_1(i) =: \alpha_i, \text{and} \mathbb{P}(u_2=0|y_2=j) = \Gamma_2(j) =: \beta_j,$$ where $\alpha_i, \beta_j \in [0,1]$. 
It can be seen that the quantum strategy (with the Bell state) given by\begin{equation}
			E_0 = \alpha_0 I, E_1 = \alpha_1 I, F_0 = \beta_0 I, F_1 = \beta_1 I, \label{eq:povmclass}
		\end{equation} is equivalent to this classical strategy, as it satisfies the above relations as per \eqref{eq:static_prob1} and \eqref{eq:static_prob2}.
	\item \textit{Dynamic information structure:} For this information structure we have
	\begin{subequations}
		\begin{align}
		\bkP(u_1|y_1) &= \Tr(E^{u_1}_{y_1} \otimes I \rho)=\frac12 \Tr E^{u_1}_{y_1}, \label{eq:dyn_prob1} \\
		\bkP(u_2|y_2, u_1, y_1) &= \frac{\Tr(E^{u_1}_{y_1} \otimes F^{u_2}_{y_2} \rho)}{\frac12 \Tr E^{u_1}_{y_1}},
	\label{eq:dyn_prob2}
	\end{align}
	\end{subequations}
thanks to \eqref{jointD}.
As in the static case, consider a classical mixed strategy, $\bkP(u_1 = 0 | y_1 = i) = \alpha_i, \bkP(u_2 = 0 | y_2 = j) = \beta_j $. The same quantum strategy  $(\bie,\bif)$ given in \eqref{eq:povmclass} satisfies \eqref{eq:dyn_prob1} and  \eqref{eq:dyn_prob2} .
	\end{enumerate}
	Hence, any classical strategy can be expressed in terms of POVMs.
\end{proof}
Due to this proposition, we consider the set of classical strategies (upto isomorphisms) as a subset of the set of quantum strategies (with a Bell state), that is, we have the relation $\mathsf{C} \subseteq \mathsf{Q}$.
\subsubsection{Projective Quantum Strategies}
Projective quantum strategies are a specific subset of quantum strategies that restrict players to executing projective measurements. In these strategies, the measurement operators $\{ E_0, E_1, F_0, F_1 \}$ are characterized by being rank-1 projective matrices, meaning that $E^{2}_i = E_i$ and $F^{2}_j = F_j$, for all $i,j$. These measurements essentially project the quantum state onto specific states and their orthogonal complements. This projection can be parameterized as follows. Define,
\begin{equation*}
    P (\theta,\phi)  = \ket{\theta,\phi}\bra{\theta, \phi}.
\end{equation*}
Here, $\ket{\theta,\phi} = \cos{\theta} \ket{0} + e^{\iota \phi} \sin{\theta} \ket{1}$ represents a pure quantum state on the Bloch sphere~\cite{neilsen2004qcqi}, with $\theta \in [0,\pi/2]$, $\phi \in [0,2\pi]$, and where $\iota$ is the imaginary unit.
A projective quantum strategy is a quantum strategy with the POVM given by
\begin{align}
    E_0 &=: P(\theta_{10}, \phi_{10}), &
    E_1 &=: P(\theta_{11}, \phi_{11}), \notag\\
    F_0 &=: P(\theta_{20}, \phi_{20}), &
    F_1 &=: P(\theta_{21}, \phi_{21}).\label{EandF}
\end{align}
Let $\perp$ be the set of projective quantum strategies. Additionally, let $$\boldsymbol{\theta}=(\theta_{10},\theta_{11},\theta_{20},\theta_{21}), \quad \boldsymbol{\phi} = (\phi_{10}, \phi_{11}, \phi_{20},\phi_{21}).$$ So, a projective strategy is parametrized by  $(\boldsymbol{\theta},\boldsymbol{\phi})$. The expected cost under projective quantum strategies is denoted by
\begin{subequations}
\begin{align}
\Jspp &= \mathop{\Ebb}_{u_1,u_2,y_1,y_2}[\Be(u_1,u_2;y_1,y_2)] \label{eq:jspp}\\
\Jdpp &= \mathop{\Ebb}_{u_1,u_2,y_1,y_2}[\Be(u_1,u_2;y_1,y_2)] \label{eq:jdpp}
\end{align}
\label{jip}
\end{subequations}
where the joint distribution of $(u_1,u_2,y_1,y_2)$ is given by \eqref{jointS} in \eqref{eq:jspp} and \eqref{jointD} in \eqref{eq:jdpp} for the static and dynamic information structure respectively.
The optimal costs are given by:
\begin{subequations}
    \begin{align}
        \Jspop &= \min_{\boldsymbol{\theta},\boldsymbol{\phi}} \Jspp, \label{statproj}\\
        \Jdpop &= \min_{\boldsymbol{\theta},\boldsymbol{\phi}} \Jdpp. \label{dynproj}
    \end{align}
    \label{jipo}
\end{subequations}
Clearly, $\mathsf{\perp} \subseteq \mathsf{Q}$. However, unlike general quantum strategies, all classical strategies do not have equivalent projective quantum strategies as defined above.

Let $(\boldsymbol{\theta^{S*}},\boldsymbol{\phi^{S*}})$ be the $(\boldsymbol{\theta},\boldsymbol{\phi})$ that achieve the minimum in \eqref{statproj} and $(\boldsymbol{\theta^{D*}},\boldsymbol{\phi^{D*}})$ be the corresponding ones for \eqref{dynproj}. The corresponding POVMS are parameterised as given in \eqref{EandF}, so we denote the optimal pairs as $(\bm E^{S*}_\perp,\bm F^{S*}_\perp)$ and $(\bm E^{D*}_\perp,\bm F^{D*}_\perp)$ respectively. We refer to these
$(\bm E^{S*}_\perp,\bm F^{S*}_\perp)$ and $(\bm E^{D*}_\perp,\bm F^{D*}_\perp)$ as the \textit{optimal strategies over projective strategies}.

\subsubsection{\revii{Notation}}
\revii{Table~\ref{tab:notation} collects the symbols used throughout the paper.}
\begin{table}[ht]
\centering
\begin{tabular}{|l|l|}
\hline
\textbf{Symbol} & \textbf{Meaning}\\
\hline
$\Pbb$ & Probability of an event; also a joint or conditional law\\
$y_1,y_2$ & Observations of \A{} and \B{}\\
$u_1,u_2$ & Actions of \A{} and \B{}\\
$p_1,p_2$ & $\Pbb(y_1=0)$, $\Pbb(y_2=0)$ in the static information structure\\
$p_s$ & $\Pbb(y_1=0)$, the source, in the dynamic information structure\\
$p_c$ & $\Pbb(v=0)$, the channel, in the dynamic information structure\\
$\tilde p$ & Complement of a probability, $\tilde p := 1-p$\\
$\overline{y}$ & Complement of a binary variable, $\overline{y} := 1\oplus y$\\
$\mathsf{S},\mathsf{D}$ & Static and dynamic information structure\\
$\mathsf{C},\mathsf{Q},\perp$ & Classical, quantum and projective strategy classes\\
$\gamma_i,\Gamma_i$ & Deterministic and mixed classical strategy of player $i$\\
$E_i,F_j$ & POVM elements; $E_i := E^{0}_{i}$ and $F_j := F^{0}_{j}$\\
$A_i,B_j$ & The $\pm 1$ valued observables $2E_i-I$ and $2F_j-I$\\
$P(\theta,\phi)$ & Rank-one projection onto $\ket{\theta,\phi}$\\
$\Be,\Bek$ & The CHSH cost and a general Bell cost\\
$J^{\mathsf{I}}_{\mathsf{A}}$ & Expected cost under information structure $\mathsf{I}$, strategy class $\mathsf{A}$\\
$J^{\mathsf{I}*}_{\mathsf{A}}$ & Team optimal cost, \ie{} $J^{\mathsf{I}}_{\mathsf{A}}$ minimised over $\mathsf{A}$\\
$\evi \odot \fvj$ & Signed Bloch inner product $\evi \cdot M \fvj$, $M=\mathrm{diag}(1,-1,1)$\\
\hline
\end{tabular}
\caption{\revii{Notation used throughout the paper.}}
\label{tab:notation}
\end{table}

\subsection{Invariance of Bell Costs}
\label{sec:bell1}
Let a general Bell cost be denoted by \begin{equation}
\Bek(u_1, u_2 ; y_1,y_2) = u_1 \oplus u_2 \oplus \alpha_k(y_1,y_2),
\end{equation}
where $\alpha_k$ is a binary function of $y_1,y_2$.
There are 16 such binary functions $\alpha_k$, out of which we have listed the 8 non-trivial ones that give rise to a violation of a corresponding Bell inequality. The remaining 8 functions are considered trivial because they can be perfectly decoupled into local functions of the form $\alpha_k(y_1, y_2) = g(y_1) \oplus h(y_2)$. For these trivial cases, players can always achieve a zero expected cost (a $100\%$ winning probability) using strictly classical deterministic strategies by simply playing $u_1 = g(y_1)$ and $u_2 = h(y_2)$, meaning no quantum advantage can exist.
\begin{table}[ht]
\begin{center}
$\begin{array}{ |c|c|c|c|c|c| }
 \hline
 \Bek & \alpha_k(y_1,y_2) & \alpha_k(0,0) & \alpha_k(0,1)& \alpha_k(1,0)& \alpha_k(1,1) \\
\hline
\mathcal{B}_1 & y_1\cdot y_2 & 0 & 0 & 0 & 1 \\
\mathcal{B}_2 & y_1\cdot \overline{y_2} & 0 & 0 & 1 & 0\\
\mathcal{B}_3 & \overline{y_1}\cdot y_2 & 0 & 1 & 0 & 0\\
\mathcal{B}_4 & \overline{y_1}\cdot \overline{y_2} & 1 & 0 & 0 & 0\\
\mathcal{B}_5 & \overline{y_1} + \overline{y_2} & 1 & 1 & 1 & 0\\
\mathcal{B}_6 & \overline{y_1} + y_2 & 1 & 1 & 0 & 1\\
\mathcal{B}_7 & y_1 + \overline{y_2} & 1 & 0 & 1 & 1\\
\mathcal{B}_8 & y_1 + y_2 & 0 & 1 & 1 & 1\\
 \hline
\end{array}$
\end{center}
    \caption{Cost Functions Associated With All Bell's Inequalities. \rev{Here the overline denotes the complement of a binary variable and $+$ denotes logical OR. Note that $\Bscr_{k+4}=1\oplus\Bek$ for $k\le 4$, so the lower half of the table is the pointwise complement of the upper half.}}
    \label{tab:bellcosts}
\end{table}

In this paper we consider the well known CHSH game which is derived from the CHSH inequality. The CHSH game corresponds to the cost $\Be(u_1, u_2 ; y_1,y_2) = u_1 \oplus u_2 \oplus y_1 \cdot y_2$. But in general, one can look at general ``Bell costs" $\Bek$  which can be derived from other similar Bell inequalities. \rev{We show in Appendix \ref{sec:bellinv} that for each $\Bek$ there is an explicit, class-preserving transformation of strategies --- local relabelings of each player's inputs and outputs, respecting the causal order of the dynamic information structure --- together with an involution of the problem parameters, mapping the cost under $\Bek$ to the cost under $\Be$; since our closed-form optima are invariant under these parameter involutions, the team optimal cost under every strategy class is the same for all eight Bell costs at fixed parameters.} Moreover, the optimal cost under all strategy classes remains the same regardless of which Bell cost one chooses. As such, fixing $\Be$ leads to no loss of generalisation due to this invariance across Bell costs.

\section{Team Optimal Cost under \revii{V}arying Information Structures}

We define our first solution concept by which we will study these strategy classes $\mathsf{C}$, $\mathsf{\perp}$ and $\mathsf{Q}$. Recall that we have already established that $ \mathsf{C}, \mathsf{\perp} \subseteq \mathsf{Q}$.


\begin{definition}
    Fix $(p_1,p_2)$ and $(p_s,p_c)$ and let $\mathsf{A}$ be a strategy class. A strategy pair $ (\bie^*,\bif^*)\in \mathsf{Q} $ is called the {team optimal} for the static information structure if
    $$ J^\mathsf{S}_\mathsf{Q}(\bie^*,\bif^*;p_1,p_2) \leq \Jsqp  \quad \forall (\bie,\bif)\in \mathsf{Q}, $$
    and for the dynamic information structure if
    $$ J^\mathsf{D}_\mathsf{Q}(\bie^*,\bif^*;p_s,p_c) \leq \Jdqp  \quad \forall (\bie,\bif)\in \mathsf{Q}. $$
    If for some $p_1,p_2$ (resp., $p_s,p_c$) there exists a team optimal solution $ (\bie^*, \bif^*) \in \mathsf{A} $, then we say the strategy class $\mathsf{A}$ is team optimal for the static (dynamic) information structure for this $p_1,p_2$ (resp., $p_s,p_c$).\\
    If a strategy class $\mathsf{A}$ is team optimal for all $p_1, p_2$  (resp., $p_s,p_c$), then we say that the strategy class $\mathsf{A}$ is team optimal for the static (dynamic) information structure.
\end{definition}

We explore the team optimal strategies and costs associated with the cost function $\Be$ in the static and dynamic information structure. First we look at the team optimal over classical strategies $\mathsf C$ and then proceed to find the team optimal when players are allowed to play quantum strategies. Finally, we compare both the team optimal to find out the existence and extent of quantum advantage.

\subsection{Classical Strategies}
The following is our result for classical strategies and the two information structures.
\begin{theorem} Optimal expected cost over all classical strategies $\mathsf C$ for static information structure
	\begin{equation}
		\Jscop= \min \{ p_1p_2,\pat p_2,p_1\pbt,\pat\pbt\}.
		\label{eqn:class_inde}
	\end{equation}
\label{thm:jsco}
\end{theorem}

\begin{theorem}
Optimal expected cost over all classical strategies $\mathsf C$ for dynamic information structure
\begin{equation}
	\Jdcop= \min\{ p_sp_c,\pst p_c, p_s\pct,\pst\pct\}.
	\label{eqn:class_channel}
\end{equation}
\label{thm:jdco}
\end{theorem}

Proofs are included in Appendix \ref{sec:3.1} and \ref{sec:3.2}.

It is worth noting that the classical team optimal costs in both information structures are identical if one makes the substitutions $p_1 \leftrightarrow p_s, $ and $ p_2 \leftrightarrow p_c$. This fascinating finding will be important as we compare costs now under other strategies. \revii{The correspondence above says the following. Under classical strategies \A{}'s only lever on \B{}'s observation is the deterministic map $w(y_1):=y_1\oplus\gamma_1(y_1)$, and as $\gamma_1$ ranges over the four $\{0,1\}\to\{0,1\}$ maps, $w$ is one of $y_1$, $\overline{y_1}$, $0$ or $1$: \B{} either sees a relabelled copy of $y_1$ corrupted by the independent noise $v$, or sees pure noise. Neither produces a joint law between the agents that is unavailable in the static problem with inputs $(p_s,p_c)$, and optimising over the four cases returns $\min\{p_s,\pst\}\min\{p_c,\pct\}$ -- the static value under the identification $p_1\leftrightarrow p_s$, $p_2\leftrightarrow p_c$. Within $\mathsf C$ the dual effect thus confers no benefit or loss; the two information structures come apart only once the players measure a shared quantum state.}
\subsection{Quantum Strategies}
\subsubsection{Expected costs}

We first compute the expected costs that players obtain under quantum strategies defined in Section \ref{sec:genq}.
Before expressing our expected costs, we will first define some quantities which help us express these costs in a compact way. Since upon observing $y_1$, \A{}  measures using the $ \{  E_{y_1}, I -  E_{y_1} \} $ basis, the $\pm 1$ valued observable  \A{}  measuring is $A_{y_1} = E_{y_1} - ( I - E_{y_1}) = 2E_{y_1} - I$.
Similarly using $y_2$, \B{} measures $B_{y_2} = 2 F_{y_2} - I$. Let $$\langle A_{y_1} B_{y_2} \rangle :=\bra{\Psi} A_{y_1} \otimes B_{y_2} \ket{\Psi}.$$ It follows that
\begin{equation}
\langle A_iB_j \rangle = 1 + 4\Tr[ \{E_i \otimes F_j \} \rho] - \Tr E_i - \Tr F_j. \label{aibj}
\end{equation}
\begin{lemma}
The expected cost under a quantum strategy $(\bie,\bif)$ for
\begin{enumerate}
	\item static information structure is
	\begin{equation}
		\Jsqp = \half - \half (p_1 p_2 \cdot \langle A_0 B_0 \rangle + p_1 \pbt \cdot \langle A_0 B_1 \rangle + \pat p_2 \cdot \langle A_1 B_0 \rangle - \pat\pbt \cdot \langle A_1 B_1 \rangle).
		\label{Jsq}
	\end{equation}
	\item dynamic information structure is
	\begin{align}
		\Jdqp = &- \sum_{i,j \in \{0,1\} } (-1)^{i \cdot j} \bkP(y_1 = i) \Tr (\{ E_i \otimes F_j \} \rho) + \half p_s \Tr E_0 + \pst \pct \notag \\ &+ \half ((p_s \pct + \pst p_c) \Tr F_0 + (p_s p_c -\pst \pct ) \Tr F_1).
		\label{jdq}
	\end{align}
\end{enumerate}


\label{lemma}
\end{lemma}

The proof is in Appendix \ref{sec:3.3}.
As can be seen, there is no apparent equivalence between the two expressions for the two information structures, unlike in the expressions for the optimal cost under classical strategies.

\rev{\begin{remark}
        Equation~\eqref{jdq} reproduces the classical costs on classical strategies. For $E_0=E_1=F_0=F_1=I$: $\Tr E_i=\Tr F_j=2$, $\Tr(\{E_i\otimes F_j\}\rho)=\ei\fj=1$, and $\sum_{i,j}(-1)^{ij}\bkP(y_1=i)=2p_s$, so \eqref{jdq} gives $-2p_s+p_s+\pst\pct+(p_s\pct+\pst p_c)+(p_sp_c-\pst\pct)=\pst p_c$, matching the direct cost $\bkP(y_1=1,v=0)=\pst p_c$. Similarly $E_i=F_j=O$ gives $\pst\pct$, confirming the constants in \eqref{jdq}.
\end{remark}}

\begin{corollary} For projective strategies, the expected cost for static and dynamic information structures is given as
    \begin{align*}
        \Jspp &= 0.5 - 0.5(p_1 p_2 \cdot \langle A_0 B_0 \rangle + p_1 \pbt \cdot \langle A_0 B_1 \rangle + \pat p_2 \cdot \langle A_1 B_0 \rangle - \pat\pbt \cdot \langle A_1 B_1 \rangle)\\
        \Jdpp &= 0.5 - 0.25(p_s(\langle A_0 B_0 \rangle + \langle A_0 B_1 \rangle) +\pst (\langle A_1 B_0 \rangle - \langle A_1 B_1 \rangle)),
    \end{align*}
    where $\langle A_i B_j \rangle=\cos{2\theta_{1i}}\cos{2\theta_{2j}}+\sin{2\theta_{1i}}\sin{2\theta_{2j}}\cos{(\phi_{1i}+\phi_{2j})}$ for $i,j \in \{0,1\}$.
Moreover,
$$        \Jdp (\boldsymbol{\theta}, \boldsymbol{\phi}, p_s, p_c) = \Jsp (\boldsymbol{\theta}, \boldsymbol{\phi}, p_1 = p_s, p_2 = 0.5)  \quad \forall \boldsymbol{\theta}, \boldsymbol{\phi}.$$
\label{cor:JspJdp}
\end{corollary}
\begin{proof}
For projective measurements, using \eqref{eq:trace_e} and \eqref{eq:trace_ef}, we get $\Tr E_i = \Tr F_j =1$ and
\begin{align*}
\Tr[\{E_i \otimes F_j \}\rho]=\frac14(1+\cos{2\theta_{1i}}\cos{2\theta_{2j}} +\sin{2\theta_{1i}}\sin{2\theta_{2j}}\cos{(\phi_{1i}+\phi_{2j})})
\end{align*}
Using \eqref{aibj}, we get $\langle A_i B_j \rangle=\cos{2\theta_{1i}}\cos{2\theta_{2j}}+\sin{2\theta_{1i}}\sin{2\theta_{2j}}\cos{(\phi_{1i}+\phi_{2j})}.$ The above expressions are then obtained by direct substitution.
\end{proof}

The second result in Corollary \ref{cor:JspJdp} shows that for a dynamic information structure, under projective strategies the expected cost becomes independent of $p_c$. The third result says that for any projective strategy, the cost in the dynamic information structure is equivalent to that in a static information structure when $p_1=p_s$ and $p_2=0.5$. This \rev{is because under a projective measurement, in the dynamic information structure, the output of the channel $y_2$ has a distribution $\Ber(\half)$, regardless of the value of $p_c$.}

\subsubsection{Team Optimal Cost and Quantum Advantage}
We now ascertain the team optimal cost under quantum strategies.
We do a detailed step-by-step analytical optimization of the above expressions  over the POVMs $\bie,\bif$ and obtain closed form expressions for the optimal cost for static and dynamic information structures.
\revii{Informally, the theorem states that optimising over all POVMs gains nothing beyond the better of two options: play projectively, where the cost is set by the Bloch-sphere geometry of the four measurement directions, or play classically. Which wins depends on how far the inputs are from uniform. Near $p_1=p_2=\half$ the projective option is strictly better, the familiar CHSH advantage; for sufficiently biased inputs the classical one takes over.}

\begin{theorem} For a static information structure, the optimal expected cost under quantum strategies is given by
    \begin{equation}\Jsqop=
    \begin{cases}
    \half-\frac{1}{\sqrt{2}}\sqrt{\{p_1^2 + \pat^2\} \{p_2^2 + \pbt^2\}},& \text{if } \left|(2p_1-1)(2p_2-1)\right|
    \le
    4p_1\pat\,p_2\pbt,\\
\Jscop              & \text{otherwise.}
\end{cases}
    \end{equation}
    Projective strategies are team optimal for static information structure.\\
    Classical strategies are team optimal for $p_1, p_2$ satisfying $\left|(2p_1-1)(2p_2-1)\right|
    \geq
    4p_1\pat\,p_2\pbt$.
    \label{thm:optstatic}
\end{theorem}
\begin{proof}
    Let
\[
\bie=(E_0,E_1), \qquad \bif=(F_0,F_1),
\]
where, for each setting, the binary POVM is $\{E_i,I-E_i\}$ for \A{} and
$\{F_j,I-F_j\}$ for \B{}. Recall from Lemma~\ref{lemma} that
\begin{equation}
    \Jsq(\bie,\bif;p_1,p_2)
    = \frac12-\frac12\Big(
    p_1p_2 \langle A_0B_0\rangle
    + p_1\pbt \langle A_0B_1\rangle
    + \pat p_2 \langle A_1B_0\rangle
    - \pat\pbt \langle A_1B_1\rangle
    \Big),
    \label{eq:Jsq-start}
\end{equation}
where
\[
A_i=2E_i-I, \qquad B_j=2F_j-I.
\]
Since $\rho$ is the Bell state, its reduced density matrices are $\frac12I$.
Hence, by \eqref{aibj},
\[
\langle A_iB_j\rangle
=1+4\Tr[(E_i\otimes F_j)\rho]-\Tr E_i-\Tr F_j,
\]
whereby the RHS of \eqref{eq:Jsq-start} is linear \rev{individually in $\bie,\bif$, and bilinear in the pair $(\bie,\bif)$. Thus the minimum is attained at an extreme point of the allowed $\bie,\bif$}.

Now let
$
\mathcal E:=\{\rev{X}\in \mathbb H_2: \rev{X} \succeq 0, I-\rev{X} \succeq 0\}.
$ This is a compact convex set. We claim that the
extreme points of $\mathcal E$ are precisely
$
0, I \text{and the rank-one projections.}$
Indeed, suppose $\rev{X}\in\mathcal E$ has spectral decomposition
\[
\rev{X}=\lambda_1 P+\lambda_2(I-P), \qquad 0\le \lambda_1,\lambda_2\le 1,
\]
where $P = \ketbra{\psi}{\psi}$ is a rank-one orthogonal projection on $\ket{\psi}$.
Then $\rev{X}$ is extreme only when $\lambda_1,\lambda_2\in\{0,1\}$; otherwise one
can perturb the eigenvalues slightly in opposite directions and write $\rev{X}$ as a
nontrivial convex combination of two distinct elements from $\mathcal{E}$. Thus the only extreme
points are $0$, $I$, and rank-one projections. \revii{Conversely, each of these is an extreme point: if $X\in\{O,I\}$ or $X=P$ is a rank-one projection, then writing $X=\half(Y+Z)$ with $Y,Z\in\mathcal E$ forces $Y=Z=X$, since the constraints $0\preceq Y,Z\preceq I$ leave no room to move along an eigenvalue already at $0$ or $1$. Hence the extreme points of $\mathcal E$ are exactly $O$, $I$ and the rank-one projections.} Recall that the POVM ${0,I}$ corresponds to classical strategies.

Since the RHS of \eqref{eq:Jsq-start} is affine in $\bie,\bif$ the minimum of
$\Jsq(\bie,\bif;p_1,p_2)$ over $\mathcal{E}$ is attained where each of
$E_0,\ E_1,\ F_0,\ F_1$
is an extreme point of $\mathcal{E}$. Therefore an optimal static quantum strategy may be
chosen so that each player plays either projective or classical strategies.

    Using \eqref{eq:trace_e} and \eqref{eq:trace_ef} in \eqref{aibj}, we get
    \[
    \langle A_iB_j \rangle
    =1+4\Tr[(E_i\otimes F_j)\rho]-\Tr E_i-\Tr F_j
    =1+4(\ei\fj+\evi \odot \fvj)-2\ei-2\fj.
    \]
    Substituting this into \eqref{eq:Jsq-start} and collecting terms gives
    \begin{align}
        \Jsqp
        &= \pat\pbt
        -2p_1 \bigl(p_2\fva+\pbt\fvb\bigr)\odot \eva
        -2\pat \bigl(p_2\fva-\pbt\fvb\bigr)\odot \evb
        +p_1\ea+\pat(p_2-\pbt)\eb \notag\\
        &\quad
        +p_2\fa+\pbt(p_1-\pat)\fb-2\Big((p_2\fa+\pbt\fb)p_1\ea+(p_2\fa-\pbt\fb)\pat\eb\Big).
        \label{eq:Jsq-explicit2}
    \end{align}
    To analyse this expression we now consider two cases.

\medskip
\noindent
{\bf Case 1: all four are rank-one projections.}
Suppose $E_0,E_1,F_0,F_1$ are all rank-one projections. Using the
parametrization from Section~\ref{sec:paramquant}, write
\[
E_i=\ei I+\evi\cdot \boldsymbol{\sigma},\qquad
F_j=\fj I+\fvj\cdot \boldsymbol{\sigma},\qquad i,j\in\{0,1\}.
\]
Since each of $E_i$ and $F_j$ is a rank-one projection, its
eigenvalues are $1$ and $0$. Hence
$\Tr E_i=\Tr F_j=1,$ $i,j\in\{0,1\}$.
Using \eqref{eq:trace_e}, we obtain
$\ei=\fj=\frac12, $ $ i,j\in\{0,1\}.$

Now the eigenvalues of a Hermitian matrix of the form
$aI+\mathbf{v}\cdot\boldsymbol{\sigma}$ are $a\pm |\mathbf{v}|$. Therefore the
eigenvalues of $E_i=\ei I+\evi\cdot\boldsymbol{\sigma}$ are
$\ei\pm |\evi|=\frac12\pm |\evi|$.
Hence
$
|\evi|=\frac12, i \in \{0,1\}.
$
Similarly,
$
|\fvj|=\frac12, j\in\{0,1\}.
$
Thus
%
\[
A_i=2E_i-I=2\,\evi\cdot\boldsymbol{\sigma},\qquad
B_j=2F_j-I=2\,\fvj\cdot\boldsymbol{\sigma}.
\]
Also, by \eqref{eq:trace_ef},
$
\Tr[(E_i\otimes F_j)\rho]
=\ei\fj+\evi\odot\fvj
=\frac14+\evi\odot\fvj.
$
Substituting this into \eqref{aibj}, we get
\[
\langle A_iB_j\rangle
=1+4\Tr[(E_i\otimes F_j)\rho]-\Tr E_i-\Tr F_j
=1+4\left(\frac14+\evi\odot\fvj\right)-1-1
=4\,\evi\odot\fvj.
\]

Hence \eqref{eq:Jsq-start} becomes
\begin{align*}
    \Jsqp
    &=
    \frac12-\frac12\Big(
    4p_1p_2\,\eva\odot \fva
    +4p_1\pbt\,\eva\odot \fvb
    +4\pat p_2\,\evb\odot \fva
    -4\pat\pbt\,\evb\odot \fvb
    \Big)\\
    &=
    \frac12-2\Big(
    p_1\,\eva\odot (p_2\fva+\pbt\fvb)
    +\pat\,\evb\odot (p_2\fva-\pbt\fvb)
    \Big).
\end{align*}
We now minimize the above under the constraints
\[
\ea=\eb=\fa=\fb=\frac12,\qquad |\eva|=|\evb|=|\fva|=|\fvb|=\frac12.
\]

For fixed $\fva,\fvb$, since $|\eva|=|\evb|=\frac12$, the minimum is attained when
$\eva$ is aligned with $M(p_2\fva+\pbt\fvb)$ and $\evb$ is aligned with
$M(p_2\fva-\pbt\fvb)$, where $M$ is $\diag(1,-1,1)$. Hence
\[
\Jsqp
=
\frac12-p_1|p_2\fva+\pbt\fvb|-\pat|p_2\fva-\pbt\fvb|.
\]

Now let $$t:=4\,\fva \cdot \fvb \in[-1,1].$$
Since $|\fva|=|\fvb|=\frac12$, we have
\[
\Jsqp
=
\frac12-\frac12\Big(
p_1\sqrt{p_2^2+\pbt^2+2p_2\pbt\,t}
+\pat\sqrt{p_2^2+\pbt^2-2p_2\pbt\,t}
\Big).
\]
Thus it remains to minimize the above
over $t\in[-1,1]$.

Therefore the interior optimum is attained precisely under the above symmetric
condition.

We see that critical point $t=t^*$ satisfies
\[
\frac{p_1}{\sqrt{p_2^2+\pbt^2+2p_2\pbt t^*}}
=
\frac{\pat}{\sqrt{p_2^2+\pbt^2-2p_2\pbt t^*}},
\]
which yields
\[
t^*
=
\left(\frac{p_2^2+\pbt^2}{2p_2\pbt}\right)
\left(\frac{p_1^2-\pat^2}{p_1^2+\pat^2}\right).
\]
$t^*$ is feasible if and only if
\[
\left|
\left(\frac{p_2^2+\pbt^2}{2p_2\pbt}\right)
\left(\frac{p_1^2-\pat^2}{p_1^2+\pat^2}\right)
\right|\le 1.
\]
Squaring both sides and rearranging, this is equivalent to
\[
\left|(2p_1-1)(2p_2-1)\right|
\le
4p_1\pat\,p_2\pbt.
\]
If this condition holds, then substituting for $t^*$ gives

\begin{align}
    \Jsqp
&=
\frac{1}{2}-\frac{1}{\sqrt2}\sqrt{(p_1^2+\pat^2)(p_2^2+\pbt^2)}. \label{eq:case11}
\end{align}

If $t^*\notin[-1,1]$, then the minimum is attained at one of the boundary
points $t=\pm1$. It is easy to see that for $t=1$
\[
\Jsqp =
\begin{cases}
    \pat\pbt, & p_2\ge \frac12,\\[1mm]
    \pat p_2, & p_2\le \frac12,
\end{cases}
\]
whereas for $t=-1$
\[
\Jsqp=
\begin{cases}
    p_1\pbt, & p_2\ge \frac12,\\[1mm]
    p_1p_2, & p_2\le \frac12.
\end{cases}
\]
Thus, in the boundary case,
\[
\Jsqp=\min\{p_1p_2,\pat p_2,p_1\pbt,\pat\pbt\} = \Jscop,
\]
from Theorem~\ref{thm:jsco}.
Therefore, in this case,

\begin{equation}
    \min_{\bie,\bif}\Jsqp=
\begin{cases}
    \dfrac{1}{2}-\dfrac{1}{\sqrt2}\sqrt{(p_1^2+\pat^2)(p_2^2+\pbt^2)},
    & \text{if}
\left|(2p_1-1)(2p_2-1)\right|
\le
4p_1\pat\,p_2\pbt,\\[2ex]
    \Jscop,
    & \text{otherwise,}
\end{cases} \label{eq:case1}
\end{equation}
where the minimum is over $\bie,\bif$ satisfying Case 1. Notice that the RHS of \eqref{eq:case11} no greater than $\Jscop$.

\medskip
\noindent
{\bf Case 2: at least one of the four is degenerate.}
WLOG assume $E_0$ is degenerate.
Since $E_0$ is extreme, either $E_0=0$ or $E_0=I$.
Assume
$E_0=I$; the other case can be handled similarly. Then
$\ea=1,\eva=0.$

After the extreme-point reduction, each of $E_1,F_0,F_1$ is either degenerate
or rank-one projective. We consider the possible subcases.

\smallskip
\noindent
{\it Subcase 2a: $E_1$ is also degenerate.}
Then \A{}'s outputs are deterministic for both inputs. Hence the optimization
over $F_0,F_1$ reduces to the classical deterministic problem, whose optimal
value is $\Jsco(p_1,p_2)$.

\smallskip
\noindent
{\it Subcase 2b: $E_1$ is projective and both $F_0,F_1$ are projective.}
Then, as in case 1,
$
\eb=\fa=\fb=\frac12 \text{and}|\evb|=|\fva|=|\fvb|=\frac12.
$
Substituting $\ea=1$, $\eva=0$ into \eqref{eq:Jsq-start}, we obtain
\[
\Jsqp
=
\frac12-2\pat\,\evb\odot (p_2\fva-\pbt\fvb).
\]
For fixed $\fva,\fvb$, the minimum is attained when $\evb$ is parallel to
$M(p_2\fva-\pbt\fvb)$, and therefore
\[
\Jsqp
=
\frac12-\pat\,|p_2\fva-\pbt\fvb|.
\]
Since $|\fva|=|\fvb|=\frac12$, by Cauchy--Schwarz inequality,
\[
|p_2\fva-\pbt\fvb|
\le p_2|\fva|+\pbt|\fvb|
=\frac12.
\]
Hence
\[
\Jsqp\ge \frac12-\frac{\pat}{2}=\frac{p_1}{2}.
\]
Now write
\[
\frac{p_1}{2}=\frac{p_1p_2+p_1\pbt}{2}\ge \Jsco(p_1,p_2),
\]
since $\Jsco=\min\{p_1p_2,\pat p_2,p_1\pbt,\pat\pbt\}$ is no larger than each
of $p_1p_2$ and $p_1\pbt$, and no larger than their average. Thus in this subcase,
\[
\Jsqp\ge \Jsco(p_1,p_2).
\]

\smallskip
\noindent
{\it Subcase 2c: $E_1$ is projective, exactly one of $F_0,F_1$ is degenerate and the other is projective.}
Let $F_0=sI$ with $s\in\{0,1\}$ and $F_1$ be projective; the other possibility follows similarly. Then
$
\fa=s, \fva=0$ and $\fb=\frac12 ,|\fvb|=\frac12.
$
Substituting into \eqref{eq:Jsq-start}, we get
\[
\Jsqp
=
\frac12+\frac{p_1p_2}{2}-p_1p_2 s+2\pat\pbt\,\evb\odot\fvb.
\]
By Cauchy--Schwarz inequality, since  $|\evb|=|\fvb|=\frac12$,
$\evb\odot\fvb = \evb \cdot M \fvb \ge - |\evb| |\fvb| \ge -\frac14.$
Therefore
\[
\Jsqp
\ge
\frac12+\frac{p_1p_2}{2}-p_1p_2 s-\frac{\pat\pbt}{2} \geq \frac{p_1\pbt+\pat p_2}{2}
\ge \Jsco(p_1,p_2).
\]

\smallskip
\noindent
{\it Subcase 2d: $E_1$ is projective and both $F_0,F_1$ are degenerate.} Then \B{} plays a classical strategy, whereby \A{}'s optimal choice is also classical.
Hence the optimal value in this subcase is at least $\Jsco(p_1,p_2)$.

We have therefore shown that every subcase with $E_0$ degenerate satisfies
\[
\Jsqp\ge \Jsco(p_1,p_2).
\]
 Hence the optimum in the
branch where $E_0$ (or $E_1$ or any $F_j$'s) is degenerate is exactly $\Jscop$.

Finally, $\Jsqop $ is given by the minimum of the RHS of \eqref{eq:case1} and $\Jscop$. But since the RHS of \eqref{eq:case1} is no more than $\Jscop$, we get the result claimed in the statement of the theorem. Note that the RHS of \eqref{eq:case1} is achieved by projective strategies (i.e. Case 1) whereby projective strategies are team optimal for the static information structure. Classical strategies are team optimal for $p_1,p_2$ satisfying $\Jsqop = \Jsco(p_1,p_2)$.
\end{proof}


We have plotted the fractional quantum advantage for static information structure $(\Jsco -\Jspo)/\Jsco$ in Fig. \ref{fig:inde}. Here the two axes represent the two probabilities of which the optimal costs are a function. It is a 3D colour plot so the quantity at hand is represented as shown in the bar. The region with purple colour represents no quantum advantage. We can also see that the highest advantage is at the center of the plot, which is to be expected as that represents the usual Bell's polynomial, corresponding to the Tsirelson's bound \cite{tsirelson1980genbell}. Also there is a symmetry in exchange of the coordinates and also around the point $(\half,\half)$. The first corresponds to relabeling of the players, whereas the second corresponds to exchange of natural states from $0$ to $1$.
\begin{figure}[h]
    \centering
    \includegraphics[scale=0.8]{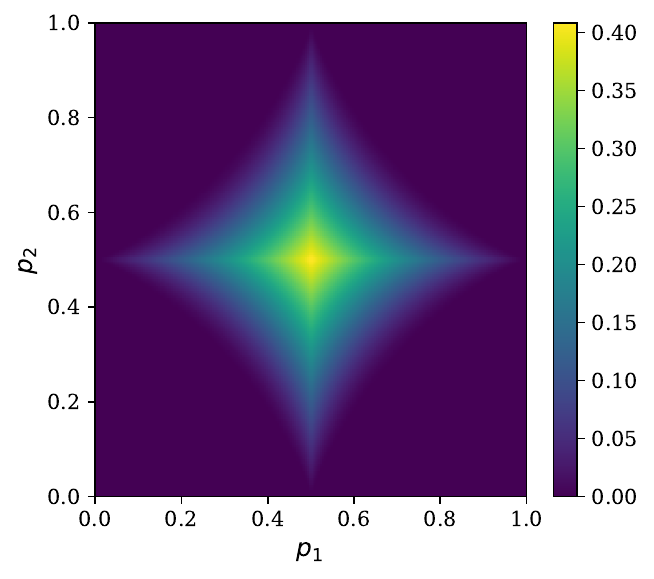}
    \caption{Fractional quantum advantage for static information structure:= $(\Jsco -\Jspo)/\Jsco$. \rev{The plot is a numerical evaluation of the closed-form optima of Theorems \ref{thm:jsco} and \ref{thm:optstatic} over the $(p_1,p_2)$ grid.} \revii{The horizontal and vertical axes are $p_1=\Pbb(y_1=0)$ and $p_2=\Pbb(y_2=0)$; the colour bar gives the fractional advantage, purple ($0$) marking the parameters at which the quantum and classical team optimal costs coincide and brighter shades a strictly smaller quantum cost. The maximum is at $p_1=p_2=\half$, the standard CHSH game, where $\Jspo(\half,\half)\approx0.146$ improves on $\Jsco(\half,\half)=0.25$ by about $41\%$.}}
    \label{fig:inde}
\end{figure}

 We now examine if projective strategies are team optimal for the dynamic information structure.

\begin{theorem} For a dynamic information structure, the optimal expected cost under quantum strategies is given by
	\begin{equation}
		\Jdqop = \min \left\{\half-\half \sqrt{p_s^2 + \pst^2 } , \Jdcop \right\}.
		\label{eq:optdyn}
	\end{equation}
    Projective strategies are team optimal for $p_s, p_c$ satisfying \begin{equation}
        \half-\half \sqrt{p_s^2 + \pst^2 } \leq \min\{ p_sp_c,\pst p_c, p_s\pct,\pst\pct\}.\label{eq:dynamic_proj}
    \end{equation}
    Classical strategies are team optimal for $p_s, p_c$ satisfying
    \begin{equation}
        \half-\half \sqrt{p_s^2 + \pst^2 } \geq \min\{ p_sp_c,\pst p_c, p_s\pct,\pst\pct\}. \label{eq:dynamic_class}
    \end{equation}
\label{thm:optdyn}
\end{theorem}

\begin{proof}
    \rev{We first justify restricting attention to extreme points. The cost \eqref{jdq} is \emph{not} jointly affine in $(\bie,\bif)$, but it is bilinear in these variables. Let
    \[
    \mathcal E:=\{X \in \mathbb H_2 : 0 \preceq X \preceq I\},
    \]
    so that each of $E_0,E_1,F_0,F_1$ ranges over $\mathcal E$ and the feasible set is $\mathcal E^4$. Then: (i) with any three of the four effects held fixed, \eqref{jdq} is affine in the remaining one, since each effect enters it only to the first power; (ii) $\mathcal E$ is convex and compact and $\Jdqp$ is continuous, so a minimiser exists; (iii) given a minimiser, replace $E_0$ by an extreme point of $\mathcal E$ minimising the affine function obtained by holding the other three effects fixed -- by (i) the cost cannot increase, and since it was already minimal it is unchanged, so the new tuple is again a minimiser, with the other three effects untouched; repeating this for $E_1$, then $F_0$, then $F_1$ yields, after four replacements, a minimiser all of whose effects are extreme points of $\mathcal E$; and (iv) as established in the proof of Theorem~\ref{thm:optstatic}, the extreme points of $\mathcal E$ are precisely $O$, $I$ and the rank-one projections. Hence an optimum may be chosen in which each of $E_0,E_1,F_0,F_1$ is either degenerate ($O$ or $I$, that is, classical) or a rank-one projection.}

    We now consider two cases.

\medskip
\noindent
{\bf Case 1: at least one of $E_0,E_1,F_0,F_1$ is degenerate.}
\rev{Throughout this case it is convenient to have \eqref{jdq} collected in the scalars and the Bloch vectors. Using \eqref{eq:trace_e} and \eqref{eq:trace_ef},
\begin{align}
    \Jdqp
    &= \pst\pct
     + p_s\,\ea\,(1-\fa-\fb)
     + \pst\,\eb\,(\fb-\fa) \non\\
    &\quad + (p_s\pct+\pst p_c)\,\fa
           + (p_sp_c-\pst\pct)\,\fb
     - p_s\,\eva\odot(\fva+\fvb)
     - \pst\,\evb\odot(\fva-\fvb).
    \label{eq:jdq_collected}
\end{align}
The dynamic cost is not symmetric between the two players, nor between \B{}'s two inputs: the term $\half p_s\Tr E_0$ in \eqref{jdq} distinguishes $E_0$ from $E_1$, and the weight $(-1)^{ij}$ distinguishes $F_0$ from $F_1$. Accordingly we treat all four branches --- $E_0$, $E_1$, $F_0$ and $F_1$ degenerate --- explicitly, and none is disposed of by symmetry.}

First assume that $E_0$ is degenerate. Then
$
\eva=0 \text{and} \ea\in\{0,1\}.
$
Substituting this into \eqref{jdq} and using \eqref{eq:trace_e} and
\eqref{eq:trace_ef}, we obtain
\begin{align}
    \Jdqp
    &=
    \pst \pct
    + p_s \ea(1-\fa-\fb)
    +\fa\bigl(p_s\pct+\pst p_c-\pst\eb\bigr)
\non \\     &+\fb\bigl(p_sp_c-\pst\pct+\pst\eb\bigr)
    +\pst\,\evb\odot(\fvb-\fva). \label{eq:case1ex}
\end{align}

We first minimize with respect to $\evb$. Since
$
|\evb|\le \min\{\eb,1-\eb\},
$
for fixed $\eb,\fa,\fb,\fva,\fvb$ the last term in \eqref{eq:case1ex} is minimized by
choosing $\evb$ anti-parallel to $M(\fvb-\fva)$, with maximal allowed norm. Hence
\begin{align}
    \Jdqp
    &=
    \pst \pct
    + p_s \ea(1-\fa-\fb)
    +\fa\bigl(p_s\pct+\pst p_c-\pst\eb\bigr)
    +\fb\bigl(p_sp_c-\pst\pct+\pst\eb\bigr) \notag\\
    &\qquad
    -\pst \min\{\eb,1-\eb\}\,|\fvb-\fva|.
    \label{eq:case1-reduced}
\end{align}

It remains to optimize over $F_0$ and $F_1$. By the extreme-point reduction,
each of $F_0,F_1$ is either degenerate or rank-one projective.

\smallskip
\noindent
{\it (i) Both $F_0$ and $F_1$ are degenerate.}
Then
$
\fva=\fvb=0 \text{and} \fa,\fb\in\{0,1\},
$
so \eqref{eq:case1-reduced} becomes affine in $\eb$. Hence the minimum over
$\eb$ is attained at $\eb\in\{0,1\}$, i.e. at a deterministic choice of $E_1$.
Thus the resulting strategy is classical, and the minimum in this branch is
$\Jdcop$.

\smallskip
\noindent
{\it (ii) Exactly one of $F_0,F_1$ is projective.}
Assume first that $F_0$ is projective and $F_1$ is degenerate. Then
$
\fa=\frac12,|\fva|=\frac12,\fb\in\{0,1\}, \fvb=0.
$
Also since $E_1$ is either degenerate or a projection, $\eb \in \{0,\half,1\}$.
Using the above in \eqref{eq:case1-reduced}, we obtain
\[
\Jdqp=
\begin{cases}
    \displaystyle \frac{\pct}{2}+\frac{p_s}{2}\ea, & \fb=0,\\[2ex]
    \displaystyle \frac{p_c}{2}+\frac{p_s}{2}(1-\ea), & \fb=1.
\end{cases}
\]
Hence
\[
\Jdqp\ge \min\left\{\frac{\pct}{2},\frac{p_c}{2}\right\} \geq \Jdcop
\]
It follows that
\[
\Jdqp\ge \Jdcop.
\]
\rev{The remaining possibility, $F_0$ degenerate and $F_1$ projective, is treated by
the same substitution: putting $\fa\in\{0,1\}$, $\fva=0$, $\fb=\half$ and
$|\fvb|=\half$ into \eqref{eq:case1-reduced} again yields one affine expression
for each value of $\fa$, and each is bounded below by $\Jdcop$. }

\smallskip
\noindent
{\it (iii) Both $F_0$ and $F_1$ are projective.}
Then
$
\fa=\fb=\frac12 \text{and}|\fva|=|\fvb|=\frac12.
$
It is easy to that \eqref{eq:case1-reduced} becomes
\[
\Jdqp
=
\frac12-\pst \min\{\eb,1-\eb\}\,|\fvb-\fva|.
\]
The minimum is attained at \(\eb=\frac12\), and since $|\fva - \fvb| \leq 1$
\[
\Jdqp
\ge
\frac12-\frac{\pst}{2}
=
\frac{p_s}{2} \geq \Jdcop.
\]

Combining (i)--(iii), we conclude that whenever $E_0$ is degenerate,
\[
\Jdqp\ge \Jdcop,
\]
and equality is attained by a deterministic classical strategy. Hence the
optimal value if $E_0$ is degenerate is $\Jdcop$.

\rev{Next assume that $E_1$ is degenerate. Then $\evb=0$ and $\eb\in\{0,1\}$, so the
last term of \eqref{eq:jdq_collected} drops, and the remaining vector term is
minimized by taking $\eva$ parallel to $M(\fva+\fvb)$ with maximal allowed norm,
contributing $-p_s\min\{\ea,1-\ea\}\,|\fva+\fvb|$. We consider the three
sub-branches on $(F_0,F_1)$.

If both $F_0$ and $F_1$ are degenerate, then $\fva=\fvb=0$ and the cost is affine
in $\ea$, so the minimum is attained at $\ea\in\{0,1\}$; the strategy is then
deterministic classical and the value is $\Jdcop$.

If exactly one of $F_0,F_1$ is projective, then $|\fva+\fvb|=\half$, and using
$\ea+\min\{\ea,1-\ea\}\le 1$ together with $\eb\in[0,1]$ we obtain
$\Jdqp\ge\half\min\{p_c,\pct\}$ when $F_1$ is the projective one, and
$\Jdqp\ge\half\min\{\pst\pct+p_sp_c,\;p_s\pct+\pst p_c\}$ when $F_0$ is. Both
dominate $\Jdcop$: the first since $\Jdcop=\min\{p_s,\pst\}\min\{p_c,\pct\}$ with
$\min\{p_s,\pst\}\le\half$, and the second since, by the
arithmetic--geometric mean inequality,
$\half(\pst\pct+p_sp_c)\ge\sqrt{p_s\pst p_c\pct}\ge\min\{p_s,\pst\}\min\{p_c,\pct\}$,
and likewise for the other term.

If both $F_0$ and $F_1$ are projective, then $\fa=\fb=\half$, the coefficients of
$\ea$ and $\eb$ in \eqref{eq:jdq_collected} vanish, and the scalar terms collapse
to $\half(p_s+\pst)(p_c+\pct)=\half$. Hence
\[
\Jdqp=\half-p_s\min\{\ea,1-\ea\}\,|\fva+\fvb|
\ge \half-\frac{p_s}{2}=\frac{\pst}{2}\ge\Jdcop.
\]
Therefore, whenever $E_1$ is degenerate, $\Jdqp\ge\Jdcop$, with equality attained
by a deterministic classical strategy.}

Now assume that $F_0$ is degenerate. Then
$
\fva=0, \fa\in\{0,1\}.
$
We distinguish two subcases.

\smallskip
\noindent
{\it Subcase 1: $F_1$ is degenerate.}
Then
$
\fvb=0, \fb\in\{0,1\}.
$
Substituting into \eqref{jdq}, the cost becomes affine in $\ea,\eb$. Hence the
minimum is attained at $\ea,\eb\in\{0,1\}$, i.e., at a deterministic classical
strategy. Therefore, the minimum in this subcase is $\Jdcop$.

\smallskip
\noindent
{\it Subcase 2: $F_1$ is projective.}
Then
$
\fb=\frac12, |\fvb|=\frac12.
$
Substituting $\fva=0$ and $\fb=\frac12$ into \eqref{jdq}, and using
\eqref{eq:trace_e} and \eqref{eq:trace_ef}, we obtain
\begin{align}
    \Jdqp
    &=
    p_s\ea\Bigl(\frac12-\fa\Bigr)
    +\pst\eb\Bigl(\frac12-\fa\Bigr)
    +\pst\pct
    +\fa\bigl(p_s\pct+\pst p_c\bigr)
    +\frac12\bigl(p_sp_c-\pst\pct\bigr)\notag\\
    &\qquad
    +\bigl(\pst\,\evb-p_s\,\eva\bigr)\odot\fvb .
    \label{eq:F0deg-proj}
\end{align}
For fixed $\ea,\eb,\eva,\evb$, the last term is minimized by choosing $\fvb$
anti-parallel to $M(\pst\,\evb-p_s\,\eva)$, and since $|\fvb|=\frac12$, this gives
\[
\bigl(\pst\,\evb-p_s\,\eva\bigr)\odot\fvb
=
-\frac12\,|\pst\,\evb-p_s\,\eva|.
\]
Hence
\begin{align}
    \Jdqp
    &=
    p_s\ea\Bigl(\frac12-\fa\Bigr)
    +\pst\eb\Bigl(\frac12-\fa\Bigr)
    +\pst\pct
    +\fa\bigl(p_s\pct+\pst p_c\bigr)
    +\frac12\bigl(p_sp_c-\pst\pct\bigr)\notag\\
    &\qquad
    -\frac12\,|\pst\,\evb-p_s\,\eva|.
    \label{eq:F0deg-proj2}
\end{align}

If $\fa=0$, then \eqref{eq:F0deg-proj2} becomes
\[
\Jdqp
=
\frac12\bigl(p_sp_c+\pst\pct\bigr)
+\frac{p_s}{2}\ea
+\frac{\pst}{2}\eb
-\frac12\,|\pst\,\evb-p_s\,\eva|.
\]
Using
\[
|\pst\,\evb-p_s\,\eva|
\le
\pst\,|\evb|+p_s\,|\eva|
\le
\pst\min\{\eb,1-\eb\}+p_s\min\{\ea,1-\ea\},
\]
we get
\begin{align*}
    \Jdqp
    &\ge
    \frac12\bigl(p_sp_c+\pst\pct\bigr)
    +\frac{p_s}{2}\Bigl(\ea-\min\{\ea,1-\ea\}\Bigr)
    +\frac{\pst}{2}\Bigl(\eb-\min\{\eb,1-\eb\}\Bigr)\\
    &\ge
    \frac12\bigl(p_sp_c+\pst\pct\bigr)
    \ge \Jdcop.
\end{align*}

If $\fa=1$, then \eqref{eq:F0deg-proj2} becomes
\[
\Jdqp
=
\frac12\bigl(p_s\pct+\pst p_c\bigr)
+\frac{p_s}{2}(1-\ea)
+\frac{\pst}{2}(1-\eb)
-\frac12\,|\pst\,\evb-p_s\,\eva|.
\]
Using the same bound,
\[
|\pst\,\evb-p_s\,\eva|
\le
\pst\min\{\eb,1-\eb\}+p_s\min\{\ea,1-\ea\},
\]
we get
\begin{align*}
    \Jdqp
    &\ge
    \frac12\bigl(p_s\pct+\pst p_c\bigr)
    +\frac{p_s}{2}\Bigl(1-\ea-\min\{\ea,1-\ea\}\Bigr)\\
    &\qquad
    +\frac{\pst}{2}\Bigl(1-\eb-\min\{\eb,1-\eb\}\Bigr)\\
    &\ge
    \frac12\bigl(p_s\pct+\pst p_c\bigr)
    \ge \Jdcop.
\end{align*}

Therefore, whenever $F_0$ is degenerate,
$
\Jdqp\ge \Jdcop.
$
Since equality is attained by a deterministic classical strategy, the optimal
value in this branch is $\Jdcop$.

\rev{Finally, assume that $F_1$ is degenerate. If any of $E_0,E_1,F_0$ is also
degenerate then one of the preceding branches applies, so we may assume that
$F_1$ is the only degenerate effect: $\fvb=0$ and $\fb\in\{0,1\}$, while
$\ea=\eb=\fa=\half$ with $|\eva|=|\evb|=|\fva|=\half$. The two vector terms of
\eqref{eq:jdq_collected} then collapse onto $\fva$ alone. Since $M$ is
orthogonal, $|M\fva|=|\fva|=\half$, and by the triangle inequality
$|p_s\eva+\pst\evb|\le p_s|\eva|+\pst|\evb|=\half$, whence
\[
(p_s\eva+\pst\evb)\odot\fva=(p_s\eva+\pst\evb)\cdot M\fva\le\tfrac14 .
\]
Substituting into \eqref{eq:jdq_collected} and using $p_s+\pst=1$, $p_c+\pct=1$
gives $\Jdqp\ge\half\pct$ when $F_1=O$ and $\Jdqp\ge\half p_c$ when $F_1=I$.
Hence in this branch
\[
\Jdqp\ \ge\ \half\min\{p_c,\pct\}\ \ge\ \Jdcop,
\]
the last step because $\Jdcop=\min\{p_s,\pst\}\min\{p_c,\pct\}$ and
$\min\{p_s,\pst\}\le\half$. The bound is tight, the value in this branch being
exactly $\half\min\{p_c,\pct\}$.

Case 1 is therefore exhausted by the four branches in which $E_0$, $E_1$, $F_0$
or $F_1$ is degenerate, in each of which the remaining three effects range over
all extreme points of $\mathcal E$.}

\medskip
\noindent
{\bf Case 2: all four $E_0,E_1,F_0,F_1$ are rank-one projections.}
In this case
\[
\ea=\eb=\fa=\fb=\frac12,\qquad
|\eva|=|\evb|=|\fva|=|\fvb|=\frac12.
\]
Using \eqref{jdq}, together with \eqref{eq:trace_e} and
\eqref{eq:trace_ef}, we obtain
\begin{align*}
    \Jdqp
    &=
    -p_s\Tr[(E_0\otimes F_0)\rho]
    -p_s\Tr[(E_0\otimes F_1)\rho]
    -\pst\Tr[(E_1\otimes F_0)\rho]
    +\pst\Tr[(E_1\otimes F_1)\rho]\\
    &\qquad
    +\frac12 p_s\Tr E_0+\pst\pct
    +\frac12\Big((p_s\pct+\pst p_c)\Tr F_0+(p_sp_c-\pst\pct)\Tr F_1\Big).
\end{align*}
Now, recall, as in the proof of Theorem~\ref{thm:optstatic},
$
\Tr E_0=\Tr F_0=\Tr F_1=1,
$
and
\[
\Tr[(E_i\otimes F_j)\rho]=\ei\fj+\evi\odot\fvj
=\frac14+\evi\odot\fvj.
\]
Substituting these into the above expression gives
\begin{align*}
    \Jdqp
    &=
    -p_s\Big(\frac14+\eva\odot\fva\Big)
    -p_s\Big(\frac14+\eva\odot\fvb\Big)
    -\pst\Big(\frac14+\evb\odot\fva\Big)
    +\pst\Big(\frac14+\evb\odot\fvb\Big)\\
    &\qquad
    +\frac12 p_s+\pst\pct
    +\frac12\Big((p_s\pct+\pst p_c)+(p_sp_c-\pst\pct)\Big).
\end{align*}
The scalar terms simplify,
yielding
\[
\Jdqp
=
\frac12-p_s\,\eva\odot(\fva+\fvb)-\pst\,\evb\odot(\fva-\fvb).
\]

For fixed $\fva,\fvb$, the minimum is attained when $\eva$ is aligned with
$M(\fva+\fvb)$ and $\evb$ is aligned with $M(\fva-\fvb)$. Since
$
|\eva|=|\evb|=\frac12,
$
this yields
\[
\Jdqp
=
\frac12-\frac{p_s}{2}\,|\fva+\fvb|-\frac{\pst}{2}\,|\fva-\fvb|.
\]
By Cauchy--Schwarz inequality,
\[ \Jdqp \geq \half - \half \sqrt{p_s^2 + \pst^2 }, \]
since
\[
|\fva+\fvb|^2+|\fva-\fvb|^2
=
2|\fva|^2+2|\fvb|^2
=
1,
\]

%

It remains to show that equality is attainable. We choose $\fva,\fvb$ with
$|\fva|=|\fvb|=\frac12$ and satisfying
\[
|\fva+\fvb|=\frac{p_s}{\sqrt{p_s^2+\pst^2}},\qquad |\fva-\fvb|=\frac{\pst}{\sqrt{p_s^2+\pst^2}}.
\]
For instance, take
\[
\fva=\frac12(1,0,0),\qquad
\rev{\fvb=\frac12(\cos\beta,0,\sin\beta),}
\]
where $\beta\in[0,\pi]$ is chosen so that
\[
\cos(\beta/2)=\frac{p_s}{\sqrt{p_s^2+\pst^2}},\qquad \sin(\beta/2)=\frac{\pst}{\sqrt{p_s^2+\pst^2}},.
\]
This attains equality in the above bound, and therefore, the optimal value in this case is,
$
\Jdqp=\half -\half \sqrt{p_s^2+\pst^2}.
$
Combining both cases, we have $$\Jdqop = \min \left\{\half-\half \sqrt{p_s^2 + \pst^2 } , \Jdcop \right\}.$$
We further note, projective strategies are team optimal for $p_s, p_c$ satisfying \eqref{eq:dynamic_proj} and classical strategies are team  optimal for $p_s, p_c$ satisfying \eqref{eq:dynamic_class}.
\end{proof}

For the dynamic information structure, as demonstrated by the proof above, projective strategies are not team optimal. In particular, for $p_s, p_c$ satisfying \eqref{eq:dynamic_class} strictly, \rev{projective strategies result in a strictly higher cost than classical ones, meaning there is a projective \textit{disadvantage}}.

\revii{\begin{rmk}\label{rem:disadv}
\emph{(Projective disadvantage.)} Take $p_s=\half$ and $p_c=0.8$. Theorem~\ref{thm:jdco} gives $\Jdcop=\min\{0.4,0.4,0.1,0.1\}=0.1$, while by Theorem~\ref{thm:optdyn} the best projective strategies achieve only $\Jdpop=\half-\half\sqrt{p_s^2+\pst^2}=\half-\tfrac{1}{2\sqrt2}\approx0.1464$: playing projectively costs the team about $46\%$ more than playing classically.

The effect is largest on a noiseless channel. Since for any $p_s,p_c$, we have $\Jdpop\le\half-\tfrac{1}{2\sqrt2}$ with equality at $p_s=\half$, and $\Jdcop\ge0$ with equality at $p_c\in\{0,1\}$, the projective disadvantage $\Jdpop-\Jdcop$ never exceeds $\half-\tfrac{1}{2\sqrt2}\approx0.1464$, and attains this value at $p_s=\half$, $p_c\in\{0,1\}$. There the team wins with \emph{certainty} classically. With $u_1=0$ the channel output is $y_2=y_1\oplus v$ with $v$ deterministic, so at $p_c=1$ we have $y_2=y_1$ and \B{} plays $u_2=y_2$, giving $u_1\oplus u_2=y_1=y_1\cdot y_2$; at $p_c=0$ we have $y_2=y_1\oplus1$, so that $y_1\cdot y_2=0$ and \B{} plays $u_2=0$. In either case the cost $u_1\oplus u_2\oplus y_1\cdot y_2$ vanishes surely, whereas the best projective strategy still loses with probability $\approx0.146$.
\end{rmk}}

\begin{figure}[h]
    \centering
    \includegraphics[scale=0.8]{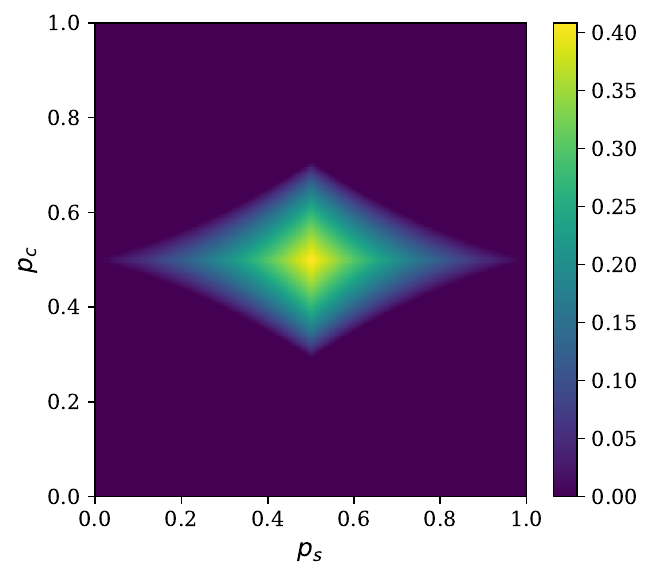}
    \caption{Fractional quantum advantage for dynamic information structure. \rev{As in Fig.~\ref{fig:inde}, the plot is a numerical evaluation of the closed-form optima of Theorems \ref{thm:jdco} and \ref{thm:optdyn} over the $(p_s,p_c)$ grid.} \revii{The horizontal and vertical axes are $p_s$ and $p_c$, and the fractional advantage and colour bar are defined as in Fig.~\ref{fig:inde}. Two differences from Fig.~\ref{fig:inde} are visible: the region of quantum advantage is strictly smaller, and in part of the complementary region projective strategies are not merely tied with classical ones but strictly worse (Remark~\ref{rem:disadv}), which never happens in the static case.}}
    \label{fig:nonclassical}
\end{figure}

Fig~\ref{fig:nonclassical} shows the fractional quantum advantage for the dynamic information structure. Similar features like Fig. \ref{fig:inde}, such as maxima at the center and symmetry across quadrants, are observed in Fig. \ref{fig:nonclassical}. The main difference is that the dynamic information structure has a smaller region of quantum advantage compared to static. Also the central horizontal lines in both figures $p_2=0.5$ and $p_c=0.5$ are identical. This is because a maximally noisy channel $(p_c=0.5)$ is as good as there being no channel (static).

\section{Other Solution Concepts}
\label{sec:othersol}

We have seen that for both static and dynamic information structures, when we optimize over the space of general quantum strategies, the optimal quantum strategy always turns out to be either projective or classical in nature. We now analyse two further solution concepts, formally defined below, concerning these classes for the static and dynamic information structures. \rev{Throughout this section $\mathsf C$ is the mixed classical class of Section~\ref{sec:class}, embedded in $\mathsf Q$ via Proposition~\ref{prop:embedding} as $E_i=\ei I$, $F_j=\fj I$ with $\ei,\fj\in[0,1]$ and $\evi=\fvj=0$; deterministic strategies are the case $\ei,\fj\in\{0,1\}$, i.e.\ $E_i,F_j\in\{O,I\}$, $j=1,2$. }

\revii{In words, a strategy pair is \emph{person-by-person optimal} if neither player can lower the team's cost by changing their own strategy alone, the other player's strategy being held fixed. It is the team analogue of a Nash equilibrium, the difference being that here both players share a single objective, and it is the natural solution concept when the players cannot communicate before play in order to coordinate on a jointly optimal pair. Team optimality is strictly stronger: a team optimal pair whose components lie in a common class is always person-by-person optimal, but the converse is not true.}


\begin{definition}
For a strategy class $\mathsf{A}$, the property of best response over $\mathsf{A}$ holds for the static information structure, if for all $p_1,p_2$  the following holds:
\begin{itemize}
\item For any $\bie \in \mathsf{A}, \exists \bif^* \in \mathsf{A}$ such that 
\begin{equation*}
J_\mathsf{Q}^\mathsf{S}(\bie,\bif^*;p_1,p_2) \leq J^\mathsf{S}_\mathsf{Q}(\bie,\bif;p_1,p_2) \quad \forall  \bif \in \mathsf{Q},
\end{equation*}
\item For any $\bif \in \mathsf{A}, \exists \bie^* \in \mathsf{A}$ such that
\begin{equation*}
J_\mathsf{Q}^\mathsf{S}(\bie^*,\bif;p_1,p_2) \leq J^\mathsf{S}_\mathsf{Q}(\bie,\bif;p_1,p_2) \quad \forall  \bie \in \mathsf{Q},
\end{equation*}
\end{itemize}
That is, for each player $i$, and each strategy belonging to $\mathsf{A}$ of this player, there is an optimal response of the other player over strategies in $\mathsf{Q}$, also belonging to $\mathsf{A}$. We say $\mathsf{A}$ has the property of best response for the dynamic information structure if the above holds with the replacement $\mathsf{S} \mapsto \mathsf{D}$ and $(p_1,p_2) \mapsto (p_s,p_c)$.
\end{definition}
\begin{definition}
	For some strategy class $\mathsf{A}$, the property of person-by-person optimality over $\mathsf{A}$ holds for the static information structure, if for all $p_1,p_2$ there exists $ \bie^*, \bif^* \in \mathsf{A} $ satisfying
	$$ J^\mathsf{S}_\mathsf{Q}(\bie^*, \bif^*;p_1,p_2) \leq J^\mathsf{S}_\mathsf{Q}(\bie^*, \bif;p_1,p_2) \quad \forall  \bif  \in \rev{\mathsf{Q}} $$
	$$ J^\mathsf{S}_\mathsf{Q}(\bie^*,\bif^*;p_1,p_2) \leq J_\mathsf{Q}^\mathsf{S}(\bie, \bif^*;p_1,p_2) \quad \forall  \bie  \in \rev{\mathsf{Q}}. $$
	That is, there exists a pair of strategies, each belonging to $\mathsf{A}$, which are a best response to each other. We say $\mathsf{A}$ has the property of person-by-person optimality for the dynamic information structure if the above holds with the replacement $\mathsf{S} \mapsto \mathsf{D}$ and $(p_1,p_2) \mapsto (p_s,p_c)$.

\end{definition}
For every team optimal solution, given that the individual strategies all happen to belong to the same strategy class, the corresponding solution will also be person-by-person optimal. But the converse need not be true.
\subsection{Best Response over a Strategy Class}
We say that a game has the property of best response over a strategy class if over all probability distributions, for any strategy of Player 1 from that strategy class, the best response or optimal strategy for Player 2 lies in the same strategy class and vice-versa.

\begin{theorem}
The property of best response over classical strategies $\mathsf{C}$
    \begin{enumerate}
        \item holds for the static information structure;
        \item holds for the dynamic information structure.
    \end{enumerate}
    \label{thm:best_class}
\end{theorem}

(Proof in \ref{sec:4.1}.)

We see that the property of best response over classical strategies holds for both the information structures. However, since we found differences in the team-optimal solution over projective strategies, it is interesting to see if the information structure affects this property.\\

\begin{theorem}
The property of best response over projective strategies
    \begin{enumerate}
        \item holds for the static information structure;
        \item does not hold for the dynamic information structure.
    \end{enumerate}
    \label{thm:best_proj}
\end{theorem}

(Proof in \ref{sec:4.2}.)

\rev{For the static information structure, every projective strategy admits a projective best response. For the dynamic information structure this fails: for some $(p_s,p_c)$ there are projective strategies of which no best response is projective, even though projective pairs that are best responses to each other exist for every $(p_s,p_c)$ (Theorem~\ref{thm:pbp_proj}).}

\subsection{Person-by-Person Optimality over a Strategy Class}

We say that a game has the property of person-by-person optimality over a strategy class if over all probability distributions, there exists a pair of strategies of \A{} and \B{} in that strategy class, which are best responses to each other.
\begin{theorem}
The property of person-by-person optimality over classical strategies
    \begin{enumerate}
        \item holds for the static information structure;
        \item holds for the dynamic information structure.
    \end{enumerate}
    \label{thm:pbp_class}
\end{theorem}

(Proof in Appendix \ref{sec:4.3}.)

For classical strategies, like the best response property, there is an agreement with the static and dynamic information structure results.
\begin{theorem}
The property of person-by-person optimality over projective strategies
    \begin{enumerate}
        \item holds for the static information structure;
        \item \rev{holds for the dynamic information structure.}
        \revii{In particular, for every $(p_s,p_c)$ the projective pair $\bie^* = (P,\, I-P)$ and $\bif^* = (P,\, P)$ if $p_s \geq \half$ (resp.\ $\bif^* = (I-P,\, P)$ if $p_s \leq \half$), where $P$ is an arbitrary rank-one projector, is person-by-person optimal, with expected cost $\half\min\{p_s,\pst\}$, independent of $p_c$.}
    \end{enumerate}
    \label{thm:pbp_proj}
\end{theorem}

(Proof in Appendix \ref{sec:4.4}.)

\begin{theorem}
\begin{enumerate}\rev{
    \item Let $p_1,p_2 \in [0,1]$ and consider the static information structure. Let $(\bm E^{S*}_\perp,\bm F^{S*}_\perp)=(\bm E^{S*}_\perp(p_1,p_2),\bm F^{S*}_\perp(p_1,p_2))$ be the optimal strategies over projective strategies, as defined below \eqref{jipo}. Then  $(\bm E^{S*}_\perp,\bm F^{S*}_\perp)$ are also team optimal over all strategy classes and hence are person-by-person optimal.
    \item Let $p_s,p_c \in [0,1]$ and consider the dynamic information structure. Let $(\bm E^{D*}_\perp,\bm F^{D*}_\perp)=(\bm E^{D*}_\perp(p_s,p_c),\bm F^{D*}_\perp(p_s,p_c))$ be the optimal strategies over projective strategies, as defined below \eqref{jipo}. $(\bm E^{D*}_\perp,\bm F^{D*}_\perp)$ are person-by-person optimal if and only if
\begin{equation}
    (2p_c-1)^2 \;\leq\; p_s^2 + \pst^{\,2}. \label{eq:pbp_threshold}
\end{equation}
}
\end{enumerate}

\label{thm:pbp_unstable}
\end{theorem}
\rev{(Proof in Appendix \ref{sec:4.5}.)}

\rev{Therefore, beyond the threshold \eqref{eq:pbp_threshold}, within the projective class, team optimality and person-by-person optimality become incompatible: every projective team optimal pair admits a profitable unilateral (non-projective) deviation for the dynamic information structure.}

\revii{Condition \eqref{eq:pbp_threshold} compares two biases. Its left side is the squared bias of the channel, $(2p_c-1)^2$, which is $0$ when the channel is maximally noisy and $1$ when it is noiseless. Its right side is $p_s^2+\pst^2=\half\big(1+(2p_s-1)^2\big)$, an increasing function of the bias of the source, equal to $\half$ at $p_s=\half$ and to $1$ at $p_s\in\{0,1\}$. Person-by-person optimality of the team optimal projective pair therefore fails precisely when the channel is close to noiseless while the source is close to unbiased. At $p_s=\half$ it fails for $|2p_c-1|>1/\sqrt2$, that is for $p_c\notin[\,\half-\tfrac{1}{2\sqrt2},\ \half+\tfrac{1}{2\sqrt2}\,]\approx[0.146,0.854]$; for $p_s\in\{0,1\}$ it never fails.}

\revii{The mechanism is geometric, and the same as in Remark~\ref{rem:disadv}. Write $r:=\half\sqrt{p_s^2+\pst^2}$ for the length of the combined Bloch vector $p_s\eva^\perp\pm\pst\evb^\perp$ that \A{}'s optimal projective pair presents to \B{}, and $\alpha_1:=p_c-\half$ for the scalar weight the channel places on \B{}'s second effect. As the proof in Appendix \ref{sec:4.5} shows, \B{} has three candidate responses: the projective one, worth $\tfrac{\alpha_1-r}{2}$, and the two classical ones $F_1=O$ and $F_1=I$, worth $0$ and $\alpha_1$. A rank-one projection pins $|\fvb|=\half$ and hence pins \B{}'s marginal to be uniform: it buys correlation at the price of bias. Once $|\alpha_1|>r$ the bias is worth more than the correlation and \B{} strictly prefers a \emph{classical} effect. At equality the two tie: the team optimal projective pair remains person-by-person optimal, but \B{} also has a classical best response attaining the same cost.}

\begin{center}
\begin{table}
\begin{tabular}{ |c|c| }
 \hline
  \textbf{Static} & \textbf{Dynamic} \\ \hline
  Classical team optimal cost $\Jsco$ & Classical team optimal cost $\Jdco$\\
  $= \min \{ p_1p_2,\pat p_2,p_1\pbt,\pat\pbt\}$ &  $= \min\{ p_sp_c,\pst p_c, p_s\pct,\pst\pct\}$\\
 \hline
 Projective strategies are team optimal  & Projective strategies are team optimal \\
 $\forall (p_1, p_2)$.& only when $\Jdpop \leq \Jdcop$.\\\hline
 $\Jspop=0.5-\frac{1}{\sqrt{2}}\sqrt{\{p_1^2+\pat^2\}\{p_2^2+\pbt^2\}}$ & $\Jdpop= \Jspo (p_1=p_s,p_2=0.5).$\\
 \rev{if $|(2p_1-1)(2p_2-1)|\le 4p_1\pat\,p_2\pbt$;} & \\
 \rev{$\Jspo=\Jsco$ otherwise.} & \\
    \hline
 There is \textit{no} $(p_1, p_2)$  where there is & There exists $(p_s, p_c)$ where there is \\
 a \rev{projective} disadvantage. &  a \rev{projective} disadvantage.\\ \hline
 \rev{Projective strategies are person-by-person} & \rev{Projective strategies are person-by-person optimal $\forall (p_s,p_c)$,}\\
 \rev{optimal, and team optimal projective} & \rev{but if $(2p_c-1)^2 > p_s^2 + \pst^2$, team optimal projective }\\
 \rev{pairs are person-by-person optimal.} & \rev{pairs are \textit{not} person-by-person optimal}
 \\\hline
 \rev{Every projective strategy admits a} & \rev{There exist $(p_s,p_c)$ and a projective strategy of one}\\
 \rev{projective best response.} & \rev{player to which \textit{no} best response is projective.} \\\hline
 Classical strategies are person-by-person optimal. &  Classical strategies are person-by-person optimal. \\\hline
 The best response to any classical & The best response to any classical\\
 strategy is also classical. &  strategy is also classical. \\\hline
\end{tabular}

\caption{Summary of properties of the general CHSH game over static and dynamic information structure.}
\label{table:summary}
\end{table}
\end{center}

\revii{\section{Discussion}
\label{sec:discussion}

We once again emphasize that the dynamic problem is a sequential team problem with
signalling, not a Bell test: once \A{}'s action reaches \B{} the players are not space-like
separated, their correlations are unconstrained by any Bell inequality, and nothing we prove about
it concerns nonlocality. And since $\mathsf C\subseteq\mathsf Q$
(Proposition~\ref{prop:embedding}), quantum strategies are never worse than classical ones in
either structure. Remark~\ref{rem:disadv} therefore exhibits a disadvantage of the \emph{projective
class}, not of quantum strategies: a finitely parameterised subclass ceases to be adequate when the
information structure changes.

In the dynamic information structure, beyond the threshold \eqref{eq:pbp_threshold} a team optimal projective pair
is not person-by-person optimal (Theorem~\ref{thm:pbp_unstable}), so players confined to the
projective class and unable to coordinate have an incentive to deviate if the parameters violate \eqref{eq:pbp_threshold}. Person-by-person optimal
projective pairs still exist for every $(p_s,p_c)$, at cost $\half\min\{p_s,\pst\}$
(Theorem~\ref{thm:pbp_proj}). Since $\min\{p_c,\pct\}\le\half$, this never falls below the classical
team optimum $\min\{p_s,\pst\}\min\{p_c,\pct\}$ of Theorem~\ref{thm:jdco}, and exceeds it unless
$p_c=\half$ or $p_s\in\{0,1\}$. Stable projective play is thus always available and, away from the
maximally noisy channel, strictly worse than classical team play.

Finally, we recall that the state is not optimised here.
Entanglement is a resource distributed before the parameters $(p_s,p_c)$ are revealed: the
measurement strategies may depend on those parameters, the state cannot. A state chosen with
knowledge of $(p_s,p_c)$ would also require the players to agree on that choice, which the absence
of pre-play communication forbids. Our results therefore describe a fixed resource being put to use
once the parameters are known, and say which measurement classes remain adequate in that setting.
Treating the state as a strategic choice is a line for future research.

We limited ourselves to the setting of two players. The behaviour of chains of more than
two players, where a hierarchy of information structures interpolates between the static and
dynamic extremes is open. It would be interesting to ascertain whether the adequacy of projective strategies fails for non-classical
information structures generally or only for those with a dual effect of the kind studied here.
}

\section{Conclusion}

In this paper, we explored the CHSH game through the lens of information structures in stochastic teams. Our main findings, summarized in Table \ref{table:summary}, highlight how the choice between a static and dynamic information structure significantly alters the landscape of optimality.
In the static setting, projective strategies are team optimal for all probability distributions. In contrast, for the dynamic information structure, the region of quantum advantage is notably smaller because the optimal cost achieved by projective strategies is independent of the channel noise. Most importantly, we found that there exist certain probability distributions in the dynamic game where projective strategies perform worse than classical ones, creating a region of strict \rev{projective} disadvantage.

\rev{Furthermore, we investigated the properties of person-by-person optimality and best response over strategy classes. While projective strategies enjoy these properties in the static game, in the dynamic game best responses to projective strategies need not be projective, and beyond a sharp threshold on the channel noise, person-by-person optimality parts ways with team optimality: team optimal projective pairs are no longer person-by-person optimal, and there are person-by-person optimal projective pairs that are strictly outperformed even by classical team play.} Ultimately, our results shed light on the subtle and fascinating role played by the information structure, demonstrating the fragility of \rev{well-known structural properties of projective strategies} when the flow of information between players changes.

\section*{Declarations}
\begin{itemize}
    \item Consent to Participate declaration: not applicable
    \item Consent to Publish declaration: not applicable
    \item Ethics declaration: not applicable
\end{itemize}

\section*{Acknowledgments}
The authors would like to thank Prof Himadri Dhar of IIT Bombay for his inputs, and three anonymous reviewers for their comments on an earlier version of this paper. The authors also declare that they have not received any funding for this work.
\bibliography{ref.bib}

@article{brask2017BellCommunication,
    author  = {Brask, Jonatan Bohr and Chaves, Rafael},
    title   = {Bell Scenarios with Communication},
    journal = {Journal of Physics A: Mathematical and Theoretical},
    volume  = {50},
    number  = {9},
    pages   = {094001},
    year    = {2017},
    doi     = {10.1088/1751-8121/aa5840}
}

@article{wolfe2021QuantumInflation,
    author  = {Wolfe, Elie and Pozas-Kerstjens, Alejandro and
    Grinberg, Matan and Rosset, Denis and Ac{\'i}n, Antonio and
    Navascu{\'e}s, Miguel},
    title   = {Quantum Inflation: A General Approach to Quantum Causal Compatibility},
    journal = {Physical Review X},
    volume  = {11},
    number  = {2},
    pages   = {021043},
    year    = {2021},
    doi     = {10.1103/PhysRevX.11.021043}
}

@article{ligthart2023ConvergentInflation,
    author  = {Ligthart, Laurens T. and Gachechiladze, Mariami and Gross, David},
    title   = {A Convergent Inflation Hierarchy for Quantum Causal Structures},
    journal = {Communications in Mathematical Physics},
    volume  = {401},
    pages   = {2673--2714},
    year    = {2023},
    doi     = {10.1007/s00220-023-04697-7}
}

@article{tavakoli2022BellNetworks,
    author  = {Tavakoli, Armin and Pozas-Kerstjens, Alejandro and
    Luo, Ming-Xing and Renou, Marc-Olivier},
    title   = {Bell Nonlocality in Networks},
    journal = {Reports on Progress in Physics},
    volume  = {85},
    number  = {5},
    pages   = {056001},
    year    = {2022},
    doi     = {10.1088/1361-6633/ac41bb}
}

@article{pozasKerstjens2022FullNetwork,
    author  = {Pozas-Kerstjens, Alejandro and Gisin, Nicolas and Tavakoli, Armin},
    title   = {Full Network Nonlocality},
    journal = {Physical Review Letters},
    volume  = {128},
    number  = {1},
    pages   = {010403},
    year    = {2022},
    doi     = {10.1103/PhysRevLett.128.010403}
}

@article{moreno2020MultistageGames,
    author  = {Moreno, George and Nery, Ranieri and
    Palhares, Alberto and Chaves, Rafael},
    title   = {Multistage Games and Bell Scenarios with Communication},
    journal = {Physical Review A},
    volume  = {102},
    number  = {4},
    pages   = {042412},
    year    = {2020},
    doi     = {10.1103/PhysRevA.102.042412}
}

@inproceedings{papadimitriou85intractable,
	title = {Intractable problems in control theory},
	volume = {24},
	doi = {10.1109/CDC.1985.268670},
	booktitle = {1985 24th {IEEE} Conference on Decision and Control},
	author = {Papadimitriou, Christos H. and Tsitsiklis, John N.},
	month = dec,
	year = {1985},
	pages = {1099 --1103},
}

@article{ kulkarni2014optimizer,
	title = "An Optimizer's Approach to Stochastic Control Problems with Nonclassical Information Structures",
	volume="60",
	number="4",
	pages = "937--949",
	journal = "IEEE Transactions on Automatic Control",
	author = "A. A. Kulkarni and T. P. Coleman",
	year = "2015",
}

@Article{ nash50equilibrium,
	author = "J. F. Nash",
	title = "Equilibrium points in ${N}$-Person Games",
	journal = "Proceedings of National Academy of Science",
	year = "1950"
}

@Article{ borkar88convex,
	title = "A convex analytic approach to Markov decision processes",
	volume = "78",
	issn = "0178-8051, 1432-2064",
	doi = "10.1007/BF00353877",
	journal = "Probability Theory and Related Fields",
	author = "Vivek S. Borkar",
	month = aug,
	year = "1988",
	pages = "583--602"
}

@Article{ witsenhausen_counterexample_1968,
	title = "A Counterexample in Stochastic Optimum Control",
	volume = "6",
	issn = "00361402",
	doi = "10.1137/0306011",
	journal = "{SIAM} Journal on Control",
	author = "H. S. Witsenhausen",
	year = "1968",
	pages = "131"
}

@article{marschak1955elements,
  title={Elements for a Theory of Teams},
  author={Marschak, Jakob},
  journal={Management science},
  volume={1},
  number={2},
  pages={127--137},
  year={1955},
  publisher={INFORMS}
}

@article{radner1962team,
  title={Team decision problems},
  author={Radner, Roy},
  journal={The Annals of Mathematical Statistics},
  pages={857--881},
  year={1962},
  publisher={JSTOR}
}

@article{neilsen2004qcqi,
 title={Quantum Computation and Quantum Information.},
 author={Nielsen Michael A., Chuang Isaac L. },
 journal={Cambridge University Press},
 year={2004},
 publisher={Cambridge University}
}

@article{clauser1969chsh,
  title = {Proposed Experiment to Test Local Hidden-Variable Theories},
  author = {Clauser, John F. and Horne, Michael A. and Shimony, Abner and Holt, Richard A.},
  journal = {Phys. Rev. Lett.},
  volume = {23},
  issue = {15},
  pages = {880--884},
  numpages = {0},
  year = {1969},
  month = {Oct},
  publisher = {American Physical Society},
  doi = {10.1103/PhysRevLett.23.880},
  url = {https://link.aps.org/doi/10.1103/PhysRevLett.23.880}
}

@article{bell1964epr,
  title = {On the Einstein Podolsky Rosen paradox},
  author = {Bell, J. S.},
  journal = {Physics Physique Fizika},
  volume = {1},
  issue = {3},
  pages = {195--200},
  numpages = {6},
  year = {1964},
  month = {Nov},
  publisher = {American Physical Society},
  doi = {10.1103/PhysicsPhysiqueFizika.1.195},
  url = {https://link.aps.org/doi/10.1103/PhysicsPhysiqueFizika.1.195}
}

@article{brunner2014bell,
  title = {Bell nonlocality},
  author = {Brunner, Nicolas and Cavalcanti, Daniel and Pironio, Stefano and Scarani, Valerio and Wehner, Stephanie},
  journal = {Rev. Mod. Phys.},
  volume = {86},
  issue = {2},
  pages = {419--478},
  numpages = {60},
  year = {2014},
  month = {Apr},
  publisher = {American Physical Society},
  doi = {10.1103/RevModPhys.86.419},
  url = {https://link.aps.org/doi/10.1103/RevModPhys.86.419}
}

@article{eisert1999quantumgames,
	title = {Quantum Games and Quantum Strategies},
	author = {Eisert, Jens and Wilkens, Martin and Lewenstein, Maciej},
	journal = {Phys. Rev. Lett.},
	volume = {83},
	issue = {15},
	pages = {3077--3080},
	numpages = {0},
	year = {1999},
	month = {Oct},
	publisher = {American Physical Society},
	doi = {10.1103/PhysRevLett.83.3077},
	url = {https://link.aps.org/doi/10.1103/PhysRevLett.83.3077}
}

@article{aumann74subjectivity,
	title = {Subjectivity and correlation in randomized strategies},
	journal = {Journal of Mathematical Economics},
	volume = {1},
	number = {1},
	pages = {67-96},
	year = {1974},
	issn = {0304-4068},
	doi = {https://doi.org/10.1016/0304-4068(74)90037-8},
	url = {https://www.sciencedirect.com/science/article/pii/0304406874900378},
	author = {Robert J. Aumann}
}

@article{tsirelson1980genbell,
	title={Quantum generalizations of Bell's inequality},
	author={B. S. Cirel'son},
	journal={Letters in Mathematical Physics},
	year={1980},
	volume={4},
	pages={93-100},
	url={https://api.semanticscholar.org/CorpusID:120680226}
}
\appendix
\section{Appendix}
\label{sec:appendix}
\subsection{Proofs of Theorems}
\label{sec:proofs}

\subsubsection{Theorem \ref{thm:jsco}}
\label{sec:3.1}
\begin{proof}
\rev{
A general mixed classical strategy is a pair of kernels $\Gamma_1,\Gamma_2:\{0,1\}\to[0,1]$ with
$\bkP(u_i = 0\mid y_i) = \Gamma_i(y_i)$, so $\Jsc$ is a function of the four scalars
$\Gamma_1(0),\Gamma_1(1),\Gamma_2(0),\Gamma_2(1)\in[0,1]$. We first express $\Jsc$ explicitly in terms
of these kernels and read off that its minimum over the full mixed class $\mathsf C$ is attained on
the deterministic corners $\mathsf C_{\det}$.

Since $u_1,u_2$ are conditionally independent given $y_1,y_2$, and $y_1 \perp\!\!\perp y_2$,
\begin{align}
\Jsc(\Gamma_1,\Gamma_2;p_1,p_2)
&= p_1p_2\big[\Gamma_1(0)\big(1-\Gamma_2(0)\big)+\big(1-\Gamma_1(0)\big)\Gamma_2(0)\big] \notag\\
&\quad + p_1\pbt\big[\Gamma_1(0)\big(1-\Gamma_2(1)\big)+\big(1-\Gamma_1(0)\big)\Gamma_2(1)\big] \notag\\
&\quad + \pat p_2\big[\Gamma_1(1)\big(1-\Gamma_2(0)\big)+\big(1-\Gamma_1(1)\big)\Gamma_2(0)\big] \notag\\
&\quad + \pat\pbt\big[\Gamma_1(1)\Gamma_2(1)+\big(1-\Gamma_1(1)\big)\big(1-\Gamma_2(1)\big)\big].
\label{eq:jsc_mixed}
\end{align}
Every monomial in \eqref{eq:jsc_mixed} contains at most one of $\Gamma_1(0),\Gamma_1(1)$ and at most
one of $\Gamma_2(0),\Gamma_2(1)$, so $\Jsc$ is multilinear on $[0,1]^4$, hence separately affine in
each kernel value with the other three fixed.

A function separately affine in each coordinate of a hypercube attains its minimum at a vertex, and the vertices
of $[0,1]^4$ are exactly the deterministic pairs $(\gamma_1,\gamma_2)\in\mathsf C_{\det}$, so
\[
\Jscop = \min_{\Gamma_1,\Gamma_2\in[0,1]} \Jsc(\Gamma_1,\Gamma_2;p_1,p_2)
= \min_{\gamma_1,\gamma_2\,\in\,\{0,1\}\to\{0,1\}} \Jsc(\gamma_1,\gamma_2;p_1,p_2),
\]
and it suffices to minimize over deterministic $\gamma_1,\gamma_2$.
}

To lighten notation, let
    $
    a=\gamma_1(0),b=\gamma_1(1), c=\gamma_2(0),d=\gamma_2(1).
    $
    Then
    \begin{align}
        \Jsc
        &= p_1p_2 (a \oplus c)+ \pat p_2 (b \oplus c)+ p_1\pbt(a \oplus d)+ \pat\pbt\,\overline{(b \oplus d)}.
    \end{align}
    Thus \(\Jsc\) is a weighted sum of four nonnegative binary quantities.
        We first show that these four quantities cannot all be zero simultaneously.
    If that were the case, we would have
    \[
    a \oplus c = b \oplus c = a \oplus d = \overline{(b \oplus d)}=0,
    \]
    which implies \(a=b=c=d\), contradicting \(\overline{(b \oplus d)}=0\).
    Hence at least one of the four terms in \(\Jsc\) must be equal to \(1\), and therefore
    \[
    \Jsc \ge \min\{p_1p_2,\pat p_2,p_1\pbt,\pat\pbt\}.
    \]
    It is easy to see that the RHS above is attainable. For example, choosing \((a,b,c,d)=(1,0,0,1)\) gives \(\Jsc=p_1p_2\).
    Therefore
$    \Jscop=\min\{p_1p_2,\pat p_2,p_1\pbt,\pat\pbt\}.$

\end{proof}

\subsubsection{Theorem \ref{thm:jdco}}
\label{sec:3.2}
\begin{proof}
    \rev{As in Theorem \ref{thm:jsco}, we first express $\Jdc$ explicitly in terms of the kernels
$\Gamma_1,\Gamma_2$ and read off that its minimum over the full mixed class $\mathsf C$ is attained
on $\mathsf C_{\det}$, this time accounting for the dual effect.

Conditioning on $(y_1,v)$ and using $y_2=y_1\oplus u_1\oplus v$,
\begin{align}
\Jdc(\Gamma_1,\Gamma_2;p_s,p_c)
&= p_sp_c\big[\Gamma_1(0)\big(1-\Gamma_2(0)\big)+\big(1-\Gamma_1(0)\big)\Gamma_2(1)\big] \notag\\
&\quad + p_s\pct\big[\Gamma_1(0)\big(1-\Gamma_2(1)\big)+\big(1-\Gamma_1(0)\big)\Gamma_2(0)\big] \notag\\
&\quad + \pst p_c\big[\Gamma_1(1)\Gamma_2(1)+\big(1-\Gamma_1(1)\big)\Gamma_2(0)\big] \notag\\
&\quad + \pst\pct\big[\Gamma_1(1)\big(1-\Gamma_2(0)\big)+\big(1-\Gamma_1(1)\big)\big(1-\Gamma_2(1)\big)\big].
\label{eq:jdc_mixed}
\end{align}
The dual effect is manifest: $\Gamma_1(0)$ multiplies both $\Gamma_2(0)$ and $\Gamma_2(1)$, since
\A{}'s action selects which of \B{}'s two kernel values is queried. Nevertheless every monomial in
\eqref{eq:jdc_mixed} contains at most one of $\Gamma_1(0),\Gamma_1(1)$ and at most one of
$\Gamma_2(0),\Gamma_2(1)$, so $\Jdc$ is multilinear on $[0,1]^4$, hence separately affine in each
kernel value with the other three fixed. As in Theorem \ref{thm:jsco}, its minimum is therefore
attained at a vertex of $[0,1]^4$, i.e.\ at a deterministic pair
$(\gamma_1,\gamma_2)\in\mathsf C_{\det}$:
\[
\Jdcop = \min_{\Gamma_1,\Gamma_2\in[0,1]} \Jdc(\Gamma_1,\Gamma_2;p_s,p_c)
= \min_{\gamma_1,\gamma_2\,\in\,\{0,1\}\to\{0,1\}} \Jdc(\gamma_1,\gamma_2;p_s,p_c),
\]
so it suffices to minimize over deterministic $\gamma_1,\gamma_2$.}

Let   $a:=\gamma_1(0),b:=\gamma_1(1),c:=\gamma_2(0), d:=\gamma_2(1).
    $
    Since in the dynamic information structure
    $
    y_2=y_1\oplus u_1\oplus v,
    $
    where \(v\sim \mathrm{Ber}(\pct)\), and under a deterministic classical strategy
    $
    u_1=\gamma_1(y_1) \text{and} u_2=\gamma_2(y_2),
    $
    the expected cost is
    \[
    \Jdcp
    =
    p_sp_c\,(a\oplus \gamma_2(a))
    +
    p_s\pct\,(a\oplus \gamma_2(\bar a))
    +
    \pst p_c\,(b\oplus \gamma_2(1\oplus b)\oplus (1\oplus b))
    +
    \pst\pct\,\gamma_2(b).
    \]

    We now evaluate this expression for the four possible choices of
    \((c,d)=(\gamma_2(0),\gamma_2(1))\). If
    \begin{itemize}
        \item   \((c,d)=(0,0)\), then
    $
    \Jdcp=p_s a+\pst p_c,
    $ so the minimum over \(a,b\in\{0,1\}\) is \(\pst p_c\).
       \item  \((c,d)=(0,1)\), then
    $
    \Jdcp=p_s\pct+\pst b,
    $
    so the minimum over \(b\in\{0,1\}\) is \(p_s\pct\).
    \item \((c,d)=(1,0)\), then
    $
    \Jdcp=p_sp_c+\pst \bar b,
    $
    so the minimum over \(b\in\{0,1\}\) is \(p_sp_c\).
    \item \((c,d)=(1,1)\), then
    $
    \Jdcp=p_s\bar a+\pst\pct,
    $
    so the minimum over \(a,b\in\{0,1\}\) is \(\pst\pct\).
    \end{itemize}
    Therefore $\Jdcop=\min\{p_sp_c,\;p_s\pct,\;\pst p_c,\;\pst\pct\}.$

\end{proof}

\subsubsection{Lemma \ref{lemma}}
\label{sec:3.3}
\begin{proof}
    \begin{enumerate}
        \item \textit{Static information structure.}
        For fixed \(y_1,y_2\), under the quantum strategy,
        \[
        \bkP(u_1,u_2\mid y_1,y_2)=\Tr\!\big[(E^{u_1}_{y_1}\otimes F^{u_2}_{y_2})\rho\big].
        \]
        Since \(A_{y_1}=E_{y_1}^0-E_{y_1}^1\) and \(B_{y_2}=F_{y_2}^0-F_{y_2}^1\), we have
        \begin{equation}
            \langle A_{y_1}B_{y_2}\rangle
            =\sum_{u_1,u_2}(-1)^{u_1\oplus u_2}\bkP(u_1,u_2\mid y_1,y_2).
            \label{eq:AB_expand}
        \end{equation}
        Now define
        $
        \Be' := 1-2\Be = (-1)^{\Be}.
        $
        Then
        \begin{align}
            \Ebb[\Be']
            &=
            \sum_{u_1,u_2,y_1,y_2}
            (-1)^{u_1\oplus u_2\oplus y_1y_2}\,
            \bkP(u_1,u_2,y_1,y_2) \notag\\
            &=
            \sum_{y_1,y_2}
            (-1)^{y_1y_2}\bkP(y_1,y_2)
            \sum_{u_1,u_2}
            (-1)^{u_1\oplus u_2}\bkP(u_1,u_2\mid y_1,y_2) \notag\\
            &=
            \sum_{y_1,y_2}
            (-1)^{y_1y_2}\bkP(y_1,y_2)\langle A_{y_1}B_{y_2}\rangle,
            \label{eq:static_reward}
        \end{align}
        where \eqref{eq:AB_expand} was used in the last step.
        Since \(y_1\) and \(y_2\) are independent,
        \[
        \bkP(0,0)=p_1p_2,\quad
        \bkP(0,1)=p_1\pbt,\quad
        \bkP(1,0)=\pat p_2,\quad
        \bkP(1,1)=\pat\pbt.
        \]
        Therefore,
        $
        \Ebb[\Be']
        =
        p_1p_2\langle A_0B_0\rangle
        +p_1\pbt\langle A_0B_1\rangle
        +\pat p_2\langle A_1B_0\rangle
        -\pat\pbt\langle A_1B_1\rangle.
        $
        Finally, using
        $
        \Ebb[\Be]=\frac{1-\Ebb[\Be']}{2},
        $
        we obtain \eqref{Jsq}.

\item For the dynamic information structure, define
$
\Be' := 1-2\Be = (-1)^{\Be}.
$
Then
\begin{align*}
    \Ebb[\Be']
    &= \sum_{u_1,u_2,y_1,y_2} (-1)^{u_1 \oplus u_2 \oplus y_1 y_2}
    \Tr(\{E^{u_1}_{y_1}\otimes F^{u_2}_{y_2}\}\rho)\bkP(y_2|y_1,u_1)\bkP(y_1).
\end{align*}
Now
$
\bkP(y_2|y_1,u_1)
= p_c \, \overline{y_2\oplus y_1\oplus u_1}
+ \pct \, (y_2\oplus y_1\oplus u_1),
$
so for \(y_1=0\), \(y_2=u_1\) with probability \(p_c\) and \(y_2=\bar u_1\) with probability \(\pct\), while for \(y_1=1\), \(y_2=\bar u_1\) with probability \(p_c\) and \(y_2=u_1\) with probability \(\pct\). Hence
\begin{align*}
    \Ebb[\Be']
    &= p_s \sum_{u_1,u_2} (-1)^{u_1\oplus u_2}
    \Big(
    p_c \Tr(\{E^{u_1}_0\otimes F^{u_2}_{u_1}\}\rho)
    + \pct \Tr(\{E^{u_1}_0\otimes F^{u_2}_{\bar u_1}\}\rho)
    \Big) \\
    &\quad + \pst \sum_{u_1,u_2}
    \Big(
    (-1)^{u_2\oplus 1} p_c \Tr(\{E^{u_1}_1\otimes F^{u_2}_{\bar u_1}\}\rho)
    + (-1)^{u_2} \pct \Tr(\{E^{u_1}_1\otimes F^{u_2}_{u_1}\}\rho)
    \Big),
\end{align*}
where we used
$
u_1\oplus u_2\oplus \bar u_1 = u_2\oplus 1,
\qquad
u_1\oplus u_2\oplus u_1 = u_2.
$
Expanding the four sums over \(u_1,u_2\), using \(E^0_i=E_i\), \(E^1_i=I-E_i\), \(F^0_j=F_j\), \(F^1_j=I-F_j\), and \(\Tr \rho =1\), we obtain
\begin{align*}
    \Ebb[\Be']
    &= 2 \sum_{i,j\in\{0,1\}} (-1)^{ij}\bkP(y_1=i)\Tr(\{E_i\otimes F_j\}\rho)
    -2p_s \Tr(\{E_0\otimes I\}\rho) \\
    &\quad -2(p_s\pct+\pst p_c)\Tr(\{I\otimes F_0\}\rho)
    -2(p_sp_c-\pst\pct)\Tr(\{I\otimes F_1\}\rho)
    +1-2\pst\pct.
\end{align*}
Since \(\Be' = 1-2\Be\), we have
$
\Jdqp = \Ebb[\Be] = \frac{1-\Ebb[\Be']}{2},
$
and therefore
\begin{align*}
    \Jdqp
    &= - \sum_{i,j\in\{0,1\}} (-1)^{ij}\bkP(y_1=i)\Tr(\{E_i\otimes F_j\}\rho)
    + p_s \Tr(\{E_0\otimes I\}\rho) + \pst\pct \\
    &\quad + (p_s\pct+\pst p_c)\Tr(\{I\otimes F_0\}\rho)
    + (p_sp_c-\pst\pct)\Tr(\{I\otimes F_1\}\rho).
\end{align*}
Finally, using
$
\Tr(\{E_i\otimes I\}\rho)=\half \Tr E_i
\text{and}
\Tr(\{I\otimes F_j\}\rho)=\half \Tr F_j,
$
we get \eqref{jdq}.
\end{enumerate}
\end{proof}



\subsubsection{Theorem \ref{thm:best_class}}
\label{sec:4.1}
\begin{proof}
    \begin{enumerate}
        \item For the static information structure, since the problem is symmetric, it is enough to
        prove the following: if \A{} plays a classical strategy, then an optimal strategy for \B{} is
        also classical. \rev{Let \A{} play an arbitrary mixed classical strategy $\bie$ with $E_i=\ei I$, $\ei\in[0,1]$, $\evi=0$ (the deterministic strategies $\{O,I\}$ are the case $\ei\in\{0,1\}$), and let \B{} play an arbitrary quantum strategy $\bif$. Since $\evi=0$, \eqref{eq:trace_ef} gives $\Tr[(E_i\otimes F_j)\rho]=\ei\fj$ for every POVM $F_j$, so by \eqref{aibj},
        \[
            \langle A_iB_j\rangle = (1-2\ei)(1-2\fj),
        \]
        which depends on \B{} only through $\fa,\fb\in[0,1]$, not on $\fvj$.}
        Substituting these into \eqref{Jsq}, we get for POVMs $(\bie,\bif)$,
        \begin{align*}
            \Jsqp
            &= 0.5 - 0.5\big( \rev{(1-2\ea)} p_1 + \rev{(1-2\eb)}\pat (p_2-\pbt) \big) \\
            &\quad + \fa\big( \rev{(1-2\ea)}p_1p_2 + \rev{(1-2\eb)}\pat p_2 \big) \\
            &\quad + \fb\big( \rev{(1-2\ea)}p_1\pbt - \rev{(1-2\eb)}\pat\pbt \big).
        \end{align*}
        Thus \(\Jsqp\) is affine in \(\fa,\fb\). Hence it is optimal to set each of \(\fa,\fb\) equal
        to \(0\) or \(1\)\rev{, where positivity forces $|\fvj|\le\min\{\fj,1-\fj\}=0$, i.e.\ $F_j\in\{O,I\}$}.
        Hence \rev{a deterministic classical strategy of \B{} is a best response to any mixed classical strategy of \A{}}.

        \item For the dynamic information structure, we prove both directions.

        \smallskip
        \textit{Case 1: Suppose \A{} plays a classical strategy.}

        \rev{Then \(\ea,\eb\in[0,1]\) with $\eva =\evb=0$.  } Using \eqref{jdq}, we get that since $\evi=0$ the cost is independent of $\fvj$. Thus,
        \begin{align*}
            \Jdqp
            &= -p_s\ea\fa - p_s\ea\fb - \pst\eb\fa + \pst\eb\fb
            + p_s\ea + \pst\pct \\
            &\quad + (p_s\pct+\pst p_c)\fa + (p_sp_c-\pst\pct)\fb \\
            &= \fa (p_s\pct+\pst p_c-p_s \ea-\pst \eb)
            + \fb (p_sp_c-\pst\pct-p_s \ea+\pst \eb)
            + p_s\ea+\pst\pct.
        \end{align*}
        Thus \(\Jdqp\) is affine in \(\fa,\fb\). Hence it is optimal to set each of \(\fa,\fb\) equal
        to \(0\) or \(1\). Therefore \(F_0,F_1\in\{O,I\}\), i.e., it is optimal for \B{} to play a
        classical strategy.

        \smallskip
        \textit{Case 2: Suppose \B{} plays a classical strategy.}

        \rev{Then \(\fa,\fb\in[0,1]\) with $\fva=\fvb=0$ (mixed classical), while \A{} plays an
        arbitrary $\bie\in\mathsf{Q}$; recall from Section \ref{sec:paramquant} that
        \(\ea,\eb\in[0,1]\) and $|\evi|\le\min\{\ei,1-\ei\}$ hold for any such $\bie$, so no
        restriction beyond $\bie\in\mathsf{Q}$ is imposed. Since $\fva=\fvb=0$ the cost is
        independent of $\evi$.} Using \eqref{jdq}, we get
        \begin{align*}
            &\Jdqp\\
            &= -p_s\ea\fa - p_s\ea\fb - \pst\eb\fa + \pst\eb\fb
            + p_s\ea + \pst\pct  + (p_s\pct+\pst p_c)\fa + (p_sp_c-\pst\pct)\fb \\
            &= \ea\, p_s(1-\fa-\fb) + \eb\, \pst(\fb-\fa)  + \fa(p_s\pct+\pst p_c)
            + \fb(p_sp_c-\pst\pct)
            + \pst\pct.
        \end{align*}
        Thus \(\Jdqp\) is affine in \(\ea,\eb\). Hence it is optimal to set each of \(\ea,\eb\) equal
        to \(0\) or \(1\). Therefore \(E_0,E_1\in\{O,I\}\), i.e., it is optimal for \A{} to play a
        classical strategy.

        Thus, the property of best response over classical strategies holds for the dynamic
        information structure.
    \end{enumerate}
\end{proof}

\subsubsection{Theorem \ref{thm:best_proj}}
\label{sec:4.2}
\begin{proof}
    \begin{enumerate}
        \item For the static information structure, since the problem is symmetric, it is enough to
        show the following: if \A{} plays a projective quantum strategy, then it is optimal for \B{}
        to play a projective quantum strategy.

Suppose \(E_0,E_1\) are projective. We have
        $
        \ea=\eb=\half, |\eva|=|\evb|=\half.
        $
        Then, using \eqref{eq:trace_ef} and \eqref{aibj}, \eqref{Jsq} becomes
\begin{align*}
    \Jsqp
    &= \half - 2\Big[f_0^x(p_1p_2e_0^x+\pat p_2e_1^x)
    +f_0^y(-p_1p_2e_0^y-\pat p_2e_1^y)
    +f_0^z(p_1p_2e_0^z+\pat p_2e_1^z) \\
    &\hspace{1.3cm}
    +f_1^x(p_1\pbt e_0^x-\pat\pbt e_1^x)
    +f_1^y(-p_1\pbt e_0^y+\pat\pbt e_1^y)
    +f_1^z(p_1\pbt e_0^z-\pat\pbt e_1^z)\Big] \\
    &= \half - 2p_2\, \fva\odot (p_1\eva+\pat\evb)
    - 2\pbt\, \fvb\odot (p_1\eva-\pat\evb).
\end{align*}
        Hence, by Cauchy--Schwarz inequality,
        \begin{align*}
            \Jsqp
            &\ge \half
            -2p_2\,|\fva|\,|p_1\eva+\pat\evb|
            -2\pbt\,|\fvb|\,|p_1\eva-\pat\evb|.
        \end{align*}
        For fixed \(E_0,E_1\), the right-hand side is nonincreasing in \(|\fva|\) and \(|\fvb|\).
        Hence it is optimal to maximize both \(|\fva|\) and \(|\fvb|\). Since
        \[
        |\fva|\le \min\{\fa,1-\fa\}\le \half,\qquad
        |\fvb|\le \min\{\fb,1-\fb\}\le \half,
        \]
        the maximum possible value is \(\half\), which is attained only when
        $
        \fa=\fb=\half,|\fva|=|\fvb|=\half.
        $
        This is exactly the condition that \(F_0\) and \(F_1\) are projective. Hence it is optimal
        for \B{} to play a projective quantum strategy.

        \item For the dynamic information structure, we give a counterexample. Let \A{} play a
        projective strategy $E_0,E_1$. Then
        $
        \ea=\eb=\half, |\eva|=|\evb|=\half.
        $
        Using \eqref{jdq}, we get
        \begin{align}
            \Jdqp
            &= \half p_s+\pst\pct+\fa\rev{\Big(-\half+p_s\pct+\pst p_c\Big)}+\fb\rev{\Big(\frac{\pst-p_s}{2}+p_sp_c-\pst\pct\Big)} \notag\\
            &\quad -\fva\odot(p_s\eva+\pst\evb)-\fvb\odot(p_s\eva-\pst\evb). \rev{\label{eq:jdq_proj_general}}
        \end{align}
        \rev{Note that the coefficients of $\fa$ and $\fb$ simplify to $-\half(2p_s-1)(2p_c-1)$ and $p_c - \half$, respectively.}

        Now for the counterexample choose
        $
        p_s=\half, p_c>\half,
        $
        and let \A{} play the same projective strategy for both values of \(y_1\), namely
        \[
        \eva=\evb=\left(\frac{1}{2\sqrt{2}},0,\frac{1}{2\sqrt{2}}\right).
        \]
        Then
        $
        p_s\eva+\pst\evb=\eva $ and $ p_s\eva-\pst\evb=0,
        $
        and therefore
        \begin{align*}
            \Jdqp
            &= \frac14+\frac12\pct+\fa\Big(-\frac12+\frac12\pct+\frac12p_c\Big)
            +\fb\Big(\frac12p_c-\frac12\pct\Big)-\fva\odot\eva \\
            &= \frac14+\frac12\pct+\frac12(2p_c-1)\fb-\fva\odot\eva.
        \end{align*}
        Since \(|\eva|=\half\), we have
        $
        -\fva\odot\eva \ge -|\fva|\,|\eva| = -\half|\fva|,
        $
        with equality when \(\fva\) is parallel to \(M\eva\). Hence
        \[
        \Jdqp \ge \frac14+\frac12\pct+\frac12(2p_c-1)\fb-\half|\fva|.
        \]
        Also,
        $
        |\fva|\le \min\{\fa,1-\fa\}\le \half.
        $
        Thus the last term is minimized by taking
        $
        \fa=\half, |\fva|=\half.
        $
        With this choice,
        \[
        \Jdqp \ge \frac12\pct+\frac12(2p_c-1)\fb.
        \]
        Since \(p_c>\half\), the coefficient of \(\fb\) is positive, and hence the optimal $\fb$ must satisfy
        $
        \fb=0.
        $
        But then \(|\fvb|\le \min\{\fb,1-\fb\}=0\), so \(\fvb=0\), and hence \(F_1=O\). Therefore
        the only optimal response of \B{} is given by
        $
        \fa=\half, |\fva|=\half, \fb=0,\fvb=0,
        $
        so \(F_0\) is projective while \(F_1\) is classical. In particular, the optimal response of
        \B{} is not projective.

        Thus, the property of best response over projective strategies does not hold for the dynamic
        information structure.
    \end{enumerate}
\end{proof}

\subsubsection{Theorem \ref{thm:pbp_class}}
\label{sec:4.3}
\begin{proof}
    \begin{enumerate}
        \item For the static information structure, we first make a few observations that will be
        required in the proof.

        (a) When \A{} plays classical strategies: \rev{let $E_i=\ei I$, $\ei\in[0,1]$, $\evi=0$ (mixed class; deterministic is $\ei\in\{0,1\}$), $f_j^0\in[0,1]$. Using \eqref{eq:trace_ef} and \eqref{aibj},}
        \begin{align}
            \langle A_i B_j\rangle = \rev{(1-2e_i^0)(1-2f_j^0),}
            \label{eq:pbp_class_stat_a}
        \end{align}
        \rev{independent of $\fvj$; for deterministic \A{} this reads $(-1)^{e_i^0}(1-2f_j^0)$.}
        Using this in \eqref{Jsq} and collecting the \(f_j^0\)-terms, we get
        \begin{align}
            \Jsqp
            &= 0.5 - 0.5\big( \rev{(1-2\ea)}p_1 + \rev{(1-2\eb)}\pat(p_2-\pbt)\big) \non \\
            &\quad + \fa p_2\big( \rev{(1-2\ea)}p_1 + \rev{(1-2\eb)}\pat\big) \non \\
            &\quad + \fb \pbt\big( \rev{(1-2\ea)}p_1 - \rev{(1-2\eb)}\pat\big).
            \label{eq:pbp_class_stat_b}
        \end{align}

        (b) When \B{} plays classical strategies: then \rev{\(\fj\in[0,1]\)}, \(e_i^0\in[0,1]\), and similarly
        $
        \langle A_i B_j\rangle = \rev{(1-2f_j^0)}(1-2e_i^0).
        $
        Using this in \eqref{Jsq} and collecting the \(e_i^0\)-terms, we get
        \begin{align}
            \Jsqp
            &= 0.5 - 0.5\big( \rev{(1-2\fa)}p_2 + \rev{(1-2\fb)}\pbt(p_1-\pat)\big) \non \\
            &\quad + \ea p_1\big( \rev{(1-2\fa)}p_2 + \rev{(1-2\fb)}\pbt\big) \non \\
            &\quad + \eb \pat\big( \rev{(1-2\fa)}p_2 - \rev{(1-2\fb)}\pbt\big).
            \label{eq:pbp_class_stat_c}
        \end{align}

        \rev{The pairs exhibited below are deterministic, hence in $\mathsf C$; in each region the unilateral deviation is verified over the full ranges $\fj\in[0,1]$ (resp.\ $\ei\in[0,1]$) and, since the cost is independent of the deviating player's Bloch vectors when the opponent is classical, over all of $\mathsf Q$. This establishes person-by-person optimality over the mixed class.}

        We now verify person-by-person optimality in four exhaustive regions.

        \smallskip
        \textit{Case 1: \(p_1\ge \half,\; p_2\ge \half\).}
        We claim that
        $
        \{E_0,E_1\}=\{O,O\} \text{and} \{F_0,F_1\}=\{O,O\}
        $
        is person-by-person optimal.

        Suppose \A{} plays \(\{O,O\}\), i.e. \(\ea=\eb=0\). Then from \eqref{eq:pbp_class_stat_b},
        the coefficient of \(\fa\) is
        $
        p_2(p_1+\pat)=p_2>0,
        $
        and the coefficient of \(\fb\) is
        $
        \pbt(p_1-\pat)\ge 0.
        $
        Hence it is optimal for \B{} to play \(\{O,O\}\).

        Conversely, suppose \B{} plays \(\{O,O\}\), i.e. \(\fa=\fb=0\). Then from \eqref{eq:pbp_class_stat_c},
        the coefficient of \(\ea\) is
        $
        p_1(p_2+\pbt)=p_1>0,
        $
        and the coefficient of \(\eb\) is
        $
        \pat(p_2-\pbt)\ge 0.
        $
        Hence it is optimal for \A{} to play \(\{O,O\}\).

        \smallskip
        \textit{Case 2: \(p_1>\half,\; p_2<\half\).}
        We claim that
        $
        \{E_0,E_1\}=\{O,I\}, \{F_0,F_1\}=\{O,O\}
        $
        is person-by-person optimal.

        Suppose \A{} plays \(\{O,I\}\), i.e. \(\ea=0,\eb=1\). Then from \eqref{eq:pbp_class_stat_b},
        the coefficient of \(\fa\) is
        $
        p_2(p_1-\pat)>0,
        $
        and the coefficient of \(\fb\) is
        $
        \pbt(p_1+\pat)=\pbt>0.
        $
        Hence it is optimal for \B{} to play \(\{O,O\}\).

        Conversely, suppose \B{} plays \(\{O,O\}\), i.e. \(\fa=\fb=0\). Then from \eqref{eq:pbp_class_stat_c},
        the coefficient of \(\ea\) is
        $
        p_1(p_2+\pbt)=p_1>0,
        $
        while the coefficient of \(\eb\) is
        $
        \pat(p_2-\pbt)<0.
        $
        Hence it is optimal for \A{} to play \(\{O,I\}\).

        \smallskip
        \textit{Case 3: \(p_1<\half,\; p_2>\half\).}
        We claim that
        $
        \{E_0,E_1\}=\{O,O\}, \{F_0,F_1\}=\{O,I\}
        $
        is person-by-person optimal.

        Suppose \A{} plays \(\{O,O\}\), i.e. \(\ea=\eb=0\). Then from \eqref{eq:pbp_class_stat_b},
        the coefficient of \(\fa\) is
        $
        p_2(p_1+\pat)=p_2>0,
        $
        and the coefficient of \(\fb\) is
        $
        \pbt(p_1-\pat)<0.
        $
        Hence it is optimal for \B{} to play \(\{O,I\}\).

        Conversely, suppose \B{} plays \(\{O,I\}\), i.e. \(\fa=0,\fb=1\). Then from \eqref{eq:pbp_class_stat_c},
        the coefficient of \(\ea\) is
        $
        p_1(p_2-\pbt)>0,
        $
        and the coefficient of \(\eb\) is
        $
        \pat(p_2+\pbt)=\pat>0.
        $
        Hence it is optimal for \A{} to play \(\{O,O\}\).

        \smallskip
        \textit{Case 4: \(p_1\le \half,\; p_2\le \half\).}
        We claim that
        $
        \{E_0,E_1\}=\{O,I\},\{F_0,F_1\}=\{I,O\}
        $
        is person-by-person optimal.

        Suppose \A{} plays \(\{O,I\}\), i.e. \(\ea=0,\eb=1\). Then from \eqref{eq:pbp_class_stat_b},
        the coefficient of \(\fa\) is
        $
        p_2(p_1-\pat)\le 0,
        $
        and the coefficient of \(\fb\) is
        $
        \pbt(p_1+\pat)=\pbt>0.
        $
        Hence it is optimal for \B{} to play \(\{I,O\}\).

        Conversely, suppose \B{} plays \(\{I,O\}\), i.e. \(\fa=1,\fb=0\). Then from \eqref{eq:pbp_class_stat_c},
        the coefficient of \(\ea\) is
        $
        p_1(-p_2+\pbt)=p_1(\pbt-p_2)\ge 0,
        $
        and the coefficient of \(\eb\) is
        $
        \pat(-p_2-\pbt)=-\pat<0.
        $
        Hence it is optimal for \A{} to play \(\{O,I\}\).

        Thus, for the static information structure, the property of person-by-person optimality over
        classical strategies holds.

        \item For the dynamic information structure, we again make two observations.

        (a) When \A{} plays classical strategies: note that \(E_0,E_1\in\{O,I\}\) represent all 4 classical
        strategies. Then \(\ea,\eb\in\{0,1\}\), \(\fa,\fb\in[0,1]\), and using \eqref{jdq}, we get
        \begin{align}
            \Jdq
            &= \fa(p_s\pct+\pst p_c-p_s\ea-\pst\eb)
            + \fb(p_sp_c-\pst\pct-p_s\ea+\pst\eb)
            + p_s\ea+\pst\pct .
            \label{eq:pbp_class_dyn_a}
        \end{align}

        (b) When \B{} plays classical strategies: then \(\fa,\fb\in\{0,1\}\), \(\ea,\eb\in[0,1]\). Starting
        from \eqref{jdq} and collecting the \(\ea,\eb\)-terms, we get
        \begin{align}
            \Jdq
            &= \ea p_s(1-\fa-\fb)+\eb \pst(\fb-\fa)
            + \fa(p_s\pct+\pst p_c)+\fb(p_sp_c-\pst\pct)+\pst\pct .
            \label{eq:pbp_class_dyn_b}
        \end{align}

        We now verify person-by-person optimality in two exhaustive regions.

        \smallskip
        \textit{Case 1: \(p_s\ge \half\).}
        We claim that
        $
        \{E_0,E_1\}=\{O,I\},\{F_0,F_1\}=\{O,O\}
        $
        is person-by-person optimal.

        Suppose \A{} plays \(\{O,I\}\), i.e. \(\ea=0,\eb=1\). Then from \eqref{eq:pbp_class_dyn_a},
        the coefficient of \(\fa\) is
        \[
        p_s\pct+\pst p_c-\pst
        = p_s\pct-\pst\pct
        =\pct(p_s-\pst)\ge 0,
        \]
        and the coefficient of \(\fb\) is
        \[
        p_sp_c-\pst\pct+\pst
        = p_sp_c+\pst p_c
        = p_c\geq 0.
        \]
        Hence it is optimal for \B{} to play \(\{O,O\}\).

        Conversely, suppose \B{} plays \(\{O,O\}\), i.e. \(\fa=\fb=0\). Then from \eqref{eq:pbp_class_dyn_b},
        the coefficient of \(\ea\) is
        $
        p_s>0,
        $
        and the coefficient of \(\eb\) is
        $
        0.
        $
        Hence \(E_1\) can be chosen arbitrarily, and it is optimal to set \(E_0=O\). In particular,
        \(\{E_0,E_1\}=\{O,I\}\) is optimal for \A{}.

        \smallskip
        \textit{Case 2: \(p_s<\half\).}
        We claim that
        $
        \{E_0,E_1\}=\{O,I\}, \{F_0,F_1\}=\{I,O\}
        $
        is person-by-person optimal.

        Suppose \A{} plays \(\{O,I\}\), i.e. \(\ea=0,\eb=1\). Then from \eqref{eq:pbp_class_dyn_a},
        the coefficient of \(\fa\) is
        \[
        p_s\pct+\pst p_c-\pst
        = \pct(p_s-\pst)<0,
        \]
        and the coefficient of \(\fb\) is
        $
        p_sp_c-\pst\pct+\pst
        = p_c>0.
        $
        Hence it is optimal for \B{} to play \(\{I,O\}\).

        Conversely, suppose \B{} plays \(\{I,O\}\), i.e. \(\fa=1,\fb=0\). Then from \eqref{eq:pbp_class_dyn_b},
        the coefficient of \(\ea\) is
        $
        p_s(1-1-0)=0,
        $
        and the coefficient of \(\eb\) is
        $
        \pst(0-1)=-\pst<0.
        $
        Hence \(E_0\) can be chosen arbitrarily, and it is optimal to set \(E_1=I\). In particular,
        \(\{E_0,E_1\}=\{O,I\}\) is optimal for \A{}.

        Thus, for the dynamic information structure, the property of person-by-person optimality over
        classical strategies holds.
    \end{enumerate}
\end{proof}

\subsubsection{Theorem \ref{thm:pbp_proj}}
\label{sec:4.4}
\begin{proof}
    \begin{enumerate}\rev{
        \item For the static information structure, from Theorem \rev{\ref{thm:optstatic}}, we know that for every \(p_1,p_2\), \(\Jsqop\) is
        attained by projective strategies. Since both players have the same
        objective function, any team-optimal pair of strategies is necessarily person-by-person optimal.
        Hence, for every \(p_1,p_2\), there exists a pair of projective strategies that is person-by-person
        optimal.

        \item For the dynamic information structure we exhibit, for every $(p_s,p_c)$, a projective pair that is person-by-person optimal. Suppose $p_s \geq \half$; the case $p_s \leq \half$ is symmetric, as noted at the end. Let $P = \ketbra{0}{0}$, whose Bloch vector is $(0,0,\half)$, and let
        \[
        \bie^* = (E_0, E_1) = (P,\, I-P),
        \qquad
        \bif^* = (F_0, F_1) = (P,\, P),
        \]
        so that $\eva = -\evb = (0,0,\half)$ and $\ea = \eb = \half$.

        First, suppose \A{} plays $\bie^*$. Then $p_s\eva + \pst\evb = \left(0,0,p_s-\frac12\right)$ and $p_s\eva - \pst\evb = (0,0,\half)$, and \eqref{eq:jdq_proj_general} becomes
        \[
            \Jdqp
            = \half p_s+\pst\pct  + \Big(p_c - \half\Big)\fb - \Big(p_s - \half\Big)\Big((2p_c-1)\fa + f_0^z\Big) - \half f_1^z.
        \]
        \B{} minimizes the $F_0$- and $F_1$-terms separately. Consider the $F_1$-term $(p_c-\half)\fb - \half f_1^z$. Since $f_1^z \leq |\fvb| \leq \min\{\fb, 1-\fb\}$, it is at least $(p_c-1)\fb$ when $\fb \leq \half$, and at least $p_c\fb - \half$ when $\fb \geq \half$. The first bound is nonincreasing in $\fb$ and the second nondecreasing, and both equal $\tfrac{p_c-1}{2}$ at $\fb = \half$. Hence the $F_1$-term is minimized at $\fb = \half$, $\fvb = (0,0,\half)$, i.e., at $F_1 = P$. Similarly, since $2p_s-1 \geq 0$, the $F_0$-term is minimized by maximizing $(2p_c-1)\fa + f_0^z$, which is at most $2p_c\fa$ when $\fa \leq \half$ and at most $(2p_c-2)\fa + 1$ when $\fa \geq \half$; both bounds are maximized at $\fa = \half$, where they equal $p_c$. Hence the $F_0$-term is minimized at $\fa = \half$, $\fva = (0,0,\half)$, i.e., at $F_0 = P$. (When $p_s = \half$ the $F_0$-term vanishes and every $F_0$ is optimal.) Thus $\bif^*$ is an optimal response of \B{} to $\bie^*$ over all POVMs.

        Conversely, suppose \B{} plays $\bif^*$. For a projective $\bif^*$ (so $\fa = \fb = \half$), collecting the $\ea, \eb$-terms in \eqref{jdq} using \eqref{eq:trace_e} and \eqref{eq:trace_ef}, the coefficient of $\ea$ is $-p_s(\fa+\fb)+p_s = 0$ and that of $\eb$ is $-\pst(\fa-\fb) = 0$, so
        \[
        \Jdq(\bie,\bif^*;p_s,p_c) = \half - p_s\,\eva\odot(\fva+\fvb) - \pst\,\evb\odot(\fva-\fvb) = \half - p_s\, e_0^z,
        \]
        since $\fva = \fvb = (0,0,\half)$. As $e_0^z \leq |\eva| \leq \min\{\ea, 1-\ea\} \leq \half$, the minimum is attained at $\eva = (0,0,\half)$, which forces $\ea = \half$, i.e., $E_0 = P$, while $E_1$ is arbitrary. Thus $\bie^*$ is an optimal response of \A{} to $\bif^*$ over all POVMs.

        Hence $(\bie^*, \bif^*)$ is person-by-person optimal. By Corollary \ref{cor:JspJdp}, $\langle A_0 B_0\rangle = \langle A_0 B_1 \rangle = 1$ and $\langle A_1 B_0\rangle = \langle A_1 B_1 \rangle = -1$, so its cost is $\half - \half p_s = \half\pst$, independent of $p_c$. For $p_s \leq \half$, the same computations show that $\bie^* = (P, I-P)$, $\bif^* = (I-P, P)$ is person-by-person optimal, with cost $\half p_s$. Hence the property of person-by-person optimality over projective strategies holds for the dynamic information structure, with cost $\half\min\{p_s,\pst\}$.}
    \end{enumerate}
\end{proof}

\rev{\subsubsection{Theorem \ref{thm:pbp_unstable}}
\label{sec:4.5}
\begin{proof}
\begin{enumerate}
    \item For the static information structure, from Theorem \ref{thm:optstatic}, we know that \(\Jsqop\) is always
        attainable via projective strategies played by both players. Since both players have the same
        objective function, any team-optimal pair of strategies is necessarily person-by-person optimal.
        Hence, for every \(p_1,p_2\), the optimal strategies over projective strategies are also person-by-person
        optimal.

\item We first identify which strategies are optimal over projective strategies, i.e., that which obtain the minimum in \eqref{jipo}. We then test such a pair for person-by-person optimality by letting each
player unilaterally deviate over all quantum strategies.

We assume $p_s \in (0,1)$; for $p_s \in \{0,1\}$, \eqref{eq:pbp_threshold} always holds and the claim is verified directly by the same computations.
Let $(\bie,\bif )$ denote any projective strategies, then as in Case 2 for the proof of Theorem \ref{thm:optdyn}, we get
\[
\Jdq(\bie,\bif ; p_s, p_c) = \half - p_s\,\eva\odot(\fva+\fvb) - \pst\,\evb\odot(\fva-\fvb),
\]
and the minimum over projective pairs is $\Jdpop = \half - \half\sqrt{p_s^2+\pst^2}$. Since $|\eva| = |\evb| = \half$ and $|\fva+\fvb|^2 + |\fva-\fvb|^2 = 1$, aligning $\eva$, $\evb$ with $M(\fva+\fvb)$, $M(\fva-\fvb)$, yields the optimal cost. Therefore the strategy is optimal if and only if
\[
\eva = \frac{M(\fva+\fvb)}{2|\fva+\fvb|},
\qquad
\evb = \frac{M(\fva-\fvb)}{2|\fva-\fvb|},
\qquad
\big(|\fva+\fvb|,\,|\fva-\fvb|\big) = \frac{(p_s,\,\pst)}{\sqrt{p_s^2+\pst^2}}.
\]
Hence $(\bie,\bif)=(\bie^{D*}_{\perp},\bif^{D*}_{\perp})$ is optimal over projective strategies if and only if the parameters $\eva^\perp,\evb^\perp$ and $\fva^\perp,\fvb^\perp$ ($D*$ dropped to lighten the notation) of $\bie^{D*}_{\perp},\bif^{D*}_{\perp}$  satisfy
\begin{equation}
p_s\eva^{\perp} + \pst\evb^{\perp} = \sqrt{p_s^2+\pst^2}\;M\fva^{\perp},
\qquad
p_s\eva^{\perp} - \pst\evb^{\perp} = \sqrt{p_s^2+\pst^2}\;M\fvb^{\perp},
\label{eq:opt_alignment}
\end{equation}

We now check this for person-by-person optimality. Let $$r :=  |p_s\eva^{\perp} + \pst\evb^{\perp}| = |p_s\eva^{\perp} - \pst\evb^{\perp}| =  \half\sqrt{p_s^2+\pst^2}.$$

Now let us fix \A{} to play $\bie^{D*}_{\perp}$ and let \B{} play a general strategy $\bif$. Then by \eqref{jdq} and using the fact that \A{} is playing a projective strategy we get

        \begin{align}
            \Jdq(\bie^{D*}_{\perp},\bif ; p_s, p_c)
            &= \half p_s+\pst\pct+\fa\Big(-\half+p_s\pct+\pst p_c\Big)+\fb\Big(\frac{\pst-p_s}{2}+p_sp_c-\pst\pct\Big) \notag\\
            &\quad -\fva\odot(p_s\eva^\perp+\pst\evb^\perp)-\fvb\odot(p_s\eva^\perp-\pst\evb^\perp). \notag
        \end{align}
\B{} would then minimize the $F_0$- and $F_1$-terms
\[
\alpha_0\,\fa - (p_s\eva^\perp+\pst\evb^\perp)\odot\fva
\qquad\text{and}\qquad
\alpha_1\,\fb - (p_s\eva^\perp-\pst\evb^\perp)\odot\fvb
\]
separately, where $\alpha_0 = -\half(2p_s-1)(2p_c-1)$ and $\alpha_1 = p_c-\half$. Consider the $F_1$-term. Since $(p_s\eva-\pst\evb)\odot\fvb \leq r\min\{\fb,1-\fb\}$, with equality when $\fvb$ is parallel to $M(p_s\eva-\pst\evb)$ of maximal norm, the $F_1$-term is at least $(\alpha_1-r)\fb$ for $\fb \leq \half$ and at least $(\alpha_1+r)\fb - r$ for $\fb \geq \half$, and both bounds are attained. The minimum value of the $F_1$-term over all POVM elements is therefore
\[
\min\Big\{0,\ \frac{\alpha_1-r}{2},\ \alpha_1\Big\}.
\]
The first value in the above expression (\ie, $0$) can be achieved by playing
$F_1 = O$.
The second value is achieved by playing the projective strategy $ \fvb^\perp = \frac{1}{2r}M(p_s\eva^\perp-\pst\evb^\perp)$, which by \eqref{eq:opt_alignment} is exactly the $F^{D*}_{1\perp}$.
The third value ($\alpha_1$) can be achieved by playing $F_1 = I$.

Hence $F^{D*}_{1\perp}$ is an optimal response if and only if $\frac{\alpha_1-r}{2} \leq \min\{0,\alpha_1\}$, i.e., if and only if $|\alpha_1| \leq r$; and when $|\alpha_1| > r$, \B{} strictly improves by the classical $F_1 = O$ (if $p_c > \half$) or $F_1 = I$ (if $p_c < \half$).\\

The $F_0$-term is handled identically, with $\alpha_1$ replaced by $\alpha_0$ and $F_1$ by $F_0$. Since $(p_s\eva+\pst\evb)\odot\fva \leq r\min\{\fa,1-\fa\}$, with equality when $\fva$ is parallel to $M(p_s\eva+\pst\evb)$ of maximal norm, the $F_0$-term is at least $(\alpha_0-r)\fa$ for $\fa \leq \half$ and at least $(\alpha_0+r)\fa - r$ for $\fa \geq \half$, and both bounds are attained. The minimum value of the $F_0$-term over all POVM elements is therefore
\[
\min\Big\{0,\ \frac{\alpha_0-r}{2},\ \alpha_0\Big\}.
\]
By proceeding exactly like that for $F_1$, we get that $F^{D*}_{0\perp}$ is an optimal response if and only if $|\alpha_0| \leq r$. Since $|\alpha_0| = |2p_s-1|\,|\alpha_1| \leq |\alpha_1|$, the condition $|\alpha_j| \leq r, j=0,1$ reduces to $|p_c-\half| \leq r$, which upon squaring is precisely \eqref{eq:pbp_threshold}. Thus $\bif^{D*}_{\perp}$ is a best response to $\bie^{D*}_{\perp}$ if and only if \eqref{eq:pbp_threshold} holds. \revii{Positivity of \B{}'s effects is already fully encoded in this computation: the bound $(p_s\eva^\perp-\pst\evb^\perp)\odot\fvb \le r\min\{\fb,1-\fb\}$ is exactly the constraint $|\fvb|\le\min\{\fb,1-\fb\}$ of Section~\ref{sec:paramquant}, and the minimisation ranges over all $\fb\in[0,1]$. Condition \eqref{eq:pbp_threshold} is therefore both necessary and sufficient, with no further constraint left implicit.}\\

Conversely, fix \B{} to play $\bif^{D*}_{\perp}$. Notice that for any $\bie,\bif$ \eqref{jdq} gives
\begin{align}
\Jdq(\bie,\bif;p_s,p_c)
&= \pst\pct
 + p_s\,\ea\,(1-\fa-\fb)
 + \pst\,\eb\,(\fb-\fa) \notag\\
&\quad + (p_s\pct+\pst p_c)\,\fa
       + (p_s p_c-\pst\pct)\,\fb \notag\\
&\quad - p_s\,\eva\odot(\fva+\fvb)
       - \pst\,\evb\odot(\fva-\fvb).
\label{eq:jdq_bloch}
\end{align}
Now impose $\bif=\bif^{D*}_{\perp}$, i.e., $\fa=\fb=\half$. The coefficients of $\ea$
and $\eb$ in \eqref{eq:jdq_bloch}, namely $p_s(1-\fa-\fb)$ and
$\pst(\fb-\fa)$, both vanish. The
remaining scalars combine with the constant as
\[
\pst\pct + \tfrac12(p_s\pct+\pst p_c) + \tfrac12(p_s p_c-\pst\pct)
= \tfrac12(p_s+\pst)(p_c+\pct) = \half,
\]
using $p_s+\pst=1$ and $p_c+\pct=1$, while the vector terms are unaffected.
Hence for \emph{every} POVM $\bie$,
\[
\Jdq(\bie,\bif^{D*}_{\perp};p_s,p_c)
= \half - p_s\,\eva\odot(\fva^\perp+\fvb^\perp) - \pst\,\evb\odot(\fva^\perp-\fvb^\perp).
\]
The two terms involve the disjoint pairs $\eva,\evb$, each constrained only by
$|\eva|,|\evb|\le\half$, so each is maximized at $|\eva|=|\evb|=\half$ with
$\eva,\evb$ aligned to $M(\fva^\perp+\fvb^\perp)$ and $M(\fva^\perp-\fvb^\perp)$; this is exactly the
alignment \eqref{eq:opt_alignment}, attained by the rank-one projections
$\bie^{D*}_{\perp}$. Since the maximizer is projective and feasible for all
$(p_s,p_c)$, $\bie^{D*}_{\perp}$ is a best response to $\bif^{D*}_{\perp}$ for every
$(p_s,p_c)$, with no condition on the parameters. Combining with the forward
direction, the team optimal projective pairs are person-by-person optimal if
and only if \eqref{eq:pbp_threshold} holds.
\end{enumerate}
\end{proof}}

\subsection{Invariance of Bell Costs}
\label{sec:bellinv}
As claimed in Section \ref{sec:bell1}, we now show that the team optimal costs do not change across the different Bell costs. So our analysis done for $\Be$ does not need to be repeated for the other costs, as the same results hold.

Now, for a general $\Bek$, the expressions for all Bell costs will be the same as \eqref{eq:jsqp} and \eqref{eq:jdqp} but for the corresponding cost.
\begin{definition}
	The expected costs under quantum strategies for the static and dynamic information structure, for a general Bell cost $\Bek$, are:
\begin{subequations}
	\begin{align}
		\Jsqp(\Bek) &= \mathop{\Ebb}_{u_1,u_2,y_1,y_2}[\Bek(u_1,u_2; y_1, y_2)] \label{eq:genjsqp}\\
		\Jdqp &= \mathop{\Ebb}_{u_1,u_2,y_1,y_2}[\Bek(u_1,u_2; y_1, y_2)], \label{eq:genjdqp}
	\end{align}
\end{subequations}
where the joint probability distributions are given earlier in \ref{lemma:joint} (Equations \ref{jointS} and \ref{jointD}).
\end{definition}

We show that the optimal quantum cost is invariant across the eight Bell costs
\(\mathcal{B}_k\), \(k=1,\dots,8\).

\paragraph{Static information structure.}
The Bell costs \(\mathcal{B}_1,\mathcal{B}_2,\mathcal{B}_3,\mathcal{B}_4\) are obtained from each
other by relabeling the inputs \(y_1,y_2\). For example, $\mathcal{B}_2 (u_1,u_2,y_1,y_2) = \mathcal{B}_1 (u_1,u_2,y_1,\overline{y_2})$. These relabelings correspond to exchanging
\(E_0\) with \(E_1\), or \(F_0\) with \(F_1\), and replacing \(p_1\) by $\pat$ or \(p_2\) by $\pbt$.
Since the optimization in \eqref{Jsq} is over all admissible strategies, and the expression $\Jsqo$ is invariant with replacing \(p_1\) by $\pat$ or \(p_2\) by $\pbt$,
the optimal value
remains unchanged. Hence
\[
\Jsqop[\mathcal{B}_1]=\Jsqop[\mathcal{B}_2]=\Jsqop[\mathcal{B}_3]=\Jsqop[\mathcal{B}_4].
\]

The remaining Bell costs are complements of these four. For example $\Bscr_5=\overline{\Bscr}_1$ and therefore $\Jsqp[\Bscr_5]=1-\Jsqp[\Bscr_1]$. Replacing \(F_j\) by \(I-F_j\) in
\eqref{Jsq} maps a Bell cost to its complement and preserves the feasible set of strategies. Specifically, for any $\bie,\bif$
\begin{align*}
    1-\Jsqp[\mathcal{B}_1](E_0,E_1,F_0,F_1,p_1, p_2)) =\Jsqp[\mathcal{B}_1](E_0,E_1,I-F_0,I-F_1,p_1, p_2)
\end{align*}
Hence
\[
\Jsqop[\mathcal{B}_1]=\Jsqop[\mathcal{B}_2]=\cdots=\Jsqop[\mathcal{B}_8].
\]

\paragraph{Dynamic information structure.}
A similar argument applies to the dynamic information structure and \eqref{jdq}. Relabeling \(y_1\) and \(y_2\) corresponds to exchanging
\(E_0\) with \(E_1\), or \(F_0\) with \(F_1\), and replacing \(p_s\) by $\pst$ or \(p_c\) by $\pct$.
This does not affect the optimal value. Further, replacing \(F_j\) by \(I-F_j\) maps a Bell cost
to its complement while preserving the feasible set. Hence
\rev{for $(\bie,\bif)\in\mathsf Q$ write $\bie^{\leftrightarrow}:=(E_1,E_0)$, $\bif^{\leftrightarrow}:=(F_1,F_0)$ (input relabeling) and $\bif^{\mathrm c}:=(I-F_0,I-F_1)$ (output complementation). Each is a bijection of $\mathsf Q$ onto itself, preserving $\perp$ (as $I-P$ is again a rank-one projection) and $\mathsf C$ (in the embedding of Proposition~\ref{prop:embedding} it permutes or affinely reflects the scalars $\ei,\fj\in[0,1]$, mapping deterministic strategies to deterministic strategies).

\begin{lemma}[Static invariance]\label{lem:inv_static}
Fix $(p_1,p_2)$ and $(\bie,\bif)\in\mathsf Q$. Then
\begin{align*}
\Jsqp[\mathcal B_2]&=J^{\mathsf S}_{\mathsf Q}(\bie,\bif^{\leftrightarrow};p_1,\pbt)[\mathcal B_1],\\
\Jsqp[\mathcal B_3]&=J^{\mathsf S}_{\mathsf Q}(\bie^{\leftrightarrow},\bif;\pat,p_2)[\mathcal B_1],\\
\Jsqp[\mathcal B_4]&=J^{\mathsf S}_{\mathsf Q}(\bie^{\leftrightarrow},\bif^{\leftrightarrow};\pat,\pbt)[\mathcal B_1],\\
\Jsqp[\mathcal B_{k+4}]&=J^{\mathsf S}_{\mathsf Q}(\bie,\bif^{\mathrm c};p_1,p_2)[\mathcal B_k],\quad k\le 4.
\end{align*}
\end{lemma}
\begin{proof}
Since $\mathcal B_2(u_1,u_2;y_1,y_2)=\mathcal B_1(u_1,u_2;y_1,\bar y_2)$, put $w:=\bar y_2\sim\Ber(1-\pbt)$; \B{}'s measurement $\{F_{y_2}\}$ equals $\{F^{\leftrightarrow}_{w}\}$, so the joint law of $(u_1,u_2,y_1,w)$ under $(\bie,\bif)$ at $(p_1,p_2)$ equals that of $(u_1,u_2,y_1,y_2)$ under $(\bie,\bif^{\leftrightarrow})$ at $(p_1,\pbt)$, with pathwise-equal costs. The second identity is symmetric; the third is their composition. For the last, complementing $u_2$ maps $\mathcal B_k$ to $1\oplus\mathcal B_k=\mathcal B_{k+4}$, implemented by $\bif\mapsto\bif^{\mathrm c}$.
\end{proof}

For the dynamic structure, input relabelings interact with the causal channel $y_2=y_1\oplus u_1\oplus v$; the next lemma makes this explicit.

\begin{lemma}[Dynamic invariance]\label{lem:inv_dyn}
Fix $(p_s,p_c)$ and $(\bie,\bif)\in\mathsf Q$. Then
\begin{align*}
\Jdqp[\mathcal B_2]&=J^{\mathsf D}_{\mathsf Q}(\bie,\bif^{\leftrightarrow};p_s,\pct)[\mathcal B_1],\\
\Jdqp[\mathcal B_3]&=J^{\mathsf D}_{\mathsf Q}(\bie^{\leftrightarrow},\bif;\pst,\pct)[\mathcal B_1],\\
\Jdqp[\mathcal B_4]&=J^{\mathsf D}_{\mathsf Q}(\bie^{\leftrightarrow},\bif^{\leftrightarrow};\pst,p_c)[\mathcal B_1],\\
\Jdqp[\mathcal B_{k+4}]&=J^{\mathsf D}_{\mathsf Q}(\bie,\bif^{\mathrm c};p_s,p_c)[\mathcal B_k],\quad k\le 4.
\end{align*}
Only local, per-player relabelings are used, so the causal order and the classes $\mathsf C,\perp,\mathsf Q$ are preserved.
\end{lemma}
\begin{proof}
$\mathcal B_2$: with $w:=\bar y_2=y_1\oplus u_1\oplus\bar v$ and $\bar v\sim\Ber(1-\pct)$, $w$ is the channel output of the $\mathcal B_1$-dynamics at noise $\pct$ (source unchanged), and $F_{y_2}=F^{\leftrightarrow}_w$; hence the costs match at $(p_s,\pct)$. $\mathcal B_3$: with $z:=\bar y_1\sim\Ber(1-\pst)$, $E_{y_1}=E^{\leftrightarrow}_z$ and $y_2=z\oplus u_1\oplus\bar v$, i.e.\ the $\mathcal B_1$-dynamics with source $\pst$ and noise $\pct$; unlike the static case, relabeling the source necessarily flips the effective channel noise, which is the interaction of relabeling with the causal channel. $\mathcal B_4$ is their composition (the two noise flips cancel); the complement identity is as in Lemma~\ref{lem:inv_static}.
\end{proof}

\begin{corollary}\label{cor:inv}
For every $k\in\{1,\dots,8\}$ and every class $\mathsf A\in\{\mathsf C,\perp,\mathsf Q\}$, at the same fixed parameters,
\[
J^\mathsf{S*}_{\mathsf{A}}(p_1,p_2)[\mathcal B_k]=J^\mathsf{S*}_{\mathsf{A}}(p_1,p_2)[\mathcal B_1],\qquad
J^\mathsf{D*}_{\mathsf{A}}(p_s,p_c)[\mathcal B_k]=J^\mathsf{D*}_{\mathsf{A}}(p_s,p_c)[\mathcal B_1].
\]
\end{corollary}
\begin{proof}
By Lemmas~\ref{lem:inv_static}--\ref{lem:inv_dyn}, minimizing $J[\mathcal B_k]$ over $\mathsf A$ at fixed parameters equals minimizing $J[\mathcal B_1]$ over the same class at parameters transformed by an involution ($p_1\!\leftrightarrow\!\pat$, $p_2\!\leftrightarrow\!\pbt$; resp.\ $p_s\!\leftrightarrow\!\pst$, $p_c\!\leftrightarrow\!\pct$), since the strategy maps are bijections of $\mathsf A$. The closed forms for $\mathcal B_1$ (Theorems~\ref{thm:jsco},~\ref{thm:jdco},~\ref{thm:optstatic},~\ref{thm:optdyn}, with $\Jdpo=\Jspo(p_1{=}p_s,p_2{=}\tfrac12)$) are invariant under these involutions.
\end{proof}}


\subsection{POVMs are sufficient to capture all strategy functions}
\label{app:povm}
In all the above considerations we assumed that the players choose their actions directly based on the measurement outcome $m_i$ of the measurements, that is $u_i = m_i$. It is important to check if searching over POVMs automatically captures choosing action $u_i = f(m_i, y_i)$ for all $f: \{0,1\}^2 \to \{0,1\}$ (as mentioned in \ref{sec:genq}). Suppose \A{} plays the strategy $E_0,E_1$ for $u_i = m_i$, then for all such functions, the equivalent POVMs would be as given by the table below

\begin{center}
$\begin{array}{ |c|c|c| }
 \hline
 f(m_1,y_1) & \text{ equivalent } E_0 & \text{ equivalent } E_1 \\
 \hline
 0 & I & I \\
 1 & 0 & 0 \\
 y_1 & I & 0 \\
 \overline{y_1} & 0 & I \\
 m_1 & E_0 & E_1 \\
 \overline{m_1} & I-E_0 & I-E_1 \\
 m_1 \oplus y_1 & E_0 & I-E_1 \\
 \overline{m_1 \oplus y_1} & I -E_0 & E_1 \\
 m_1 + y_1 & E_0 & 0 \\
 m_1 + \overline{y_1} & 0 & E_1 \\
 \overline{m_1} + y_1 & I-E_0 & 0 \\
 \overline{m_1} + \overline{y_1} & 0 & I - E_1 \\
 m_1 y_1 & I & E_1 \\
 m_1 \overline{y_1} & E_0 & I \\
 \overline{m_1} y_1 & I & I - E_1 \\
 \overline{m_1} \overline{y_1} & I-E_0 & I\\

 \hline
\end{array}$
\end{center}
\vspace{0.5cm}

Thus all equivalent $E_0,E_1$ are valid POVMs. Hence our POVM optimization also optimizes over all such functions of $m_1$ and $y_1$.

\end{document}